\documentclass[11pt]{article}   	
\usepackage[letterpaper, margin=1in]{geometry}                	
\usepackage{amssymb, bm, enumitem, bbm, mathtools}
\usepackage[obeyspaces,hyphens,spaces]{url}
\usepackage{amsthm}
\usepackage{subcaption}
\usepackage{graphicx,overpic}	
\usepackage[usenames,dvipsnames]{xcolor}						

\allowdisplaybreaks
\numberwithin{equation}{section}

\usepackage{tikz}
\usetikzlibrary{arrows, positioning, calc, intersections}
\usetikzlibrary{decorations.pathreplacing, decorations.markings}
\usepackage{relsize}					
\usepackage{exscale}	
\usepackage{pgfplots}
\newtheorem{theorem}{Theorem}[section]
\newtheorem{lemma}[theorem]{Lemma}
\newtheorem{corollary}[theorem]{Corollary}
\newtheorem{remark}[theorem]{Remark}
\newtheorem{prop}[theorem]{Proposition}
\newtheorem{assume}[theorem]{Assumptions}
\newtheorem{RHP}[theorem]{Riemann-Hilbert Problem}
\newcommand{\R}{{\mathbb R}}
\newcommand{\C}{\mathbb{C}}

\renewcommand{\Re}{\mathop{ \mathrm{Re}}\nolimits}
\renewcommand{\Im}{\mathop{ \mathrm{Im} }\nolimits}
\newcommand{\ii}{\mathrm{i}}
\newcommand{\e}{\mathrm{e}}
\newcommand{\M}{\mathbf{M}}
\newcommand{\G}{\mathbf{G}}
\newcommand{\F}{\mathbf{F}}
\newcommand{\E}{\mathbf{E}}
\newcommand{\bY}{\mathbf{Y}}

\newcommand{\B}{\mathbf{B}}
\newcommand{\A}{\mathbf{A}}
\newcommand{\V}{\mathbf{V}}

\newcommand{\Del}{\mathbf{\Delta}}
\newcommand{\bXi}{\mathbf{\Xi}}
\newcommand{\bPi}{\mathbf{\Pi}}

\newcommand{\bI}{\mathbf{I}}
\newcommand{\tdel}{\tilde{\delta}}
\newcommand{\cP}{\mathcal{P}}
\newcommand{\bP}{\mathbf{P}}
\newcommand{\bT}{\mathbf{T}}
\newcommand{\rZ}{\mathcal{Z}}
\newcommand{\cD}{\mathcal{D}}
\newcommand{\w}{\mathbf{w}}
\newcommand{\h}{\mathbf{h}}
\newcommand{\Ca}{\mathcal{C}}
\newcommand{\bu}{\mathbf{u}}
\newcommand{\bL}{\mathbf{L}}
\newcommand{\bZ}{\mathbf{Z}}
\newcommand{\m}{\mathbf{m}}
\newcommand{\bC}{\mathbf{C}}
\newcommand{\q}{\mathbf{q}}
\newcommand{\bPhi}{\mathbf{\Phi}}
\newcommand{\bX}{\mathbf{X}}
\newcommand{\bJ}{\mathbf{J}}
\newcommand{\bQ}{\mathbf{Q}}
\newcommand{\bv}{\mathbf{v}}
\newcommand{\bB}{\mathbf{B}}
\newcommand{\cR}{\mathcal{R}}
\newcommand{\blambda}{\mathbf{\Lambda}}
\newcommand{\cE}{\boldsymbol{\mathcal{E}}}
\newcommand{\xcE}{\mathcal{E}}

\newcommand{\eps}{\epsilon}

\def\be{\begin{equation}}
\def\ee{\end{equation}}
\def\bpm{\begin{pmatrix}}
\def\epm{\end{pmatrix}}
\def\bi{\begin{itemize}}	
\def\ei{\end{itemize}}

\begin{document}
\title{\bf Large-time asymptotics for the defocusing Manakov system on a nonzero background \\}
\author{\hspace{0.6 cm}{Xianguo Geng$^{a,b}$,  Haibing Zhang$^{a}$\footnote{\footnotesize
 Corresponding author. {\sl Email address}: haibingzhang2025@163.com}, Jiao Wei$^{a}$}\\
\leftline{\hspace{0.6 cm}{\small{\sl $^{a}$ School of Mathematics and Statistics, Zhengzhou University, 100 Kexue Road, Zhengzhou, }}}\\
\leftline{\hspace{0.6 cm}{\small{\sl \quad Henan 450001, People's Republic of China}}}\\
\leftline{\hspace{0.6 cm}{\small{\sl $^{b}$ Institute of Mathematics, Henan Academy of Sciences, Zhengzhou, Henan 450046,}}}\\
\leftline{\hspace{0.6 cm}{\small{\sl \quad People's Republic of China}}}}
\date{}
\maketitle

\begin{abstract}
The Manakov system is a two-component nonlinear Schr\"odinger equation.  In this paper, we derive a long-time asymptotic formula for the solution of the defocusing Manakov system with nonzero boundary conditions and provide a detailed proof. We first formulate the inverse problem as a $3\times3$ matrix Riemann--Hilbert problem. We then carry out the Deift--Zhou steepest descent analysis for this Riemann--Hilbert problem and obtain the long-time asymptotics in the space-time soliton region.  In this region, the leading order of the solution takes the form of a modulated multisoliton. Apart from the error term, we also discover that the defocusing Manakov system has a dispersive correction term of order $t^{-1/2}$, but this term does not exist in the scalar case, and we provide the explicit expression for this dispersion term.
\end{abstract}

\tableofcontents

\section{Introduction}

The nonlinear schr\"odinger (NLS) equation and its vector form generalizations have been one of the most important research topics in the field of mathematical physics over the past fifty years. These equations have  been derived  in many physical fields, such as deep water waves, nonlinear optics, acoustics, and Bose–Einstein condensation( see e.g.~\cite{AH1981,EG2003,CP1999,W1999,GengX2021}  and references therein). Mathematically, both the NLS equations for scalars and vectors are completely integrable infinite-dimensional Hamiltonian systems, possessing extremely rich mathematical structures. It is well known that the initial value problem for these integrable systems can be solved using the inverse scattering transform (IST). The majority of IST literature on NLS systems has primarily focused on the case of zero boundary conditions (ZBCs)—where the potential  decays to zero as the spatial variable tends to infinity. However, recent studies indicate that non-zero boundary conditions (NZBCs) are crucial for investigating modulation instability and the generation mechanisms of rogue waves  (for example, see  Refs.~\cite{ZO2009,BF2015,BLMT2018,KFF2010}). Consequently, research on the IST for NLS-type equations with non-zero backgrounds is increasing rapidly, and significant breakthroughs have been achieved in recent years. Next, we will briefly review the related research on this topic.

The IST for the  defocusing scalar NLS equation with NZBCs was  done as early as in~\cite{ZS1973,F1987} and was recently revisited in~\cite{DPVV2012}. Although the defocusing NLS equation does not admit soliton solutions on zero backgrounds, so-called dark/gray solitons were indeed found with NZBCs. The IST for the focusing NLS equation under NZBCs has recently been studied by Biondini et al.~\cite{BK2014}. They provided a number of explicit solutions on a non-zero background, including soliton and breather solutions.  Extending the IST from the scalar to the vector case requires additional considerations. Although the IST for the vector NLS equation with ZBCs was settled a long time ago~\cite{M1974,APT2004}, the case with NZBCs, even for two-component systems (i.e., the Manakov system), has remained an open problem for nearly  three decades.  The inverse scattering analysis for the defocusing Manakov system with NZBCs was ultimately completed in Refs.~\cite{PAB2006,BD2015-1}, where the authors utilized ideas originally introduced for solving the initial value problem of the three-wave interaction equations~\cite{K1976}.  By constructing two new  auxiliary eigenfunctions  from the solutions of the adjoint spectral problem, it is ensured that the Jost solutions and the auxiliary solutions form a complete set, thereby enabling the derivation of a suitable Riemann-Hilbert (RH) problem. Soon after, the focusing Manakov system with NZBCs was studied in~\cite{BD2015-2} using similar methods. Unfortunately, the approach used in~\cite{PAB2006,BD2015-1,BD2015-2} cannot be extended to the vector NLS equation with more than two components. Recently, Prinari et al.~\cite{PBA-2011} have made a significant step towards the IST for the defocusing $N$-component NLS equation with NZBCs. However, several key issues remain open. A further IST characterization for the defocusing  $N$-component NLS equation was presented in~\cite{BKP-2016} for $N=3$.  These results were recently extended to the defocusing $N$-component ($N \geq 4$)~\cite{LSG2025} and the focusing $N$-component cases~\cite{LSG-2023}. Regarding these advances, we suggest that readers refer to \cite{P2023} to obtain a comprehensive review on this topic.

In the modern version of inverse scattering theory, the inverse problem is usually formulated as a RH problem. In the pioneering paper \cite{pz93}, Deift and Zhou introduced the nonlinear steepest descent method for oscillatory RH problems, thereby establishing the long-time asymptotic behavior of solutions to the modified KdV equation. Since then, the nonlinear steepest descent method  and its extensions have become a fundamental tool for studying long-time asymptotics of integrable systems; see for example~\cite{AAD2010,JLPS2018,JLPS20181,Lenells2017,BKS2013,CLW2023,YF2023,GL2018,CL2024main,CJ2024}. The long-time asymptotic behavior of the defocusing and  focusing scalar NLS equations on the zero background has been thoroughly studied ~\cite{PZ94,PZ2002,PZto94,DM2008,BJM2018}. Meanwhile, there have been numerous studies on the long-time asymptotics for the focusing and defocusing scalar NLS equation with NZBCs~\cite{BM2017,BLM2021,DP2019,DP2024,DLM2020,H2003-1,H2002,CuJe2016,WF2022,WF2023}.
In particular, in Ref.~\cite{CuJe2016},  Cuccagna and Jenkins investigated the asymptotic stability of dark solitons for the defocusing NLS equation using the $\bar{\partial}$ steepest descent method.  Our work was greatly motivated by their paper. Despite these advancements,  the long-time asymptotic behavior for the vector
NLS equations on a non-zero background remains unexplored.  This is also an important open problem mentioned in Refs.~\cite{BD2015-1,BKP-2016,P2023}.

This work is the first in a series of articles aimed at addressing this open problem, with a primary focus on the defocusing two-component case. The defocusing Manakov system is given by 
\be \label{E:demanakovS}
\ii \q_t +\q_{xx}+2(q_0^2 -\| \q \|^2) \q=0,
\ee
where $\q(x,t)$ is a two-component vector-valued function, i.e., $\q=(q_1,q_2)^{\top}$. We will analyze the long-time behavior of solutions to system~\eqref{E:demanakovS} satisfying the following boundary conditions:
\be \label{E:bjtj}
\lim_{x \to \pm \infty} \q(x,t)= \q_{\pm}=\q_c \e^{\ii \theta_{\pm}},
\ee
where  $\q_{\pm}$ is a constant two-component  vector, (i.e., $\q_{\pm}=(q_{1,\pm},q_{2,\pm})^{\top}$), $\|\cdot \|$ is the standard Euclidean norm, $\q_c$ is a constant vector with norm $q_0$, $\theta_{\pm} \in [0,2 \pi)$, and subscripts $x$ and $t$ denote partial differentiation. Moreover, we  assume that  the convergence in~\eqref{E:bjtj}  is sufficiently fast.  Note that system~\eqref{E:demanakovS} here differs from the classical Manakov equation~\cite{M1974} by an additional term $2 q_0^2 \q$. One can remove this term by the simple rescaling $\q(x,t) \to \q(x,t) \e^{-2\ii q_0^2 t}$, and this term was added so that the boundary conditions~\eqref{E:bjtj} are independent of time.  The boundary condition~\eqref{E:bjtj} is usually referred to as the parallel NZBCs, as $\q_+=\q_- \e^{\ii(\theta_+ - \theta_-)}$.

The strategy we employ to study the long-time behavior of the solutions to~\eqref{E:demanakovS}$-$\eqref{E:bjtj} is to apply the Deift-Zhou steepest descent method to the RH problem derived in Ref.~\cite{BD2015-1}. Importantly, as will be seen, the Deift-Zhou analysis in the vector case presents new features and technical difficulties. We outline some of them below:  $\mathbf{(a)}$    The RH problem associated with~\eqref{E:demanakovS} and~\eqref{E:bjtj} involves a $3 \times 3$  matrix rather than a $2 \times 2$ matrix. Additionally, it can be seen that the original jump matrix in Ref.~\cite{BD2015-1} has an extremely complex structure~(see~\cite[Eq.3.2]{BD2015-1}), which makes the triangular factorization extremely difficult to carry out. The key observation is that the symmetry properties of the reflection coefficients can be utilized to simplify the jump matrix, thereby obtaining the corresponding triangular factorization. $\mathbf{(b)}$  The initial RH problem is singular at the origin. To reduce it eventually to a small-norm RH problem, this singularity has to be removed. Our treatment is inspired by the scalar case studied by Cuccagna and Jenkins~\cite{CuJe2016}, but the present analysis is more involved because of the \(3\times3\) structure. More precisely, we remove the discrete spectrum and the singularity at the origin simultaneously by multiplying on the right by the inverse of an outer parametrix. This step introduces a further difficulty: the inverse outer model may have singularities near the branch points \(\pm q_0\). We show that these singularities are in fact removable after the transformation, and hence the problem can still be reduced to a small-norm RH problem.
$\mathbf{(c)}$  The contour deformation is also more delicate in the vector case. One has to factorize the \(3\times3\) jump matrix according to the sign structures of three different phase functions. This substantially increases the complexity of the Deift--Zhou analysis.

Finally, we expect the methods presented in this paper can be applied to study the long-time behavior of arbitrary $N$-component defocusing NLS equations with similar NZBCs. This expectation is supported by the fact
that, in the $N$-component case, the corresponding inverse problem can be
formulated as a block $3\times3$ RH problem; see~\cite{LSG2025}.

\subsection{The preparation}

\paragraph{\bf Basic notation}
Throughout this paper, an asterisk denotes complex conjugation, while the superscripts
$\top$ and $\dagger$ denote the transpose and the conjugate transpose of a matrix,
respectively. The symbols $C>0$ and $c>0$ stand for generic positive constants whose
values may change from line to line. Unless otherwise stated, $\log (z)$ always denotes
the principal branch of the logarithm.  We write
$
\mathbb{R}_{+}=(0,\infty),
$
$
\mathbb{R}_{-}=(-\infty,0),
$
and denote by $\mathbb{C}_{+}$ and $\mathbb{C}_{-}$ the upper and lower half-planes of
the complex plane, respectively. For $1\leq p\leq\infty$, $L^p(\mathbb{R}_{\pm})$
denotes the usual Lebesgue space on $\mathbb{R}_{\pm}$. Finally, for a domain
$D\subset\mathbb{C}$, $\overline{D}$ denotes its closure.

We first state the assumptions adopted in this paper.
\begin{assume}\label{As:1}
Suppose that the initial data $\q_0(x)$ is sufficiently smooth on $\R$ and satisfies the following assumptions:
\bi
\item    Suppose that there exists a constant $\varepsilon >0$ such that
\be \label{E:aszssj}
\partial_x^j  \left( \q_0(x)-\q_{\pm} \right) \e^{\pm 2 \varepsilon x}   \in L^1(\R_{\pm}), \qquad  j=0,1,2.
\ee

\item Suppose that the analytic scattering coefficient $a_{11}(z)$ associated with $\q_0(x)$ has only finitely many simple zeros $\{\zeta_j\}_{j=0}^{N-1}$, all lying on the upper semicircle
$$
C_0=\{z\in\mathbb{C}: |z|=q_0,\ \Im z>0\}.
$$
Moreover, we order the eigenvalues $\zeta_j$ such that
$$
\Re \zeta_0>\Re \zeta_1>\cdots>\Re \zeta_{N-1}.
$$

\item  We  assume that the discrete scattering data $\{\zeta_j, \tau_j \}_{j=0}^{N-1}$ generated by $\q_0(x)$ satisfies
\be
  \frac{\tau_j}{\zeta_j}<0,\  \ \ \text{for all\ \ $0\leq j \leq N-1$}.
\ee

\item   We assume that the initial value $\q_0(x)$ ensures the scattering coefficient  $a_{11}(z)$  exhibits general behavior as $z$ approaches the branch points $\pm q_0$, specifically
\be \label{E:a11qx}
\lim_{z \to \pm q_0} (z\mp q_0)  a_{11}(z) \ne 0.
\ee
\ei
\end{assume}
\begin{remark}
The exponential convergence of the initial data to the background is imposed
mainly for technical convenience.  Under this assumption, the reflection coefficients can be analytically extended beyond the real axis.  Such an assumption is commonly adopted in
the Deift--Zhou nonlinear steepest descent analysis.
\end{remark}

\begin{remark}
As noted in the Ref.~\cite{BD2015-1}, $a_{11}(z)$ may possesses zeros within the upper semicircle $C_0$. The presence of such a discrete spectrum does not pose fundamental difficulties,  but it does increase the complexity of the calculation; therefore, we have omitted it for simplicity.
\end{remark}

\begin{remark}
In this paper, the residue constant $\tau_j$ associated with $\zeta_j$ is defined by
$
\tau_j=c_j/a'_{11}(\zeta_j)
$,
where $c_j$ is the constant defined in Theorem~2.23 of Ref.~\cite{BD2015-1}.
It follows from Ref.~\cite[Eq.~(3.5a) ]{BD2015-1} that
$\tau_j/\zeta_j\in\mathbb{R}$. As noted in Ref.~\cite{BD2015-1}, the reconstructed soliton solutions are singular for $\tau_j/\zeta_j>0$ and regular for $\tau_j/\zeta_j<0$. 
Throughout this paper, we focus on the regular case.
\end{remark}

\begin{remark}
Assume that equation \eqref{E:a11qx} holds, it implies that  the analytic scattering coefficient $a_{11}(z)$ possesses first-order singularities at the branch points $\pm q_0$. As discussed in the scalar case (see~\cite[Appendix C]{CuJe2016}), this represents the general situation.
\end{remark}

To define the asymptotic regions, we need to examine the velocities of the solitons. According to the results in Ref.~\cite{BD2015-1}, from the perspective of inverse scattering, the 1-soliton solution of system~\eqref{E:demanakovS} with NZBCs~\eqref{E:bjtj}  can be recovered from the scattering data $\big\{ \zeta , \tau , r_1(z)=r_2(z)\equiv 0\big\}$, where $\zeta$ denotes the discrete spectrum lying on the upper semicircle $C_0$,   $\tau$ is the residue constant satisfying $\frac{\tau}{\zeta} <0$, and $\{ r_1, r_2\}$ represent the reflection coefficients (under our notation). By solving the pure soliton RH problem in Ref.~\cite{BD2015-1}, the 1-dark-soliton solution can be given by the following explicit expression:
\be \label{E:1soliton}
\begin{aligned}
\q_{sol}(x,t)
=
\q_+\e^{\ii\theta}
&\left[
\cos\theta
-\ii\sin\theta
\tanh\Bigl(
q_0\sin\theta
\bigl(x-2q_0\cos\theta t-x_{0}\bigr)
\Bigr)
\right],\\
\zeta=&q_0\e^{\ii\theta},
\qquad
x_{0}
:=
\frac{1}{2q_0\sin\theta}
\log\left(
-\frac{\tau}
{2\zeta \sin\theta}
\right).
\end{aligned}
\ee
The above expression shows that each soliton travels with velocity
$v=2\Re\zeta$. Setting $\xi=x/(2t)$ and using the fact that
$\Re\zeta<q_0$, we divide the $(x,t)$-plane into the following three
asymptotic regions:
\bi
\item[1.]  The soliton region $\cR_{sol}$: $\{(x,t)|\  |\xi|<q_0 \}$.
\item[2.]  The solitonless region: $\{(x,t)|\  |\xi| >q_0 \}$.
\item[3.]  The transition  region $\{(x,t)|\  |\xi| \approx q_0 \}$. This  region connects the soliton region and the solitonless region.
\ei
In this paper, we focus on the long-time asymptotics in the soliton region,
whereas the asymptotic analysis in the remaining regions will be left for
future work.

The above division is consistent with that in the scalar case. However, in
the present vector setting, the sign of $\xi$ leads to different
Deift--Zhou steepest descent analyses. We therefore split the soliton
region into the right and left parts. Let
$$\mathcal I_+=[m_0,m_1]\subset(0,q_0), \quad \text{and} \quad 
\mathcal I_-=[\widetilde m_0,\widetilde m_1]\subset(-q_0,0)$$ be two
closed subintervals. Since the solitons have distinct velocities, these
intervals can always be chosen so that each of them satisfies one of the
following two alternatives:
(i) it contains exactly one soliton velocity;
(ii) it contains no soliton velocities.
Accordingly, we shall derive separately the long-time asymptotic formula
near a single soliton and the formula away from solitons. Patching these
local descriptions together then gives the asymptotics throughout the two
parts of the soliton region.

To state the main result, we  introduce several quantities determined by the
scattering data. The scattering matrix $\A(z)$ is defined by (2.7). Let
$
\B(z)=\A^{-1}(z)
$,
and denote by $a_{ij}$ and $b_{ij}$ the $(i,j)$-entries of the matrices
$\A$ and $\B$, respectively. We define the reflection coefficients by
$$
r_1(z)=\frac{a_{21}(z)}{a_{11}(z)},\qquad
r_2(z)=\frac{a_{31}(z)}{a_{33}(z)},\qquad
r_3(z)=\frac{a_{23}(z)}{a_{33}(z)} .
$$
We also need another set of modified reflection coefficients,
$$
\hat r_1(z)=\frac{b_{21}(z)}{b_{11}(z)},\qquad
\hat r_3(z)=\frac{b_{23}(z)}{b_{33}(z)} .
$$
Let
$
z_0=\frac{q_0^2}{\xi}$, $ z_1=\xi$.
We define the negative-valued function $\nu(\xi)$ by
$$
\nu=\nu(\xi):=
\begin{cases}
-\frac{1}{2\pi}\log\bigl(1+\frac{1}{\gamma(z_0)}
\left|\hat r_3(z_0)\right|^2\bigr),
& \xi \in \mathcal I_+,\\
-\frac{1}{2\pi}\log\bigl( 1+\frac{1}{\gamma(z_0)}
\left|r_3(z_0)\right|^2\bigr),
& \xi \in \mathcal I_-,
\end{cases}
$$
where
$
\gamma(z)=1-\frac{q_0^2}{z^2}.
$
We denote the discrete scattering data generated by the initial data by
$
\left\{\zeta_j,\tau_j\right\}_{j=0}^{N-1}
$.
For a fixed \(j_0\in\{0,1,\dots,N-1\}\), The modified residue constant $\widetilde{\tau}_{j_0}$ associated with $\zeta_{j_0}$ is defined by
\be \label{E:taujfbdsss}
\widetilde \tau_{j_0}
=
\tau_{j_0}\mathcal D_{j_0}\mathcal R_{j_0}(\xi),
\ee
where
$$
\mathcal D_{j_0}:=
\left(
\prod_{\ell<j_0}
\left|
\frac{\zeta_{j_0}-\zeta_{\ell}}
{\zeta_{j_0}-\zeta_{\ell}^*}
\right|^2 
\right) \cdot
\exp\left\{
-\frac{\Im\zeta_{j_0}}{\pi}
\int_0^{+\infty}
\frac{
\log\left(1-\frac{1}{\gamma(s)}|r_1(s)|^2-|r_2(s)|^2\right)}
{|s-\zeta_{j_0}|^2}\,\mathrm d s
\right\},
$$
and
\[
\mathcal R_{j_0}(\xi):=
\begin{cases}
\exp\left\{
\frac{\Im\zeta_{j_0}}{\pi}
\int_{z_0}^{+\infty}
\frac{
\log\left(1+\frac{1}{\gamma(s)}|\hat r_3(s)|^2\right)}
{|s-\zeta_{j_0}|^2}\,\mathrm d s
\right\},
& \xi\in\mathcal I_+,\\
\exp\left\{
-\frac{\Im\zeta_{j_0}}{\pi}
\int_{-\infty}^{z_0}
\frac{
\log\left(1+\frac{1}{\gamma(s)}|r_3(s)|^2\right)}
{|s-\zeta_{j_0}|^2}\,\mathrm d  s
\right\},
& \xi\in\mathcal I_- .
\end{cases}
\]
Next, we define the function $\alpha(\xi)$ by
\be \label{E:alphaxidebds}
\begin{aligned}
\alpha(\xi)
=&
\frac{1}{2\pi}
\int_0^{+\infty}
s^{-1}
\log\biggl(
1-\frac{1}{\gamma(s)}|r_1(s)|^2-|r_2(s)|^2
\biggr)\,\mathrm d  s
+
2\sum_{\ell:\ \Re\zeta_\ell>\xi}
\arg \zeta_\ell  \\
&+
\begin{cases}
-\frac{1}{2\pi}
\int_{z_0}^{+\infty}
s^{-1}
\log\bigl(
1+\frac{1}{\gamma(s)}|\hat r_3(s)|^2
\bigr)\,\mathrm  d s,
& \xi\in \mathcal I_+,\\
\frac{1}{2\pi}
\int_{-\infty}^{z_0}
s^{-1}
\log\bigl(
1+\frac{1}{\gamma(s)}|r_3(s)|^2
\bigr)\,\mathrm  d  s,
& \xi \in \mathcal  I_- .
\end{cases}
\end{aligned}
\ee
Then we have the following results.

\subsection{Main results}
\begin{theorem}[Asymptotics near a soliton]\label{asy-th-1}
Let $\mathbf q(x,t)$ be a smooth solution of the defocusing Manakov
system~\eqref{E:demanakovS} with the NZBCs~\eqref{E:bjtj}, and assume that
the initial data $\mathbf q_0(x)$ satisfies Assumptions~\ref{As:1}. Let
$\alpha(\xi)$ be defined by~\eqref{E:alphaxidebds}. Fix
$\diamondsuit\in\{+,-\}$, and assume that $\mathcal I_\diamondsuit$
contains exactly one soliton velocity, namely
\[
\mathcal I_\diamondsuit\cap \{\Re\zeta_j\}_{j=0}^{N-1}
=
\{\Re\zeta_{j_0}\}.
\] Then, as
$t\to\infty$, the following asymptotic formula holds uniformly for
$\xi\in\mathcal I_\diamondsuit$:
\begin{align}\label{E:asyofq}
\mathbf q(x,t)
={}&
\e^{\ii\alpha(m_\star)}
\mathbf q_{ sol}^{(j_0)}(x,t)
+
\frac{A(\xi)\e^{\ii\Phi(x,t)}}{\sqrt t}
\biggl((\M_{ sol}^{(j_0)})^{-1}(x,t,z_1)\biggr)_{11}
\bpm
-\frac{q_{2,+}^*}{q_0}\\[1mm]
\frac{q_{1,+}^*}{q_0}
\epm
+\mathcal O(t^{-1}\log t),
\end{align}
where $m_\star:=\sup\mathcal I_\diamondsuit$.
Here $\mathbf q_{sol}^{(j_0)}(x,t)$ is the one-dark-soliton solution
associated with the modified scattering data $\{\zeta_{j_0},\widetilde\tau_{j_0}\}$,
where $\widetilde\tau_{j_0}$ is given by~\eqref{E:taujfbdsss}. Explicitly,
\[
\mathbf q_{sol}^{(j_0)}(x,t)
=
\mathbf q_+\e^{\ii\theta_{j_0}}
\left[
\cos\theta_{j_0}
-\ii\sin\theta_{j_0}
\tanh\Bigl(
q_0\sin\theta_{j_0}
\bigl(x-2q_0\cos\theta_{j_0}t-x_{j_0}\bigr)
\Bigr)
\right],
\]
\[
\zeta_{j_0}=q_0\e^{\ii\theta_{j_0}},
\qquad
x_{j_0}:=
\frac{1}{2q_0\sin\theta_{j_0}}
\log\left(
-\frac{\widetilde\tau_{j_0}}
{2\zeta_{j_0}\sin\theta_{j_0}}
\right).
\]
Moreover,
\[
A(\xi)
:=
-\gamma(z_1)\sqrt{\frac{\gamma(z_1)\nu(\xi)}{2}}>0,
\qquad
\Phi(x,t):=
\begin{cases}
\Phi_1(x,t), & \xi\in\mathcal I_+,\\
\Phi_2(x,t), & \xi\in\mathcal I_-.
\end{cases}
\]
The phase $\Phi_1$ is given by
\[
\begin{aligned}
\Phi_1(x,t)
={}&
-\frac{\pi}{4}
+\arg \hat r_3(z_0)
-\arg \Gamma(-i\nu) +\xi^2 t
+2\nu\log\left(\frac{q_0^2}{\xi^2\sqrt{2t}}\right)
+\arg\left(\prod_{\ell=0}^{N-1}\frac{\zeta_{\ell}^*}{\zeta_{\ell}}\right)
\\
&+
\arg\biggl(
\prod_{\ell:\ \Re\zeta_\ell\le m_*}
\frac{z_1-\zeta_\ell^*}{z_1-\zeta_\ell}
\biggr)
-
\frac{1}{\pi}
\int_{z_0}^{+\infty}
\log(s-z_0)\,
 \mathrm  d \log\bigl(
1+\frac{1}{\gamma(s)}
|\hat r_3(s)|^2
\bigr)
\\
&-
\frac{1}{2\pi}
\int_{z_0}^{+\infty}
\frac{s-2z_1}{s(s-z_1)}
\log\bigl(
1+\frac{1}{\gamma(s)}
|\hat r_3(s)|^2
\bigr)\,\mathrm  d s
\\
&-
\frac{1}{2\pi}
\int_0^{+\infty}
s^{-1}
\log\bigl(
1-\frac{1}{\gamma(s)}|r_1(s)|^2-|r_2(s)|^2
\bigr)\,\mathrm  d  s
\\
&+
\frac{1}{2\pi}
\int_{-\infty}^{0}
\frac{3s-2z_0-z_1}{(s-z_1)(s-z_0)}
\log\bigl(
1-\frac{1}{\gamma(s)}|r_1(s)|^2-|r_2(s)|^2
\bigr)\,\mathrm  d s.
\end{aligned}
\]
The phase $\Phi_2$ is given by
\[
\begin{aligned}
\Phi_2(x,t)
={}&
-\frac{\pi}{4}
+\arg r_3(z_0)
-\arg \Gamma(-i\nu)
+2\nu\log\left(\frac{q_0^2}{\xi^2\sqrt{2t}}\right)+
\arg\biggl(
\prod_{\ell:\ \Re\zeta_\ell>m_*}
\frac{z_1-\zeta_\ell}{z_1-\zeta_\ell^*}
\biggr)
+\xi^2t
\\
&+
\frac{1}{2\pi}
\int_0^{+\infty}
\frac{2s^2-z_0s-q_0^2}
{s(s-z_1)(s-z_0)}
\log\bigl(
1-\frac{1}{\gamma(s)}|r_1(s)|^2-|r_2(s)|^2
\bigr)\,\mathrm  d  s
\\
&+
\frac{1}{2 \pi}
\int_{-\infty}^{z_0}
\frac{s-2z_1}{s(s-z_1)}
\log\bigl(
1+\frac{1}{\gamma(s)}
|r_3(s)|^2
\bigr)\,\mathrm  d  s
\\
&+
\frac{1}{\pi}
\int_{-\infty}^{z_0}
\log(z_0-s)\,
\mathrm  d\log\bigl(
1+\frac{1}{\gamma(s)}
|r_3(s)|^2
\bigr).
\end{aligned}
\]
Finally, $\M_{ sol}^{(j_0)}(x,t,z)$  denotes the solution of the pure-soliton
RH problem~\ref{rhp:Msol} with the scattering data
replaced by $\{\zeta_{j_0},\widetilde\tau_{j_0}\}$.

\end{theorem}

The above result applies when the observation ray is located near a prescribed
soliton velocity. To complete the analysis in the soliton region, we next consider
compact velocity intervals that stay away from all soliton velocities.

\begin{theorem}[Asymptotics away from solitons]\label{Th:2m}
Under the assumptions of Theorem~\ref{asy-th-1}, suppose that
$\mathcal I_\diamondsuit$ contains no soliton velocity, namely
\[
\mathcal I_\diamondsuit \cap \{\Re\zeta_j\}_{j=0}^{N-1}=\emptyset.
\]
Then, as $t\to\infty$, the following asymptotic formula holds uniformly for
$\xi\in\mathcal I_\diamondsuit$:
\begin{align} \label{E:bjgs}
\mathbf q(x,t)
={}&
\e^{\ii\alpha(m_\star)}\mathbf q_+
+
\frac{\widetilde A(\xi)\e^{\ii\widetilde\Phi(x,t)}}{\sqrt t}
\bpm
-\dfrac{q_{2,+}^*}{q_0}\\[1mm]
\dfrac{q_{1,+}^*}{q_0}
\epm
+\mathcal O(t^{-1}\log t),
\end{align}
where
\[
\widetilde A(\xi)
=
-\frac{A(\xi)}{\gamma(z_1)}
=
\sqrt{\frac{\gamma(z_1)\nu(\xi)}{2}}>0,
\qquad
\widetilde\Phi(x,t)=\Phi(x,t)+\pi .
\]
Here $m_\star$ and $\Phi(x,t)$ are defined as in
Theorem~\ref{asy-th-1}.
\end{theorem}

Based on the above results, patching these local descriptions together
gives the asymptotics throughout the two parts of the soliton region. In
particular, the leading asymptotic term can be expressed as a finite sum of
soliton contributions.

\begin{corollary}\label{C:soliton-decomposition}
Let $\mathbf q(x,t)$ be a smooth solution of the defocusing Manakov
system~\eqref{E:demanakovS} with the NZBCs~\eqref{E:bjtj}, and assume that
the initial data $\mathbf q_0(x)$ satisfies Assumptions~\ref{As:1}. Let $\mathcal I_+$ and $\mathcal I_-$ be arbitrary closed subintervals of
$(0,q_0)$ and $(-q_0,0)$, respectively. Then, for each
$\diamondsuit\in\{+,-\}$, the following asymptotic formula holds uniformly
for $\xi\in\mathcal I_\diamondsuit$ as $t\to\infty$:
\be \label{E:soliton-decomposition}
\mathbf q(x,t)
=
\e^{\ii \tilde{\alpha}(\xi)}
\left[
\mathbf q_+
+
\sum_{k=0}^{N-1}
\left(
\prod_{j<k}\frac{\zeta_j}{\zeta_j^*}
\right)
\left(
\mathbf q_{sol}^{(k)}(x,t)-\mathbf q_+
\right)
\right]
+\mathcal O(t^{-1/2}).
\ee
Here $\tilde{\alpha}(\xi)$ is defined by
\be \label{E:tildealphaxidebds}
\begin{aligned}
\tilde{\alpha}(\xi)
=\alpha(\xi)-2\sum_{\ell:\ \Re\zeta_\ell>\xi}
\arg \zeta_\ell,
\end{aligned}
\ee
 and
$\mathbf q_{sol}^{(k)}(x,t)$ denotes the one-dark-soliton solution
associated with the modified scattering data
$\{\zeta_k,\widetilde\tau_k\}$, where $\widetilde\tau_k$ is defined
by~\eqref{E:taujfbdsss}.
\end{corollary}

\begin{proof}
For any $k\in\{0,\ldots,N-1\}$, choose $\varrho>0$ sufficiently small
such that $|\xi-\Re\zeta_k|\leq \varrho$. It is straightforward to verify that, for
$j>k$, the soliton generated by $\zeta_j$ satisfies
\[
\mathbf q_{sol}^{(j)}(x,t)
=
\mathbf q_+ + \mathcal O(\e^{-ct}),
\qquad
|\xi-\Re\zeta_k|\leq \varrho,\quad t\to\infty.
\]
For $j<k$, one has instead
\[
\mathbf q_{sol}^{(j)}(x,t)
=
\frac{\zeta_j}{\zeta_j^*}\mathbf q_+
+\mathcal O(\e^{-ct}),
\qquad
|\xi-\Re\zeta_k|\leq \varrho,\quad t\to\infty.
\]
When $j=k$, the soliton $\mathbf q_{sol}^{(j)}(x,t)$ has the asymptotic
behavior given by~\eqref{E:asyofq}. Therefore, a direct algebraic calculation
shows that the right-hand side of~\eqref{E:soliton-decomposition} is equivalent to the
right-hand side of~\eqref{E:asyofq}.

It remains to consider the part of $\mathcal I_\diamondsuit$ away from all
soliton velocities. Choose $\varrho>0$ sufficiently small so that the
neighborhoods
\[
\{\xi:\ |\xi-\Re\zeta_k|<\varrho\},\qquad k=0,\ldots,N-1,
\]
are mutually disjoint. Then the set
\[
\mathcal K_\diamondsuit
:=
\mathcal I_\diamondsuit
\setminus
\bigcup_{k=0}^{N-1}
\{\xi:\ |\xi-\Re\zeta_k|<\varrho\}
\]
is a finite union of compact intervals, and none of these intervals contains
a soliton velocity. Hence Theorem~\ref{Th:2m} applies on each connected
component of $\mathcal K_\diamondsuit$.
On any such compact component, the number of solitons which have passed the
observation ray is constant. Denote this number by $K$. Then, uniformly for $\xi$ in this component as $t\to\infty$,  the
asymptotic behavior of the $j$-th one-soliton profile is given by
\[
\mathbf q_{sol}^{(j)}(x,t)=
\begin{cases}
\dfrac{\zeta_j}{\zeta_j^*}\mathbf q_+
+\mathcal O(\e^{-ct}), & j<K,\\[2mm]
\mathbf q_+
+\mathcal O(\e^{-ct}), & j\ge K.
\end{cases}
\]
Thus, the right-hand side
of~\eqref{E:soliton-decomposition} is equivalent to the right-hand side of~\eqref{E:bjgs}.
Consequently, formula~\eqref{E:soliton-decomposition} follows.
\end{proof}

We now make several remarks on these results. 

\begin{remark}
If the first radiative correction is retained, then the coefficient of the
$t^{-1/2}$ term can be written explicitly; see formula~\eqref{E:asjjgdddgs}.
In Corollary~\ref{C:soliton-decomposition}, this contribution has been
absorbed into the error term $\mathcal O(t^{-1/2})$.
\end{remark}

\begin{remark}
Our results only provide the long-time asymptotic behavior when $x/(2t)$ belongs to arbitrary closed subintervals of $(-q_0,0)$ and $(0,q_0)$. Although we believe that the asymptotic formulas on the two sides are compatible as $\xi\to 0$, we do not obtain an asymptotic formula that is uniform in the region $\xi\approx 0$.
\end{remark}

\begin{remark}
The $\mathcal{O}(t^{-1/2})$ term in the asymptotic formula is a genuinely vectorial effect. Since the Manakov system is invariant under constant unitary transformations
and the boundary values are parallel, one may apply a constant unitary
transformation so that
$$
\mathbf q_+=(0,q_{2,+})^{\top},\qquad |q_{2,+}|=q_0 .
$$
In this setting, the $\mathcal O(t^{-1/2})$ radiative correction
appears only in the first component, which has zero background, while its
direct contribution to the second, nonzero-background component vanishes. This is consistent with the scalar nonzero-background defocusing NLS asymptotics~\cite{CuJe2016}, where such an additional radiative term is absent in the soliton region.
\end{remark}

This paper is organized as follows. In section~\ref{s:rhch}, we review the RH characterization of~\eqref{E:demanakovS} and~\eqref{E:bjtj}. Section~\ref{se:3} is devoted to the proof of Theorems~\ref{asy-th-1} and~\ref{Th:2m} by employing the Deift-Zhou steepest descent method.  Some technical proofs are presented in Appendices~\ref{App:AAA} to \ref{App:ccc}.

\section{A Riemann-Hilbert  formulation}\label{s:rhch}
In this section, we review the IST for the defocusing Manakov system~\eqref{As:1} with the NZBCs~\eqref{E:bjtj}. Most of these results can be found in Ref.~\cite{BD2015-1}.

It is well-known that the defocusing Manakov system is completely integrable because it possesses a $3\times 3$ matrix Lax pair~\cite{BD2015-1}:
\be \label{E:laxp1}
\bPhi_x=\widehat{\bX} \bPhi, \qquad \bPhi_t=\widehat{\bT} \bPhi,
\ee
where
\begin{align*}
&\widehat{\bX}(x,t,k)=-\ii k \bJ +\bQ, \qquad \widehat{\bT}(x,t,k)=2 \ii k^2 \bJ- \ii \bJ \left(\bQ_x-\bQ^2+q_0^2 \right)-2k \bQ,\\
&\bJ=\bpm 1& \mathbf{0}^{\top}\\
\mathbf{0} &-\bI_{2 \times 2}
\epm,  \quad
\bQ=\bpm
0& \q^{\dagger }\\
\q &  \mathbf{0}_{2 \times 2}
\epm.
\end{align*}
It can be expected that when $x \to \pm \infty$, the solutions to the scattering problem will be approximated by those of the asymptotic scattering problem
$$ \bPhi_x = \widehat{\bX}_{\pm} \bPhi, \qquad    \bPhi_t = \widehat{\bT}_{\pm} \bPhi,$$
where $\widehat{\bX}_{\pm}= \lim_{x \to \pm \infty} \widehat{\bX}$ and $\widehat{\bT}_{\pm}=\lim_{x \to \pm \infty} \widehat{\bT}$. The eigenvalues of $\widehat{\bX}_{\pm}$ are $\ii k$ and $\pm\ii\lambda$, where 
\be \label{E:lambda}
 \lambda(k)=\sqrt{k^2-q_0^2}. 
\ee 
The function $\lambda(k)$ is multi-valued and hence gives rise to a branching structure. Following \cite{BD2015-1}, we introduce the two-sheeted Riemann surface associated with~\eqref{E:lambda}. The branch points are determined by $\lambda(k)=0$, namely $k=\pm q_0$. We choose the branch cut to be $ (-\infty,-q_0]\cup[q_0,\infty)$, and fix the branch of $\lambda(k)$ by imposing $ \lambda(0)=\ii q_0$.  Next, we introduce the
uniformization variable by defining
$$
z=k+\lambda.
$$
The inverse transformation can be obtained by
\begin{align}\label{E:intr}
k = \frac{1}{2}(z + \frac{q^2_0}{z}),\quad \lambda = \frac{1}{2}(z - \frac{q^2_0}{z}).
\end{align}
Following the notation in Ref.~\cite{BD2015-1}, we denote the orthogonal vector of a two-component complex-valued vector $\bv=(v_1,v_2)$ as $\bv^{\perp}=\left(v_2,-v_1  \right)^{\dagger}$. We further introduce  three matrices
\be \label{E:trmatrix}
\E_{\pm}(z)=
\bpm
1&0&-\frac{\ii q_0}{z}\\
\ii \frac{\q_{\pm}}{z} & \frac{\q_{\pm}^{\perp}}{q_0}&\frac{\q_{\pm}}{q_0}
\epm, \qquad
 \mathbf{\Lambda}(z)=\mathrm{diag} \left( -\lambda,  k,  \lambda  \right),\qquad
\mathbf{\Omega}(z)=\mathrm{diag} \left( -2k \lambda,  k^2+\lambda^2,  2 k \lambda  \right),
\ee
which satisfy the relation
$$
\E_{\pm}^{-1} \widehat{\bX}_{\pm} \E_{\pm}=\ii \mathbf{\Lambda},\qquad
\E_{\pm}^{-1} \widehat{\bT}_{\pm} \E_{\pm}=-\ii \mathbf{\Omega}.
$$
Then, it is easy to see that the Jost solutions $\boldsymbol{\mu}_+(x,t,z)$ and $\boldsymbol{\mu}_-(x,t,z)$ defined by the following integral equations are the unique solutions:
\begin{subequations}
	\begin{align}
		&\boldsymbol{\mu}_-(x,t,z) = \E_-(z) + \int_{-\infty}^{x} \E_{-}(z) \e^{\ii   (x-y) \mathbf{\Lambda}(z)} \E^{-1}_{-}(z) \Delta \bQ_{-}(y,t) \boldsymbol{\mu}_-(y,t,z) \e^{-\ii   (x-y) \mathbf{\Lambda}(z)}  \mathrm{d}y, \label{E:mu-}\\
		&\boldsymbol{\mu}_+(x,t,z) = \E_+(z) - \int_{x}^{+\infty} \E_{+}(z) e^{\ii   (x-y) \mathbf{\Lambda}(z)} \E^{-1}_{+}(z) \Delta \bQ_{+}(y,t) \boldsymbol{\mu}_+(y,t,z) \e^{-\ii   (x-y) \mathbf{\Lambda}(z)} \mathrm{d}y, \label{E:mu+}
	\end{align}
\end{subequations}
where $\Delta \bQ_{\pm}=\bQ-\bQ_{\pm}$ with $\bQ_{\pm}=\lim_{x \to \pm \infty} \bQ(x,t)$.  
We define $ \bPhi_{\pm}(x,t,z)=\boldsymbol{\mu}_{\pm}(x,t,z) \e^{\ii \mathbf{\Lambda}(z)x - \ii \mathbf{\Omega}(z)t}$. Then, $\bPhi_{\pm}(x,t,z)$  are the fundamental solutions of Lax pair~\eqref{E:laxp1} for $z \in \R \setminus \{0,\pm q_0 \}$. This is because
\be \label{E:degamma}
\mathrm{det}  \boldsymbol{\mu}_{\pm}(x,t,z)=\mathrm{det} \E_{\pm}(z) =1-\frac{q_0^2}{z^2}=:\gamma(z).
\ee
Therefore, there exists a matrix $\A(z)$ independent of  $x$  and  $t$  such that
\be
\bPhi_-(x,t,z)=\bPhi_+(x,t,z) \A(z), \quad  z \in \R \setminus \{0, \pm q_0 \}.
\ee
By~\cite[Theorem~2.3]{BD2015-1}, under suitable decay assumptions on
$\Delta\bQ_{\pm}$ as $x\to\pm\infty$, the $(1,1)$- and $(3,3)$-entries of
$\A(z)$ can be analytically continued to the upper and lower half-planes,
respectively. In general, the other entries are defined only on the real
axis. The following proposition shows that, if the initial data approach the
background in the sense of~\eqref{E:aszssj}, then all entries of $\A(z)$
admit analytic continuations off the real axis.  Define  $\bB(z) = \A^{-1}(z)$.  Throughout the paper, we let $a_{ij}$ and $b_{ij}$ denote the $(i,j)$-th entries of the scattering matrices $\A$ and $\B$, respectively.
We now summarize
the properties of $\A(z)$ and $\bB(z)$ that follow from
\cite[Lemmas~2.13 and~2.16, Sections~2.6--2.7]{BD2015-1}.

\begin{prop}\label{P:Aas}
Suppose the initial data $\q_0(x)$ satisfies Assumptions~\ref{As:1}. Then scattering matrices $\A(z)$ and $\bB(z)$ have the following properties:
\bi
\item  $\A(z)$ and $\bB(z)$ satisfy the following symmetry properties
\be \label{E:ABdc}
\A(\hat{z})= \bPi(z) \A(z) \bPi^{-1}(z), \qquad
\bB(\hat{z})= \bPi(z) \bB(z) \bPi^{-1}(z),
\ee
where $\hat{z}=\frac{q_0^2}{z}$ and  $\bPi(z)$ is given by
\be \label{E:pi}
\bPi(z)= \bpm
0&0&-\ii \frac{q_0}{z}\\
0&1&0\\
\ii \frac{q_0}{z}&0&0
\epm.
\ee
\item $\A(z)$ and $\bB(z)$ have the following relationship:
\be \label{E:ABrel}
 (\A(z^*))^{\dagger} = \mathbf{\Gamma}^{-1}(z) \bB(z) \mathbf{\Gamma}(z),
\quad
\mathbf{\Gamma}(z)=\bpm
-1 & 0&0\\
0& \gamma(z)&0\\
0&0&1
\epm.
\ee

\item
Define
$
S_{\varepsilon}
=
\left\{z\in\mathbb{C}: |\operatorname{Im} z|\leq \varepsilon\right\}
\setminus (B_1\cup B_2)
$, where $B_1$ and $B_2$ are the disks centered at
$\frac{\ii q_0^2}{2\varepsilon}$ and $-\frac{\ii q_0^2}{2\varepsilon}$,
respectively, both with radius $\frac{q_0^2}{2\varepsilon}$. Then, as
$z\to\infty$ within $S_{\varepsilon}$, we have
\begin{equation}\label{E:Azinfty}
a_{13}(z)=\mathcal{O}\left(\frac{1}{z}\right),\qquad
a_{21}(z)=\mathcal{O}\left(\frac{1}{z^2}\right),\qquad
a_{23}(z)=\mathcal{O}\left(\frac{1}{z}\right),\qquad
a_{31}(z)=\mathcal{O}\left(\frac{1}{z}\right).
\end{equation}

\item 
As $z \to 0$ within $S_{\varepsilon}$, we have
\be \label{E:Az0}
\begin{aligned}
 a_{13}(z)=\mathcal{O}(z),\quad
 a_{21}(z)=\mathcal{O}(1), \quad  a_{23}(z)= \mathcal{O}(z),  \quad
a_{31}(z)=\mathcal{O}(z).
\end{aligned}
\ee

\item  The diagonal entries $a_{11}$ and $a_{33}$ exhibit the following asymptotic behavior:
\be \label{E:djs0}
\begin{aligned}
&a_{11}(z) =1+\mathcal{O}(1/z), \quad  S_{\varepsilon } \ni z \to \infty; \quad
a_{11}(z)=\e^{-\ii (\theta_+-\theta_-)}+\mathcal{O}(z), \quad S_{\varepsilon} \ni z \to 0,\\
& a_{33}(z)=\e^{-\ii (\theta_+-\theta_-)}+\mathcal{O}(1/z), \quad  S_{\varepsilon } \ni z \to \infty; \quad a_{33}(z)=1+\mathcal{O}(z), \quad  S_{\varepsilon} \ni z \to 0.
\end{aligned}
\ee

\item   Near the branch points  $\pm q_0$, we have
\be \label{E:asAnearq0}
\A(z)=\frac{1}{z\mp q_0}\A_{\pm}+\mathcal{O}(1), \quad  z \to \pm q_0, \quad z \in \C \setminus \{ \pm q_0\},
\ee
where
\be \label{E:Apm}
\A_{\pm}=a_{11,\pm}
\bpm
1&0&\mp \ii \\
0&0&0\\
\mp \ii &0&-1
\epm  +
a_{12,\pm}
\bpm
0&1&0\\
0&0&0\\
0& \mp \ii& 0
\epm, \qquad  a_{11,\pm} \ne 0.
\ee

\ei
\end{prop}
\begin{proof}
The proofs of~\eqref{E:Azinfty} and~\eqref{E:Az0} are given in Appendix~\ref{App:pofzas}. 
The proof of~\eqref{E:djs0} is completely analogous and is therefore omitted. 
The corresponding proofs for the remaining assertions can be found in Ref.~\cite{BD2015-1}.
\end{proof}

Next, let us define the reflection coefficients  $r_{1}(z)$,  $r_2(z)$ and $r_3(z)$ by
\be \label{E:fsxsr12}
r_1(z)=\frac{a_{21}(z)}{a_{11}(z)}, \quad r_2(z)=\frac{a_{31}(z)}{a_{11}(z)}, \quad
r_{3}(z)=\frac{a_{23}(z)}{a_{33}(z)}.
\ee
Due to symmetry~\eqref{E:ABdc}, it is straightforward to verify that  $r_1(z) =\frac{\ii q_0}{z} r_3(\hat{z})$, where $\hat{z}=\frac{q_0^2}{z}$. Since the scattering matrix elements  $a_{ij}(z)$  is analytic for $ z \in S_{\varepsilon}$, the reflection coefficients  $r_1$, $r_2$  and  $r_3$  are also analytic in  $ S_{\varepsilon}$. The following lemma is a direct consequence of Proposition~\ref{P:Aas} and characterizes the asymptotic properties of the reflection coefficients.
\begin{lemma}\label{L:fsxsxz}
Suppose the initial data $\q_0(x)$ satisfies Assumptions~\ref{As:1}.
Then the  reflection coefficients defined by~\eqref{E:fsxsr12} have the following properties:
\bi
\item
 For $z \in S_{\varepsilon }$ with $z \to \infty$, the functions  $\{r_j(z) \}_{j=1}^3$ have  the following asymptotic behavior:
\be
r_{1}(z)=\mathcal{O}(\frac{1}{z^2}), \quad
r_3(z)=\mathcal{O}(\frac{1}{z}), \quad
r_{2}(z)=\mathcal{O}(\frac{1}{z}).
\ee

\item  The functions $\{r_j(z) \}_{j=1}^3$ have well-defined limits as  $S_{\varepsilon} \ni z \to 0$.  Furthermore, we have
\be \label{E:r123z0}
\lim_{S_{\varepsilon} \ni z \to 0} r_2(z) =\lim_{S_{\varepsilon} \ni z \to 0} r_3(z)=0.
\ee

\item The functions  $\{r_j(z) \}_{j=1}^3$ have well-defined limits at the branch points $\pm q_0$.

\ei
\end{lemma}


\begin{remark}\label{R:fsjs}
Although the limit of $r_{1}(z)$ exists as $z\to0$ within $S_{\varepsilon}$, in general
one cannot conclude that
$
\lim\limits_{S_{\varepsilon} \ni z \to 0} r_1(z) =0.
$
\end{remark}

Following Ref.~\cite{BD2015-1}, we define a piecewise meromorphic function $\M(x,t,z)$ as follows~(see~\cite[Eq.(3.1a) and Eq.(3.1b)]{BD2015-1}):
\be \label{E:exM}
\M(x,t,z)=\begin{cases}
\left(\frac{\boldsymbol{\mu}_{-1}}{a_{11}}, \frac{\m}{b_{33}}, \boldsymbol{\mu}_{+3}   \right), & \mathrm{Im} z >0,\\
\left(\boldsymbol{\mu}_{+1}, -\frac{\bar{\m} }{b_{11}},   \frac{\boldsymbol{\mu}_{-3}}{a_{33}}  \right), & \mathrm{Im} z <0,
\end{cases}
\ee
where $
\bar{\m}(x,t,z)=-\bJ[\bPhi_{-1}^*  \times \bPhi_{+3}^*] (x,t,z^*) / \gamma(z),\
\m(x,t,z)=-\bJ[\bPhi_{-3}^*  \times \bPhi_{+1}^*] (x,t,z^*) / \gamma(z)$. Here $``\times"$ denotes the usual cross product. Then, according to the results in Ref.~\cite{BD2015-1}, one can conclude that $\M(x,t,z)$ satisfies the following RH problem:
\begin{RHP}\label{RHP:zc}
Find a $3 \times 3$ matrix-valued function $\M(x,t,z)$ with the following properties:
\bi
\item $\M(x,t,\cdot) : \mathbb{C}\setminus \{\rZ \cup \R \} \to \mathbb{C}^{3 \times 3}$ is analytic, where $\rZ=\{ \zeta_j \}_{j=0}^{N-1} \cup \{ \zeta_j^* \}_{j=0}^{N-1}$.
\item $\M(x,t,z)$ satisfies the jump condition:
\be \label{E:Jump}
\M_+(x,t,z)=\M_-(x,t,z) \V(x,t,z), \quad z \in \R \setminus \{0 \}.
\ee
Here,
\be \label{E:Vex}
\V=\e^{\Theta(x,t,z)}
\bpm
1-\frac{1}{\gamma(z)}|r_1(z)|^2-|r_2(z)|^2& \frac{1}{\gamma(z)}(-r_1(z)+r_2(z)r_3(z))^*& -r^*_2(z) \\
r_1(z)-r_2(z)r_3(z) &1+\frac{1}{\gamma(z)}|r_3(z)|^2& -r_3(z)\\
r_2(z) & -\frac{1}{\gamma(z)}r_3^*(z)& 1
\epm
\e^{-\Theta(x,t,z)},
\ee
where $\Theta(x,t,z)=\mathrm{diag}\left(\theta_1(x,t,z),\theta_2(x,t,z),\theta_3(x,t,z)   \right)$ with
\be \label{E:theta123}
\begin{aligned}
\theta_1(x,t,z)&=-\ii\lambda(z)x+2\ii k(z)\lambda(z)t,\\
\theta_2(x,t,z)&=\ii k(z)x-\ii\bigl(k^2(z)+\lambda^2(z)\bigr)t,\\
\theta_3(x,t,z)&=\ii\lambda(z)x-2\ii k(z)\lambda(z)t.
\end{aligned}
\ee
\item  $\M$ admits the asymptotic behavior:
$$  \M=\M_{\infty}+\mathcal{O}(\frac{1}{z}), \quad z \to \infty ; \quad \M =\frac{\ii}{z} \M_0 + \mathcal{O}(1), \quad z \to 0,  $$
where
\be \label{E:M0infty}
\M_{\infty}=\bpm
1&0&0\\
\mathbf{0}& \q_+^{\perp}/q_0& \q_+/q_0
\epm, \quad
\M_0=\bpm
0&0&-q_0\\
\q_+&\mathbf{0}&\mathbf{0}
\epm.
\ee
\item  $\M(x,t,z)$ satisfies the growth conditions near the branch points $\pm q_0$:
\be \label{E:gcc}
\begin{cases}
\M_1(x,t,z)=\mathcal{O}(z \mp q_0), & z \in \C_+ \to \pm q_0,\\
\M_3(x,t,z)=\mathcal{O}(z \mp q_0), & z \in \C_- \to \pm q_0.
\end{cases}
\ee
\item $\M$ satisfies the symmetries
\begin{equation}\label{E:RHP11}
\M(x,t,z)=\M(x,t,\hat{z}) \bPi(z), \quad
(\M^{-1})^{\top}(x,t,z)=-\frac{1}{\gamma(z)}\bJ \M^*(x,t,z^*) \mathbf{\Gamma}(z).
\end{equation}

\item  The following residue conditions hold at  each point $\zeta_j$, $j=0,...,N-1$:
\begin{equation}\label{E:mlstjzc}
\mathrm{Res}_{z =\zeta_j}\M
= \lim_{z\to \zeta_j}\M \bpm 0 & 0 &0\\
0&0& 0 \\ \tau_{j} \e^{\theta_{31}(x,t,\zeta_j)}&0&0 \epm,
\end{equation}
where $\theta_{mn}(x,t,z):=\theta_m(x,t,z)-\theta_n(x,t,z)$ for $1 \leq m \leq 3$ and $1 \leq n \leq 3$.
\ei
\end{RHP}

\begin{remark}
The jump matrix used here is obtained by simplifying the jump matrix in
Ref.~\cite{BD2015-1}; see Eq.~(3.2) therein. In Appendix~\ref{Apppof258}, we show how to
reduce the jump matrix in Ref.~\cite{BD2015-1} to the form given in~\eqref{E:Vex}.
\end{remark}

For the self-consistency of this paper, we establish the uniqueness of the solution to the above RH problem.
\begin{lemma}\label{L:wyxRH1}
The solution of RH problem~\ref{RHP:zc} is unique, if it exists. Moreover, for any solution $\M$, the determinant satisfies $\det \M =\gamma(z)$.
\end{lemma}
\begin{proof}
See Appendix~\ref{AppBproof}.
\end{proof}
Therefore, $\M(x,t,z)$ defined by~\eqref{E:exM} is the unique solution to RH problem~\ref{RHP:zc}. Combined with the asymptotic behavior of $\M(x,t,z)$ as $z \to \infty$ (see~\cite[Corollary 2.29]{BD2015-1}), we summarize the following reconstruction theorem:
\begin{theorem}\label{Th:cggs}
Suppose that there exists a sufficiently smooth  solution $\q(x,t)$ to the defocusing Manakov system~\eqref{E:demanakovS} with  NZBCs \eqref{E:bjtj}, and its initial data satisfies  Assumptions~\ref{As:1}. Then RH problem~\ref{RHP:zc} has a unique solution $\M(x,t,z)$. Moreover, $\q(x,t)$ can be reconstructed from $\M$ as follows:
\be \label{E:cggs}
\q(x,t)=-\ii \lim_{z \to \infty}z \left(\M_{21}(x,t,z), \M_{31}(x,t,z) \right)^{\top},
\ee
 where  $\M_{ij}$ denotes the $(i,j)$-entry of the matrix-valued function $\M$.
\end{theorem}

\begin{figure}[h]
\centering

\begin{subfigure}{0.31\textwidth}
\centering
\begin{tikzpicture}
\node at (0,0) {\includegraphics[width=4.2cm]{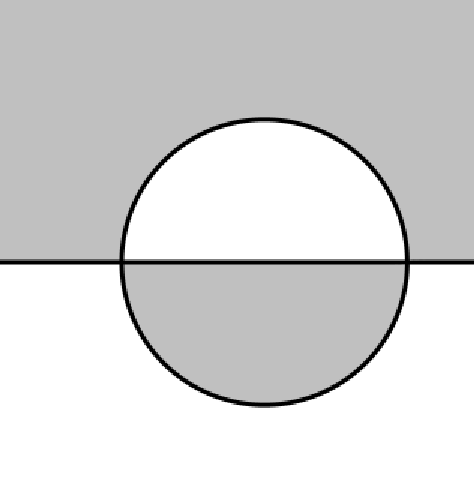}};
\draw[dashed,very thick](-1,2.2)--(-1,-2.2);
\node at (-1.2,-0.3) {\small $0$};
\node at (1.7,-0.3) {\small $z_0$};
\node at (-0.15,-0.35) {\small $z_1$};
\filldraw[black](-0.15,-0.125) circle(1pt);
\end{tikzpicture}
\caption{}
\label{fig:signature-a}
\end{subfigure}
\hfill
\begin{subfigure}{0.31\textwidth}
\centering
\begin{tikzpicture}
\node at (0,0) {\includegraphics[width=4.5cm]{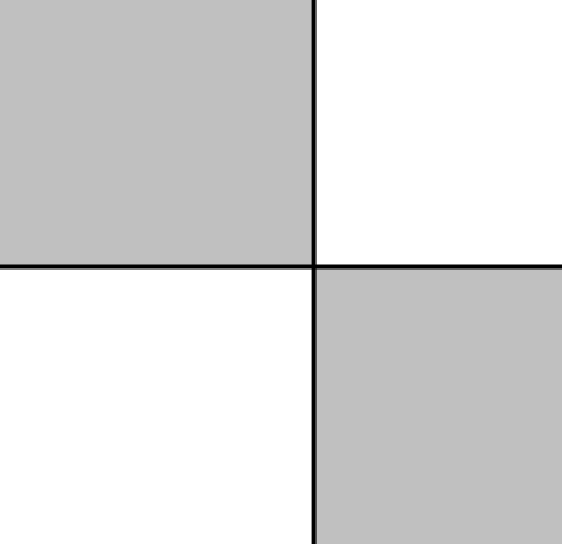}};
\draw[dashed,very thick](-0.3,2.2)--(-0.3,-2.2);
\node at (-0.45,-0.2) {\small $0$};
\node at (0.5,-0.2) {\small $z_1$};
\node at (1.9,-0.2) {\small $z_0$};
\filldraw[black](1.9,0.04) circle(1pt);
\end{tikzpicture}
\caption{}
\label{fig:signature-b}
\end{subfigure}
\hfill
\begin{subfigure}{0.31\textwidth}
\centering
\begin{tikzpicture}
\node at (0,0) {\includegraphics[width=4cm]{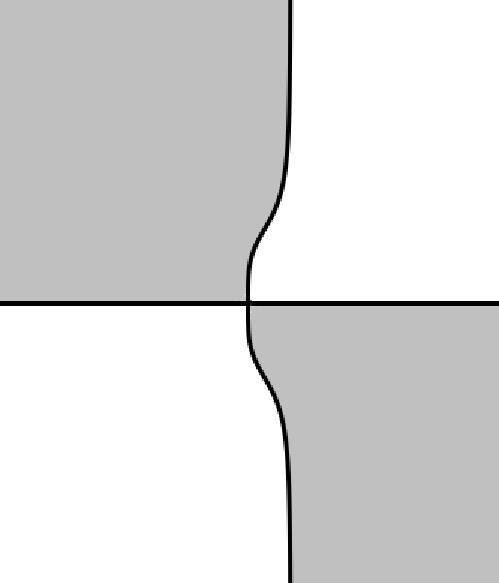}};
\node at (-0.2,-0.4) {\small $0$};
\draw[dashed,very thick](-0.05,2.3)--(-0.05,-2.3);
\end{tikzpicture}
\caption{}
\label{fig:signature-c}
\end{subfigure}

\vspace{0.3cm}

\begin{subfigure}{0.31\textwidth}
\centering
\begin{tikzpicture}
\node at (0,0) {\includegraphics[width=4.8cm]{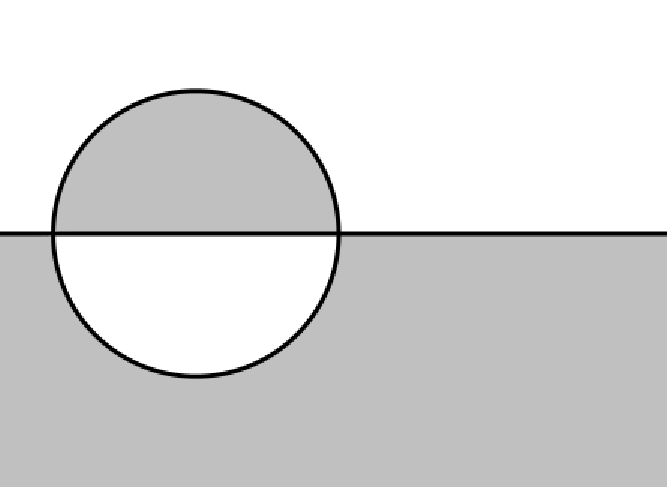}};
\draw[dashed,very thick](0.08,2)--(0.08,-1.8);
\node at (0.25,-0.2) {\small $0$};
\node at (-2.2,-0.2) {\small $z_0$};
\filldraw[black](-0.5,0.06) circle(1pt);
\node at (-0.5,-0.2) {\small $z_1$};
\end{tikzpicture}
\caption{}
\label{fig:signature-d}
\end{subfigure}
\hfill
\begin{subfigure}{0.31\textwidth}
\centering
\begin{tikzpicture}
\node at (0,0) {\includegraphics[width=4.2cm]{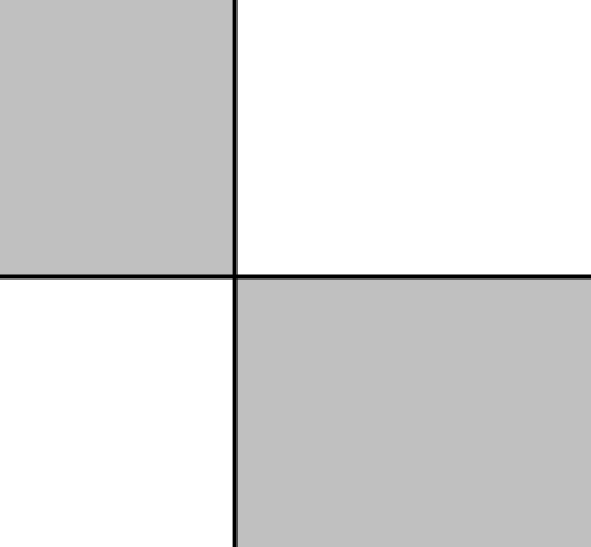}};
\draw[dashed,very thick](0.2,2)--(0.2,-2);
\node at (-0.6,-0.2) {\small $z_1$};
\node at (-1.8,-0.25) {\small $z_0$};
\node at (0.4,-0.2) {\small $0$};
\filldraw[black](-1.8,-0.01) circle(1pt);
\end{tikzpicture}
\caption{}
\label{fig:signature-e}
\end{subfigure}
\hfill
\begin{subfigure}{0.31\textwidth}
\centering
\begin{tikzpicture}
\node at (0,0) {\includegraphics[width=4cm]{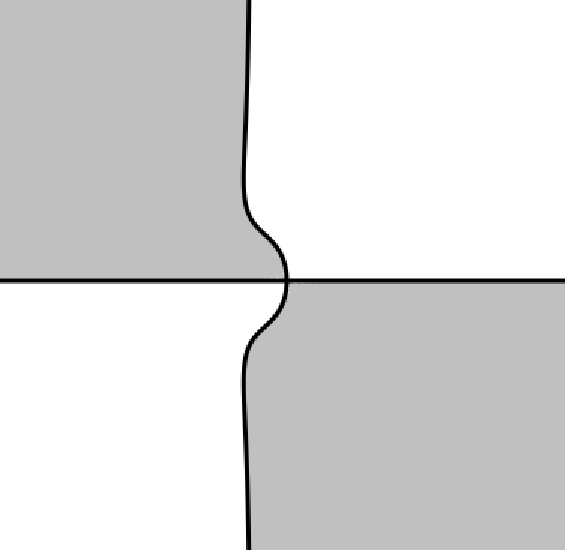}};
\node at (0.2,-0.2) {\small $0$};
\draw[dashed,very thick](0.05,2)--(0.05,-2);
\end{tikzpicture}
\caption{}
\label{fig:signature-f}
\end{subfigure}

\caption{The signature tables for $\phi_{32}$, $\phi_{21}$ and $\phi_{31}$. 
Panels (a)--(c) correspond to $\xi=0.5$ and $q_0=1$, while panels (d)--(f) correspond to $\xi=-0.5$ and $q_0=1$. 
The grey regions correspond to $\{z:\mathrm{Re}\,\phi_{ij}<0\}$ and the white regions correspond to $\{z:\mathrm{Re}\,\phi_{ij}>0\}$.}
\label{fig:signature table}
\end{figure}

\section{Long-time asymptotics}\label{se:3}
In this section, we will investigate asymptotics of the RH problem \ref{RHP:zc} as $t \to \infty$ by using Deift–Zhou nonlinear steepest
 descent method~\cite{pz93}. The main idea is to convert the original RH
 problem into a small-norm one via a series of explicit and invertible transformations. 

We define the phase functions $\phi_{i j}(\xi,z)$ by $\theta_{i j}(x,t,z)=t\phi_{i j}(\xi,z)$, $1\leq i,j \leq 3$, which originate from the oscillatory exponential terms in the jump matrix $\V$. Let's  consider the following three phase functions:
\be
\phi_{32}(\xi,z)=-2 \ii \xi \frac{q_0^2}{z}+\ii \frac{q_0^4}{z^2}, \quad
\phi_{21}(\xi,z)=2 \ii \xi z-\ii z^2, \quad
\phi_{31}(\xi,z)=4 \ii \xi \lambda(z)-4 \ii k(z)\lambda(z).
\ee
Set $z_1=\xi$ and $z_0=\frac{q_0^2}{\xi}$, then $z_0$ is the unique saddle point of $\phi_{32}(\xi,k)$, $z_1$ is the unique saddle point of $\phi_{21}(\xi,k)$, and the function $\phi_{31}(\xi,k)$ has no saddle point on $\R$.  The signature tables for $\Re \phi_{21}(\xi,k)$, $\Re \phi_{32}(\xi,k)$ and  $\Re \phi_{31}(\xi,k)$ are shown in Fig.~\ref{fig:signature table}.

By observing the signature tables of these phase functions, we note that the Deift--Zhou steepest descent analysis for $\xi>0$ and $\xi<0$ is quite similar. Indeed, when $\xi\in \mathcal I_+$, we need to find suitable factorizations, guided by the signature tables, on the intervals $(-\infty,0)$, $(0,z_1)$, $(z_1,z_0)$, and $(z_0,+\infty)$, and then extend the corresponding factors to the lens regions. When $\xi\in \mathcal I_-$, it suffices to carry out the same procedure on the intervals $(-\infty,z_0)$, $(z_0,z_1)$, $(z_1,0)$, and $(0,+\infty)$. The ideas used to find these factorizations are essentially the same, and the outer and local parametrices required in the final analysis also have analogous forms. Therefore, for simplicity, we only present the proofs of Theorems~\ref{asy-th-1} and~\ref{Th:2m} in the case $\xi\in \mathcal I_+$; the computations for $\xi\in \mathcal I_-$ are completely analogous. From now on, we assume that $\xi\in \mathcal I_+$.

\subsection{First transformation: conjugation and factorization}
The first transformation is introduced in preparation for the lens-opening
procedure in the next subsection. Its purpose is to put the jump matrix into
forms that admit suitable factorizations along the entire real axis. Here
``suitable'' means that, after the lenses are opened, the resulting jump
matrices are exponentially decaying as $t\to\infty$, except in neighborhoods of
certain critical points.
We now show how to factorize the jump matrix on the four intervals
$(-\infty,z_0)$, $(z_0,z_1)$, $(z_1,0)$, and $(0,+\infty)$. 
Throughout the paper, we sometimes omit the dependence on $x$, $t$, $z$  for convenience.

First,
a direct computation gives $\V= \V^U \V^L$, where
\begin{equation}\label{E:sjfjV}
\V^U=
\begin{pmatrix}
1  & -\frac{1}{\gamma} r_1^*(z^*) \e^{\theta_{12}} & -r_2^* (z^*) \e^{\theta}_{13}\\
0 & 1 & -r_3(z) \e^{\theta_{23}}\\
0 & 0 &1
\end{pmatrix}, \quad
\V^L=
\begin{pmatrix}
1 & 0 & 0\\
r_1(z) \e^{\theta_{21}}&1  &0 \\
r_2(z) \e^{\theta_{31} }& -\frac{1}{\gamma}r_3^*(z^*) \e^{\theta_{32}} &1
\end{pmatrix}.
\end{equation}
In view of the signature table of the phase functions
$\phi_{ij}(\xi,z)$; see Fig.~\ref{fig:signature table}, this factorization is
well adapted to the interval $(-\infty,0)$. However, on the positive real axis, a new factorization has to be found.
The first step is to derive an $\mathbf L\mathbf D\mathbf U$ factorization
of $\V$, namely a decomposition of $\V$ into lower-triangular, diagonal,
and upper-triangular factors.
Through the application of  the
standard Gaussian-elimination procedure, together with the symmetry properties
of the scattering matrices $\A(z)$ and $\B(z)$, a direct computation shows that 
\be \label{E:LDUf}
\begin{aligned}
\V=&
{\scriptsize
\begin{pmatrix}
1&0&0\\
-\hat{r}_1(z)\dfrac{s_{2-}(z)}{s_{1-}(z)} \e^{\theta_{21}}&1&0\\
\hat{ r}_2(z)\dfrac{s_{3-}(z)}{s_{1-}(z)} \e^{\theta_{31}}
&
\frac{1}{\gamma(z)}\hat{r}_3^*(z)\dfrac{s_{3-}(z)}{s_{2-}(z)} \e^{\theta_{32}}
&1
\end{pmatrix}
}
{\scriptsize
\begin{pmatrix}
\dfrac{1}{a_{11}(z)b_{11}(z)} & & \\
 &\dfrac{a_{11}(z)b_{11}(z)}{a_{33}(z)b_{33}(z)}& \\
 & &a_{33}(z)b_{33}(z)
\end{pmatrix}
}\\
&\times
{\scriptsize
\begin{pmatrix}
1&
\dfrac{1}{\gamma(z)}\hat {r}_1^*(z)\dfrac{s_{1+}(z)}{s_{2+}(z)} \e^{\theta_{12}}
&
-\hat {r}_2^*(z)\dfrac{s_{1+}(z)}{s_{3+}(z)} \e^{\theta_{13}}
\\
0&1&
\hat {r}_3(z)\dfrac{s_{2+}(z)}{s_{3+}(z)} \e^{\theta_{23}}
\\
0&0&1
\end{pmatrix}
}, \qquad z \in \R,
\end{aligned}
\ee
where
\begin{subequations}\label{E:hatr123}
\begin{align}
\hat{r}_1(z)&=\frac{b_{21}(z)}{b_{11}(z)}, \qquad
\hat{r}_2(z)=\frac{a_{31}(z)}{a_{33}(z)}, \qquad
\hat{r}_3(z)=\frac{b_{23}(z)}{b_{33}(z)},
\label{E:hatr123a}\\
s_{1}(z)&=
\begin{cases}
a_{11}(z),& z\in\mathbb C_+,\\[1mm]
\dfrac{1}{b_{11}(z)},& z\in\mathbb C_-,
\end{cases}
\qquad
s_3(z)=s_1(\hat z),\qquad
s_2(z)=\frac{1}{s_1(z)s_3(z)}.
\label{E:hatr123b}
\end{align}
\end{subequations}
Here the notation $s_{j\pm}$, $j=1,2,3$, refers to the boundary values of $s_j$
taken from the upper and lower half-planes, respectively.
Then we remove the diagonal factor on the positive real axis.  To this end, we introduce
a diagonal matrix $\Del(z)$ satisfying
\be \label{E:deltaju}
\begin{aligned}
\Del_+(z)=\Del_-(z)
\bpm
a_{11}(z) b_{11}(z)& & \\
 & \frac{a_{33}(z) b_{33}(z)}{a_{11}(z) b_{11}(z)} & \\
 & & \frac{1}{a_{33}(z) b_{33}(z)}
\epm, \qquad \  z \in \R_+  \setminus \{q_0\}.
\end{aligned}
\ee
A natural choice satisfying the above jump condition is
\be \label{E:Deldef}
\Del(z)= \bpm
\delta_1(z) & & \\
 &\frac{1}{\delta_1(z) \delta_1(\hat{z})} & \\
 & & \delta_1(\hat{z})
\epm,
\ee
where
\be \label{E:soldel1}
\delta_1(z)=\mathrm{exp}\bigg \{-\frac{1}{2 \pi \ii} \int_{\R_+} \frac{\log \left(1-\frac{1}{\gamma(s)}|r_1(s)|^2-|r_2(s)|^2\right)}{s-z} \mathrm  d s \bigg\}, \qquad z\in \C \setminus \R_+.
\ee
The jump relation~\eqref{E:hgdsb} for $\delta_1(z)$, together with the  symmetry relations
$$
a_{33}(z)=a_{11}(\hat z),\qquad b_{33}(z)=b_{11}(\hat z),
$$
implies that $\Del(z)$ defined by~\eqref{E:Deldef} satisfies the jump condition \eqref{E:deltaju}.

The following lemma collects the basic properties of the function $\delta_1(z)$.
\begin{lemma} \label{L:d1xz}
The function $\delta_1(z)$ has the following properties:
\bi
\item[\rm(i)]
The boundary values \(\delta_{1\pm}(z)\) exist and are continuous for
\(z\in\mathbb R_+\setminus\{q_0\}\). Moreover,
\begin{equation} \label{E:hgdsb}
\delta_{1+}(z)=\delta_{1-}(z)a_{11}(z)b_{11}(z),
\qquad z\in\mathbb R_+\setminus\{q_0\}.
\end{equation}

 \item[$\mathrm{(ii)}$]   As $z \to 0$ with $z \in \C \setminus \R_+$, one has
\be \label{E:delta0jx}
\delta_1(z) \to \delta_1(0)=  \mathrm{exp}\bigg \{-\frac{1}{2 \pi \ii} \int_{\R_+} \frac{\log \left(1-\frac{1}{\gamma(s)}|r_1(s)|^2-|r_2(s)|^2\right) }{s} \mathrm  d s \bigg\}.
\ee

\item[$\mathrm{(iii)}$]  As $z \to \infty$ with $z \in \C \setminus \R_+$, one has  $\delta_1(z) \to 1$.

\item[\rm(iv)]
The quotient $\delta_1(z)/s_1(z)$
admits a continuous extension to $\mathbb R_+$. 
Indeed, we have the relation
$
\delta_1(z)=s_1(z)\rho(z),
$
where
\be \label{E:deofrho}
\rho(z)=\prod_{j=0}^{N-1}\frac{z-\zeta_j^*}{z-\zeta_j}\mathrm{exp}\bigg\{\frac{1}{2\pi \ii}  \int_{\R_-} \frac{\log \left(1-\frac{1}{\gamma(s)}|r_1(s)|^2-|r_2(s)|^2\right)}{s-z}\mathrm  d s \bigg\}.
\ee

\ei
\end{lemma}
\begin{proof}
Using the symmetry satisfied by $\A(z)$ and $\B(z)$~(see~\eqref{E:ABdc} and~\eqref{E:ABrel}), one can easily verify that
$$
1-\frac{1}{\gamma(z)}|r_1(z)|^2-|r_2(z)|^2 = \frac{1}{a_{11}(z) b_{11}(z)}, \quad z \in \R_+.
$$
Then assertion~$\mathrm{(i)}$ follows from the Plemelj formula. 

Let $f(s):=\log \left(1-\frac{1}{\gamma(s)}|r_1(s)|^2-|r_2(s)|^2\right)$. Then
$$
f(s)=\mathcal O(s^{-2}),\qquad s\to+\infty,
\qquad
f(s)=\mathcal O(s),\qquad s\to0.
$$
Moreover, $f(s)$  is real-valued and smooth on the positive real axis. Hence
assertion $\mathrm{(ii)}$ follows from the standard analysis of the endpoint behavior of Cauchy-type integrals; see, for example~\cite[Lemma 2.11]{TS2016}.  Assertion~$\mathrm{(iii)}$ can be proved analogously.

It remains to prove \({\rm(iv)}\). We use the trace formula
\cite[proof of Lemma~3.9]{BD2015-1}:
\be\label{E:tafo}
s_{1}(z)=\prod_{j=0}^{N-1}  \frac{z-\zeta_j}{z-\zeta_j^*} \mathrm{exp}\bigg \{-\frac{1}{2 \pi \ii} \int_{\R} \frac{\log \left(1-\frac{1}{\gamma(s)}|r_1(s)|^2-|r_2(s)|^2\right)}{s-z} \mathrm{d}s \bigg\}, \qquad z \in \C \setminus \R.
\ee
Then assertion~$\mathrm{(iv)}$ follows directly from this representation.
\end{proof}

For convenience, we denote 
$$
\widehat \V(x,t,z)=\Del_-^{-1}(z) \V (x,t,z)\Del_+(z).
$$
In view of the
factorization \eqref{E:LDUf} and the relation
\(\delta_1(z)=s_1(z)\rho(z)\),  we obtain $\widehat \V=\V^L_{(1)} \V^U_{(1)} $ for $z \in \R_+$, where
\be \label{E:jumpM1ex}
\V^L_{(1)} =
\begin{pmatrix}
 1 & 0 &0 \\
 -\tilde{r}_1(z) \e^{\theta_{21}} &1  &0 \\
 \tilde{r}_2(z) \e^{\theta_{31}} & \frac{1}{\gamma}\tilde{r}_3^*(z^*) \e^{\theta_{32}} &1
\end{pmatrix}, \quad
\V^U_{(1)}=
\begin{pmatrix}
1  &\frac{1}{\gamma}\tilde{r}_1^*(z^*) \e^{\theta_{12}} & -\tilde{r}_2^*(z^*) \e^{\theta_{13}}\\
0  & 1 & \tilde{r}_3(z) \e^{\theta_{23}}\\
  0&0  &1
\end{pmatrix},
\ee
with
\be \label{E:tilder123}
\tilde{r}_1(z)=\rho^2(z)\rho(\hat{z})\hat{r}_1(z), \qquad
\tilde{r}_2(z)=\frac{\rho(z)}{\rho(\hat{z})}\hat{r}_2(z), \qquad
\tilde{r}_3(z)=\rho(z)\rho^2(\hat{z})\hat{r}_3(z).
\ee
In the above derivation, we have used the symmetry relations
$\rho^*(z^*)=\frac{1}{\rho(z)}$ and $\rho^*(\hat{z}^*)=\frac{1}{\rho(\hat{z})}$.
Since $|\rho(z)|=|\rho(\hat z)|=1$ for $z\in(0,+\infty)$, we have
$$
|\tilde r_1(z)|
=
|\hat r_1(z)|,\qquad 
|\tilde r_3(z)|
=
|\hat r_3(z)|,
\qquad z\in (0,+\infty).
$$
After examining the sign structure of the phase functions $\phi_{ij}(\xi,z)$, we observe that the  factorization for $\widehat \V$ is suitable on the interval $(z_1,z_0)$. On the intervals $(0,z_1)$ and $(z_0,+\infty)$, however, further modifications are required. To this end, we introduce two permutation matrices $\mathcal A$ and $\mathcal B$ to rearrange the entries of the jump matrix and thereby obtain factorizations compatible with the desired exponential decay.  More precisely, we set
\begin{equation}\label{E:permutationAB}
\mathcal A=
\begin{pmatrix}
1&0&0\\
0&0&1\\
0&1&0
\end{pmatrix},
\qquad
\mathcal B=
\begin{pmatrix}
0&1&0\\
1&0&0\\
0&0&1
\end{pmatrix}
.
\end{equation}
Then, by applying the standard Gaussian-elimination procedure, we obtain
\begin{align}
\mathcal A\widehat \V(z)\mathcal A
&=\V_{(2)}^L(z)\V_{(2)}^D(z) \V_{(2)}^U(z),
\qquad z\in (z_0,+\infty),\label{E:permuted-factorization-2}\\
\mathcal B\widehat \V(z)\mathcal B
&=\V_{(3)}^L(z) \V_{(3)}^D(z) \V_{(3)}^U(z),
\qquad z\in (0,z_1),\label{E:permuted-factorization-3}
\end{align}
where
\begin{align}
\label{E:V2-factors}
\V_{(2)}^L(z)
&=
{\scriptsize\begin{pmatrix}
1&0&0\\
\tilde r_2(z)  \e^{\theta_{31}}&1&0\\
-\tilde r_1(z) \e^{\theta_{21}}&
\frac{\tilde r_3(z)}
{1+\frac{1}{\gamma} \tilde r_3(z) \tilde r_3^*(z^*) }
\e^{\theta_{23}}&1
\end{pmatrix}}, \quad
\V_{(2)}^D(z)
=
{\scriptsize\begin{pmatrix}
1&0&0\\
0&1+\frac{1}{\gamma}|\tilde r_3(z)|^2&0\\
0&0&\frac{1}{1+\frac{1}{\gamma}|\tilde r_3(z)|^2}
\end{pmatrix}},
\notag\\[1mm]
\V_{(2)}^U(z)
&=
{\scriptsize\begin{pmatrix}
1&- \tilde r_2^*(z^*) \e^{\theta_{13}}
&\dfrac{1}{\gamma} \tilde r_1^* (z^*) \e^{\theta_{12}}\\
0&1&
\frac{\frac{1}{\gamma}\tilde  r_3^*(z^*)}
{1+\frac{1}{\gamma} \tilde r_3(z) \tilde r_3^*(z^*)}
\e^{\theta_{32}}\\
0&0&1
\end{pmatrix}}, \quad
\V_{(3)}^D(z)
=
{\scriptsize \begin{pmatrix}
1-\frac{1}{\gamma}|\tilde r_1(z)|^2&0&0\\
0&\frac{1}{1-\frac{1}{\gamma}|\tilde r_1(z)|^2}&0\\
0&0&1
\end{pmatrix}},
\notag\\[1mm]
\V_{(3)}^L(z)
&=
\begin{pmatrix}
1&0&0\\
\frac{\frac{1}{\gamma} \tilde  r_1^*(z^*)}
{1-\frac{1}{\gamma} \tilde r_1(z) \tilde r_1^*(z^*)}
\e^{\theta_{12}}&1&0\\
\frac{\frac{1}{\gamma}(\tilde r_3^*(z^*)+ \tilde r_1^*(z^*) \tilde r_2(z))}
{1-\frac{1}{\gamma}\tilde r_1(z) \tilde r_1^*(z^*)}
\e^{\theta_{32}}
&
\bigl(\tilde r_2(z)+\frac{1}{\gamma} \tilde r_1(z) \tilde r_3^*(z^*)\bigr)
\e^{\theta_{31}}&1
\end{pmatrix},
\notag\\[1mm]
\V_{(3)}^U(z)
&=
\begin{pmatrix}
1&
-\frac{\tilde r_1(z)}
{1-\frac{1}{\gamma}\tilde r_1(z) \tilde r_1^*(z^*)}
\e^{\theta_{21}}
&
\frac{\tilde r_3(z)+\tilde r_1(z) \tilde r_2^*(z^*)}
{1-\frac{1}{\gamma}\tilde r_1(z) \tilde r_1^*(z^*)}
\e^{\theta_{23}}\\
0&1&
\bigl(-\tilde r_2^*(z^*)-\frac{1}{\gamma}\tilde r_1^*(z^*) \tilde r_3(z)\bigr)
\e^{\theta_{13}}\\
0&0&1
\end{pmatrix}.
\end{align}
Thus the function $\widehat \V $ admits the following factorizations on $(z_{0},+\infty)$
 and $(0,z_{1})$:
$$
 \widehat \V =\begin{cases}
\left( \mathcal A \V^L_{(2)}\mathcal A \right)  \left( \mathcal A \V^D_{(2)}\mathcal A \right)  \left( \mathcal A \V^U_{(2)}\mathcal A \right), & z \in (z_0,+\infty),\\
\left( \mathcal B \V^L_{(3)}\mathcal B \right)  \left( \mathcal B \V^D_{(3)}\mathcal B \right)  \left( \mathcal B \V^U_{(3)}\mathcal B \right), & z \in (0,z_1).
\end{cases} 
$$
In the next step we introduce a diagonal transformation $\tilde \Del(z)$ to remove the diagonal factors $\left( \mathcal A \V^D_{(2)}\mathcal A \right)$ and $\left( \mathcal B \V^D_{(3)}\mathcal B \right)$.
The matrix-valued function $\tilde \Del(z)$ is defined by
$$
\tilde{\Del}(z)=\begin{pmatrix}
 1/\delta( \hat z)&  & \\
  &  \delta(z) \delta(\hat z)& \\
  &  &1/\delta(z)
\end{pmatrix},
$$
where $\delta(z)$ is determined by the following scalar RH problem:
\be \label{E:delta}
\begin{cases}
\delta_+(z)=\delta_-(z) \left(1+\frac{1}{\gamma(z)}|\tilde{r}_3(z)|^2\right), & z \in (z_0, +\infty),\\
\delta(z) \to 1, & z \to \infty.
\end{cases}
\ee
Since $\gamma(z)>0$ on $(z_0,+\infty)$, there exists a constant
$C>1$ such that
$$
1<
1+\frac{1}{\gamma(z)}|\tilde r_3(z)|^2
<C,
\qquad
z\in(z_0,+\infty),\quad \xi\in \mathcal I_+ .
$$ 
Therefore, the above RH problem is uniquely solvable. By the Plemelj's formula, it can be expressed as
\be \label{E:Pbsdel}
\delta(z)= \mathrm{exp} \bigg\{ \frac{1}{2 \pi \ii}  \int_{z_0}^{+\infty} \frac{\log (1+\frac{1}{\gamma(s)}|\tilde{r}_3(s)|^2)}{s-z} \mathrm{d}s\bigg\}.
\ee
Let $\tilde{\delta}(z)=\delta(\hat{z})$. Using the relation
$$ |\rho(z)|=1, \qquad   -\frac{1}{\gamma(\hat{z})}  |\hat{r}_1(\hat{z})|^2= \frac{1}{\gamma(z)} |\hat{r}_3(z)|^2, \quad \ z \in (0,+\infty),$$
one can find that $\tdel(z)$ satisfies the following RH problem
\be \label{E:tildelta}
\begin{cases}
\tdel_+(z)=\tdel_-(z) \left(1-\frac{1}{\gamma(z)}|\tilde{r}_1(z)|^2\right)^{-1}, & z \in (0,z_1),\\
\tilde{\delta}(z) \to \delta(0), & z \to \infty.
\end{cases}
\ee
Then $\tdel(z)$ can be expressed as
\be \label{E:tildeldel}
\tdel(z)= \delta(0) \mathrm{exp} \bigg\{ \frac{-1}{2 \pi \ii}  \int_{0}^{z_1} \frac{\log (1-\frac{1}{\gamma(s)}|\tilde{r}_1(s)|^2)}{s-z} \mathrm{d}s\bigg\}.
\ee
Therefore, by using the jump relations in \eqref{E:delta} and \eqref{E:tildelta}, it is straightforward to verify that,  on the intervals \((z_0,+\infty)\) and  \((0,z_1)\), respectively, the diagonal factors $\bigl( \mathcal A \V^D_{(2)}\mathcal A \bigr)$ and  $\bigl( \mathcal B \V^D_{(3)}\mathcal B \bigr)$
disappear from $\tilde \Del_-^{-1} \widehat \V  \tilde \Del_+$. Indeed, on the intervals
$(z_{0},+\infty)$   and $(0,z_{1})$, the function  $\tilde \Del_-^{-1} \widehat \V  \tilde \Del_+$ admits
factorizations of the following form:
$$
\tilde \Del_-^{-1} \widehat \V  \tilde \Del_+ =
\begin{cases}
\left[ \tilde \Del_-^{-1}  \bigl(\mathcal A \V^L_{(2)}\mathcal A \bigr)   \tilde \Del_- \right]  \left[ \tilde \Del_+^{-1}  \bigl( \mathcal A \V^U_{(2)}\mathcal A \bigr)   \tilde \Del_+  \right], & z \in (z_0,+\infty),\\
\left[ \tilde \Del_-^{-1}  \bigl(\mathcal B \V^L_{(3)}\mathcal B \bigr)   \tilde \Del_- \right]  \left[ \tilde \Del_+^{-1}  \bigl( \mathcal B \V^U_{(3)}\mathcal B \bigr)   \tilde \Del_+  \right], & z \in (0,z_1).
\end{cases}
$$
where $\tilde \Del_+^{-1}  \bigl( \mathcal A \V^U_{(2)}\mathcal A \bigr)  \tilde \Del_+$ and $\tilde \Del_+^{-1}  \bigl( \mathcal B \V^U_{(3)}\mathcal B \bigr) \tilde \Del_+$ decay, respectively, on the upper sides of these
two intervals, while $\tilde \Del_-^{-1}  \bigl( \mathcal A \V^L_{(2)}\mathcal A \bigr)   \tilde \Del_- $ and $ \tilde \Del_-^{-1}  \bigl( \mathcal B \V^L_{(3)}\mathcal B \bigr)   \tilde \Del_-$ decay, respectively, on the lower
sides. Thus, the introduction
of the transformation matrices $\Del(z)$ and $\tilde \Del(z)$ already
allows us to obtain suitable factorizations on the whole real axis.

The following lemma collect the basic properties of the functions $\delta(z)$ and $\tdel(z)$, which
will be used in the subsequent computations.
\begin{lemma}\label{L:13}
The functions $\delta(z)$ and $\tdel(z)$ have the following properties:
\begin{enumerate}
\item $\delta(z)$ and $\tdel(z)$ can be written as
\begin{align}\label{E:expdelta1sj}
\delta(z) = \e^{-i  \nu \log(z_{0}-z)}\e^{-\chi(z)}, \qquad
\tdel(z)=\delta(0) \e^{-i \tilde{\nu} \log \left(z-z_1\right)} \e^{\tilde{\chi}(z)},
\end{align}
where $\nu$, $\tilde{\nu}$, $\chi(z)$ and $\tilde{\chi}(z)$ are defined by
\begin{align*}
&\nu = - \frac{1}{2\pi}\log(1+\frac{1}{\gamma(z_0)}| \hat{r}_3(z_{0})|^{2}), \qquad
\tilde{\nu}=-\frac{1}{2 \pi} \log \bigl( 1-\frac{1}{\gamma(z_1)}\left| \hat{r}_1 \left(z_1\right)\right|^2 \bigr)=\nu,
\end{align*}
and
\be \label{E:L13-st-2}
\begin{aligned}
& \chi(z) = \frac{1}{2\pi \ii} \int_{z_{0}}^{\infty}  \log(\zeta-z) \mathrm{d} \log(1+\frac{1}{\gamma(\zeta)}|\hat{r}_3(\zeta)|^{2}), \\
 &\tilde{\chi}(z)=\chi(0)-\chi(\hat{z})+\ii \nu \log (z).
\end{aligned}
\ee

\item For each $\xi \in \mathcal{I}_+$, $\delta(z)$ and $\tdel(z)$ obey the symmetries
\begin{align}\label{E:L13-st-4}
\delta(z) = (\delta^*(z^*))^{-1},\quad  z \in \C \setminus (z_0,+\infty);\qquad
\tdel(z)=(\tdel^*(z^*))^{-1},\quad  z \in \C \setminus (0,z_1).
\end{align}
\item As $z \to z_{0}$ and $z \to z_1$ along the paths which are nontangential to $(z_0,+\infty)$ and $(0,z_1)$, we have

\begin{align}
& |\chi(z)-\chi (z_0)| \leq C |z-z_{0}|(1+|\log|z-z_{0}||),\label{E:L13-st-5} \\
& |\tilde{\chi}(z)-\tilde{\chi}(z_1)|  \leq C|z-z_1| (1+|\log | z-z_1| ),
\end{align}
where $C$ is independent of $\xi \in \mathcal{I}_+$.
\end{enumerate}
\end{lemma}
\begin{proof}
The lemma follows from~\eqref{E:Pbsdel},~\eqref{E:tildeldel} and relatively straightforward estimates.
\end{proof}

Next, we  further introduce a diagonal matrix $\bP(z)$, whose role is to regulate the behavior at the discrete spectrum.  The definition of $\bP(z)$ is inspired by the ideas used in Refs.~\cite{CuJe2016,BJM2018,CJ2024}. We partition the set $\{0,1,..,N-1 \}$ into the pair of sets
$$
\nabla^+=\big\{j:\Re \zeta_j > \xi \big\},\qquad  \nabla^-=\big\{j: \Re \zeta_j \leq \xi \big\}.
$$
We define 
\be \label{E:defiofP}
\bP(z)= \begin{pmatrix}
  \cP_1(z)&  & \\
  &  1/\left(\cP_1(z)\cP_1(\hat{z}) \right )& \\
  &  &\cP_1(\hat{z})
\end{pmatrix}, \qquad \cP_1(z)=\prod_{j \in \nabla^+}\frac{z-\zeta_j}{z-\zeta_j^*}.
\ee
Then the residue conditions after the transformation take the forms
\eqref{E:lstjN1+} and~\eqref{E:lstjN1-}, and hence are of the desired form.

We are now ready to define the first transformation. Set
$$
\widetilde \bT(z)=\Del(z) \tilde \Del(z) \bP(z)=:\mathrm{diag}\left( \widetilde T_1(z),\widetilde T_2(z),\widetilde T_3(z) \right).
$$
Then the first transformation is defined as follows:
\be \label{E:firsttrans}
\M^{(1)}(x,t,z)= \widetilde\bT(\infty)^{-1} \M_{\infty}^{-1}\M(x,t,z) \widetilde \bT (z), \qquad z \in \C.
\ee
A direct computation shows that
\be \label{E:Tinftybds}
\widetilde\bT(\infty)=
\bpm
1/\delta(0) & & \\
 &\delta(0)/(\delta_1(0)  \cP_1(0))  & \\
 & & \delta_1(0) \cP_1(0)
\epm.
\ee

Let us examine the properties of $\M^{(1)}(x,t,z)$. It is straightforward to verify that $\M^{(1)}(x,t,z) \to \bI, \quad z \to \infty$, 
and
$$
\M^{(1)}(x,t,z) \to \frac{\ii}{z}   \bpm 0&0&-q_0\\ 0&0&0\\  q_0 &0&0 \epm,  \qquad z \to 0.
$$
Moreover, the jump contour for $\M^{(1)}$ remains the real axis. The corresponding
jump matrix is denoted by $\V^{(1)}$. For brevity, set 
$$\bT(z)=\tilde \Del(z) \bP(z)=:\mathrm{diag}\left( T_1(z),T_2(z),T_3(z) \right).$$
 Based on the preceding analysis,
we summarize the factorizations of $\V^{(1)}$ as follows:
\be \label{E:V1dehsfj}
\V^{(1)}=\begin{cases}
\bigl(\widetilde\bT^{-1} \V^U  \widetilde\bT  \bigr) \bigl(\widetilde\bT^{-1} \V^L \widetilde\bT  \bigr),& z \in (-\infty,0),\\
\bigl(\bT^{-1} \V^L_{(1)} \bT  \bigr) \bigl( \bT^{-1} \V^U_{(1)} \bT  \bigr),& z \in (z_1,z_0),\\
\left[\bT_-^{-1} \bigl( \mathcal A \V^L_{(2)}\mathcal A \bigr) \bT_- \right]  \left[ \bT_+^{-1} \bigl( \mathcal A \V^U_{(2)}\mathcal A \bigr)   \bT_+ \right], & z \in (z_0,+\infty),\\
\left[ \bT_-^{-1}  \bigl( \mathcal B \V^L_{(3)}\mathcal B \bigr) \bT_- \right]  \left[ \bT_+^{-1}  \bigl( \mathcal B \V^U_{(3)}\mathcal B \bigr)  \bT_+ \right], & z \in (0,z_1),
\end{cases}
\ee
where $ \V^U$  and $ \V^L$ are given by \eqref{E:sjfjV}, $\V^L_{(1)}$ and $\V^U_{(1)}$ are given by \eqref{E:jumpM1ex}, and $\V^L_{(2)}$, $\V^U_{(2)}$, $\V^L_{(3)}$ and  $\V^U_{(3)}$ are given by
\eqref{E:V2-factors}.

We now examine the residue conditions satisfied by $\M^{(1)}(x,t,z)$.
\begin{lemma}\label{L:mlstj}
At each point $\zeta_j$, only one column of $\M^{(1)}(x,t,z)$ has a simple pole, while the other two columns are analytic.  Moreover, the following residue conditions hold:
\begin{itemize}
\item For $j \in \nabla^+$, we have
\begin{equation}\label{E:lstjN1+}
\mathrm{Res}_{z=\zeta_j}\M^{(1)}(x,t,z)=\lim_{z \to \zeta_j}\M^{(1)}(x,t,z) \bpm 0&0&\alpha_j\\ 0&0&0 \\ 0&0&0 \epm,
\end{equation}
where
\begin{align}\label{E:alphaj}
\alpha_j=\frac{\tau_j^{-1}\e^{-\theta_{31}(x,t,\zeta_j)}}{\widetilde T'_1(\zeta_j)(\widetilde T^{-1}_3)'(\zeta_j)}.
\end{align}
\item For $j \in \nabla^-$, we have
\begin{equation}\label{E:lstjN1-}
\mathrm{Res}_{z=\zeta_j}\M^{(1)}(x,t,z)=\lim_{z \to \zeta_j}\M^{(1)}(x,t,z) \bpm 0&0&0\\ 0&0&0 \\ \beta_j&0&0 \epm,
\end{equation}
where
\begin{align}\label{E:betaj}
\beta_j=\tau_j\frac{\widetilde T_1(\zeta_j)}{\widetilde T_3(\zeta_j)}\e^{\theta_{31}(x,t,\zeta_j)}.
\end{align}
\end{itemize}
\end{lemma}
\begin{proof}
Let  $\M^{(1)}_{i}$  denote the $i$-th column of the matrix  $\M^{(1)}$. When  $j \in \nabla^+$, it follows from~\eqref{E:mlstjzc}, \eqref{E:defiofP} and~\eqref{E:firsttrans} that the third column of  $\M^{(1)}$  has a simple pole, and we have
\begin{align*}
\mathrm{Res}_{z=\zeta_j} \M^{(1)}_3(x,t,z)&=\bigl(\lim_{z \to \zeta_j}(z-\zeta_j)\widetilde T_3(z)\bigr) \widetilde \bT(\infty)^{-1}\M_{\infty}^{-1} \M_3(x,t,\zeta_j)\\
&=\frac{\tau_j^{-1} \e^{-\theta_{31}(x,t,\zeta_j) } }{(\widetilde T_3^{-1})'(\zeta_j)}\widetilde \bT(\infty)^{-1}\M_{\infty}^{-1} \lim_{z \to \zeta_j}(z-\zeta_j)\M_1(x,t,z)\\
&=\frac{\tau_j^{-1} \e^{-\theta_{31}(x,t,\zeta_j) } }{(\widetilde T_3^{-1})'(\zeta_j)} \bigl( \lim_{z-\zeta_j}\frac{z-\zeta_j}{\widetilde T_1(z)}\bigr) \widetilde \bT(\infty)^{-1}\M_{\infty}^{-1} \lim_{z\to \zeta_j} \bigl(\M_1 \widetilde T_1(z) \bigr)\\
&=\frac{\tau_j^{-1}\e^{-\theta_{31}(x,t,\zeta_j)}}{\widetilde T'_1(\zeta_j)(\widetilde T^{-1}_3)'(\zeta_j)}\M^{(1)}_1(x,t,\zeta_j),
\end{align*}
where in the second equality we have used the equation  $\lim\limits_{z \to \zeta_j}(z-\zeta_j)\widetilde T_3(z)=\dfrac{1}{(\widetilde T_3^{-1})'(\zeta_j)}$. The proof for ~\eqref{E:lstjN1-} is similar.
\end{proof}

Finally, we analyze the behavior of $\M^{(1)}$ near the branch points. The
following lemma shows that, although $\Del(z)$ has a singularity at
$q_{0}$, the function $\M^{(1)}$ remains bounded in a neighborhood of this
point. 
\begin{lemma}\label{L:Mzd}
As $z \to q_0$ from the upper and lower half-planes, the limits
of  $\M^{(1)}(x,t,z)$ exist. We denote them by
$\M^{(1)}_{\pm}(x,t,q_0)$, respectively. Moreover, they satisfy
\be \label{E:M13req0}
\M^{(1)}_{+}(x,t,q_0)= \M^{(1)}_{-}(x,t,q_0) \bPi (q_0), \qquad \bPi(q_0)= \bpm
0&0&-\ii \\
0&1&0\\
\ii &0&0
\epm.
\ee
\end{lemma}
\begin{proof}
Recalling  the first transformation, we have
$$
 \M^{(1)}(x,t,z)=\widetilde \bT(\infty)^{-1} (\M_{\infty})^{-1}\M(x,t,z) \Del(z) \tilde \Del(z) \bP(z).
$$
Thus, it is enough to show that $\M(x,t,z)\Del(z)$ admits upper and lower boundary limits at $q_0$.
Assertion {\rm(iv)} in Lemma~\ref{L:d1xz} shows that
$\delta_1(z)$ has the same local behavior as $s_1(z)$ near $q_0$. For
instance, when $z$ approaches $q_0$ from the upper half-plane, both
$\delta_1(z)$ and $s_1(z)$ have a simple pole. When $z$ approaches
$q_0$ from the lower half-plane, $\delta_1(z)$ has a simple zero, whereas
$\delta_1(\hat {z})$ has a simple pole. This causes no difficulty, since
the prescribed growth condition of $\M(x,t,z)$ at $q_0$
(see~\eqref{E:gcc}) precisely compensates for these singularities. For example,
as $\C_+ \ni z\to q_0$, one has
$
\M_1(x,t,z)\delta_1(z)=\mathcal O(1),
$
which guarantees the existence of the first column of $\M^{(1)}_{+}(x,t,q_0)$. 
Since the $(2,2)$- and $(3,3)$-entries of $\Del(z)$ have finite limits as
 $\C_+ \ni z\to q_0$, the second and third columns of
$\M^{(1)}_{+}$ are also well defined. A similar argument shows that the limit of $\M^{(1)}(x,t,z)$ exists as $\C_- \ni z \to q_0$.

Next, we proceed to prove~\eqref{E:M13req0}. By a direct computation, one verifies that $\widetilde \bT(z)$ satisfies
$\widetilde \bT(z)=\bPi^{-1}(z) \widetilde \bT(\hat z) \bPi(z) $
which implies that $\M^{(1)}$ obeys the symmetry relation
$
\M^{(1)}(x,t,z) =\M^{(1)}(x,t,\hat {z}) \bPi(z).
$
It then follows immediately that~\eqref{E:M13req0} holds.
\end{proof}

\begin{remark}
At the branch point  $-q_0$,   $\M^{(1)}_{\pm}(x,t,-q_0)$ satisfy a similar relation, namely
$$
\M^{(1)}_{+}(x,t,-q_0)= \M^{(1)}_{-}(x,t,-q_0) \bPi (-q_0).
$$
\end{remark}


\begin{figure}
\centering
\begin{tikzpicture}[
    scale=0.9,
    line cap=round,
    line join=round,
    >=latex,
    contour/.style={line width=2.0pt, draw=black!70},
    arrowcontour/.style={
        contour,
        postaction={decorate},
        decoration={
            markings,
            mark=at position 0.56 with {\arrow{>}}
        }
    },
    dasharrow/.style={
        line width=1.5pt,
        draw=black!70,
        dashed,
        postaction={decorate},
        decoration={
            markings,
            mark=at position 0.56 with {\arrow{>}}
        }
    },
    regionlabel/.style={font=\normalsize},
    every node/.style={font=\normalsize}
]

\coordinate (O)  at (-3.20,0);
\coordinate (A)  at (-4.95,1.55);
\coordinate (B)  at (-1.45,1.55);
\coordinate (C)  at (-4.95,-1.55);
\coordinate (D)  at (-1.45,-1.55);

\coordinate (Z1) at (0.45,0);
\coordinate (T)  at (2.30,1.55);
\coordinate (S)  at (2.30,-1.55);
\coordinate (Z0) at (4.15,0);

\coordinate (E)  at (5.25,0.92);
\coordinate (F)  at (5.25,-0.92);

\draw[arrowcontour] (-6.55,1.55) -- (A);
\draw[arrowcontour] (A) -- (B);
\draw[arrowcontour] (B) -- (T);
\draw[arrowcontour] (T) -- (7.55,1.55);

\draw[arrowcontour] (-6.55,-1.55) -- (C);
\draw[arrowcontour] (D) -- (S);
\draw[arrowcontour] (S) -- (7.55,-1.55);

\draw[arrowcontour] (E) -- (7.55,0.92);
\draw[arrowcontour] (F) -- (7.55,-0.92);

\draw[dasharrow] (-6.55,0) -- (Z1);
\draw[dasharrow] (Z1) -- (Z0);
\draw[dasharrow] (Z0) -- (7.55,0);

\draw[arrowcontour]
    (A) .. controls (-4.95,0.82) and (-4.25,0.14) .. (O);

\draw[arrowcontour]
    (O) .. controls (-2.15,0.14) and (-1.45,0.82) .. (B);

\draw[arrowcontour]
    (C) .. controls (-4.95,-0.82) and (-4.25,-0.14) .. (O);

\draw[arrowcontour]
    (O) .. controls (-2.15,-0.14) and (-1.45,-0.82) .. (D);

\draw[arrowcontour] (B) -- (Z1);
\draw[arrowcontour] (D) -- (Z1);
\draw[arrowcontour] (Z1) -- (T);
\draw[arrowcontour] (Z1) -- (S);

\draw[arrowcontour] (T) -- (Z0);
\draw[arrowcontour] (S) -- (Z0);
\draw[arrowcontour] (Z0) -- (E);
\draw[arrowcontour] (Z0) -- (F);

\fill (O) circle (3.0pt);
\fill (Z1) circle (3.0pt);
\fill (Z0) circle (3.0pt);

\node[below=3pt] at (O) {$0$};
\node[below=6pt] at (Z1) {$z_1$};
\node[below=6pt] at (Z0) {$z_0$};

\node[regionlabel] at (-5.70,0.65) {$\Omega_1^+$};
\node[regionlabel] at (-5.70,-0.65) {$\Omega_1^-$};

\node[regionlabel] at (-1.25,0.46) {$\Omega_2^+$};
\node[regionlabel] at (-1.25,-0.46) {$\Omega_2^-$};

\node[regionlabel] at (2.28,0.52) {$\Omega_3^+$};
\node[regionlabel] at (2.28,-0.52) {$\Omega_3^-$};

\node[regionlabel] at (6.45,0.42) {$\Omega_4^+$};
\node[regionlabel] at (6.45,-0.42) {$\Omega_4^-$};

\node[regionlabel] at (0.5,0.9) {$\Omega_{z_1}^+$};
\node[regionlabel] at (0.5,-0.9) {$\Omega_{z_1}^-$};

\node[regionlabel] at (4.2,0.83) {$\Omega_{z_0}^+$};
\node[regionlabel] at (4.20,-0.83) {$\Omega_{z_0}^-$};

\node[regionlabel] at (-3.2,0.8) {$\Omega_0$};
\end{tikzpicture}
\caption{Schematic diagram of the regions  $\{\Omega_j^\pm\}_{j=1}^4$,  $\Omega_{z_0}^{\pm}$, $\Omega_{z_1}^{\pm}$ and $\Omega_0$.}
\label{fig:regions}
\end{figure}

\subsection{Second transformation: opening of lenses}
In this subsection, we introduce a transformation which deforms the jump
contour from the real axis to $\Sigma^{(2)}$~(see Fig.~\ref{fig:contour}). We then show that the
corresponding jump matrix $\V^{(2)}$ converges uniformly to the identity matrix
$\bI$ as $t\to\infty$, except in neighborhoods of the stationary points
$z_{0}$ and $z_{1}$.

We now construct the transformation matrix $\G(x,t,z)$. Let
$$
\mathcal{R}=\cup_{j=1}^4 \Omega_{j}^{\pm} \cup \Omega_{z_{0}}^{\pm} \cup \Omega_{z_{1}}^{\pm} \cup \Omega_{0},
$$
 where the regions
$\{\Omega_j^\pm\}_{j=1}^4$,  $\Omega_{z_0}^{\pm}$, $\Omega_{z_1}^{\pm}$ and $\Omega_0$ are  as shown in
Fig.~\ref{fig:regions}.  For convenience, we may assume that the region 
$$\mathcal R\setminus\Omega_{0}= S_{\varepsilon} \cap \{z: |\Im z| \leq  \widetilde \varepsilon \},$$
where  $\widetilde \varepsilon>0$ is chosen sufficiently
small so that $\mathcal R$ contains no discrete spectrum.

According to the factorization \eqref{E:V1dehsfj}, we first
define $\G$ in the regions $\{\Omega_{j}^{\pm}\}_{j=1}^{4}$ so as to
remove the jumps on the real axis. More precisely, we set
$$
\G=\begin{cases}
\widetilde\bT^{-1} (\V^L)^{-1} \widetilde\bT  , & z \in  \Omega_{1}^{+},\\
\widetilde\bT^{-1} \V^U  \widetilde\bT ,& z \in  \Omega_{1}^{-},\\
\bT^{-1}  \bigl( \mathcal B (\V^U_{(3)})^{-1}\mathcal B \bigr)  \bT  ,& z \in  \Omega_{2}^{+},\\
 \bT^{-1}  \bigl( \mathcal B \V^L_{(3)}\mathcal B \bigr)  \bT, & z \in  \Omega_{2}^{-},\\
 \bT^{-1} (\V^U_{(1)})^{-1} \bT, & z \in  \Omega_{3}^{+},\\ 
\bT^{-1} \V^L_{(1)} \bT,&z \in  \Omega_{3}^{-},\\
 \bT^{-1} \bigl( \mathcal A (\V^U_{(2)})^{-1}\mathcal A \bigr)   \bT, & z \in \Omega_{4}^{+},\\
\bT^{-1} \bigl( \mathcal A \V^L_{(2)}\mathcal A \bigr) \bT , & z \in \Omega_{4}^{-}.
\end{cases}
$$
If $\G$ were defined only in this way, then the jumps on the cross-shaped
contours near the stationary points $z_{0}$ and $z_{1}$ would be rather
complicated. Therefore, we further define the transformation matrix in the
regions $\Omega_{z_{0}}^{\pm}$  and $\Omega_{z_{1}}^{\pm}$ in order to separate out some irrelevant jump factors. Thus we further  define 
$$
\G=\begin{cases}
\bpm 
1&0&\bigl[\tilde{r}_2^*(z^*)+\frac{1}{\gamma(z)}\tilde{r}_1^*(z^*)\tilde{r}_3(z)  \bigr]\frac{T_3}{T_1}(z)\e^{\theta_{13}} \\
0&1& -\tilde r_3(z) \frac{T_3}{T_2}(z)\e^{\theta_{23}}\\
0&0&1
\epm, & z \in  \Omega_{z_{1}}^{+},\\
\bpm 
1&0&0\\
0&1&0\\
\bigl[\tilde{r}_2(z)+\frac{1}{\gamma(z)}\tilde{r}_3^*(z^*)\tilde{r}_1(z)  \bigr]\frac{T_1}{T_3}(z)\e^{\theta_{31}} & \frac{1}{\gamma(z)} \tilde r^*_3(z^*)\frac{T_2}{T_3}(z)\e^{\theta_{32}} & 1
\epm, & z \in  \Omega_{z_{1}}^{-},\\
\bpm
1& -\frac{1}{\gamma(z)} \tilde r^*_1(z^*)\frac{T_2}{T_1}(z)\e^{\theta_{12}}&   \tilde r^*_2(z^*)\frac{T_3}{T_1}(z)\e^{\theta_{13}}\\
0&1&0\\
0&0&1
\epm, & z \in  \Omega_{z_{0}}^{+},\\
\bpm

1&0&0\\
-\tilde r_1(z) \frac{T_1}{T_2}(z)\e^{\theta_{21}} & 0&0\\
\tilde r_2(z)\frac{T_1}{T_3}(z)\e^{\theta_{31}}&0&1
\epm, & z \in  \Omega_{z_{0}}^{-}.
\end{cases}
$$
Finally, we define $\G$ in the region $\Omega_0$ as follows:
$$
\G(x,t,z)=\bpm
1&0&0\\
\frac{\tilde r_1(z)}{1-\frac{1}{\gamma(z)} \tilde r_1(z) \tilde r_1^*(z^*)}\frac{T_1}{T_2}(z)\e^{\theta_{21}}&1&0\\
0&0&1
\epm, \quad z \in \Omega_0.
$$
We now explain why an additional transformation matrix $\G(x,t,z)$ is
introduced in the region $\Omega_0$. From the definition of $\G$ in
$\Omega_1^+$ and $\Omega_2^+$, its $(2,1)$-entry is given by
$$
\G_{21}(x,t,z)=
\begin{cases}
-r_1(z)\frac{\widetilde T_1}{\widetilde T_2}(z)
\e^{t\phi_{21}(\xi,z)}, & z\in\Omega_1^+,\\[3mm]
\frac{\tilde r_1(z)}
{1-\frac{1}{\gamma(z)}\tilde r_1(z)\tilde r_1^*(z^*)}
\frac{T_1}{T_2}(z)
\e^{t\phi_{21}(\xi,z)}, & z\in\Omega_2^+ .
\end{cases}
$$
By Remark~\ref{R:fsjs},
$
r_1(z)=\mathcal O(1)$, $\Omega_1^+\ni z\to0$,
but this does not in general imply $r_1(z)\to 0$. Similarly, from
\eqref{E:Az0}, \eqref{E:ABrel}, and the definition of $\tilde r_1$, we have
$
\tilde r_1(z)=\mathcal O(1)$,
$\Omega_2^+\ni z\to0
$.
Moreover,
\[
\operatorname{Re}\phi_{21}(\xi,z)\to0,
\qquad z\to0.
\]
Therefore, without a further modification near the origin, the relevant
entries of the jump matrices on $\Sigma_9^{(2)}$ and $\Sigma_{10}^{(2)}$
need not be uniformly small as $t\to\infty$. To overcome this difficulty,
we introduce the transformation $\G$ in $\Omega_0$. This additional step
yields the estimate~\eqref{E:estinear0} for $\V^{(2)}$ near the origin, so
that the contribution of the jumps in a neighborhood of $0$ can be absorbed
into the error term.

The second transformation is defined as follows:
\be \label{E:sectrans}
\M^{(2)}(x,t,z)=\M^{(1)}(x,t,z) \times 
\begin{cases}
\G(x,t,z), & z \in \mathcal{R},\\
\bI, & \text{elsewhere}.
\end{cases}
\ee
The resulting jump contour is denoted by $\Sigma^{(2)}$, as shown in
Fig.~\ref{fig:contour}. The corresponding jump matrix is denoted by
$\V^{(2)}$. We use $\V^{(2)}_{j}$ to denote the jump matrix on the
subcontour of $\Sigma^{(2)}$ labeled by $j$. Below we list the expressions
for $\{\V^{(2)}_{j}\}_{j=1}^{12}$. The jump matrices on the remaining
unlabeled contours decay exponentially as $t\to\infty$, and their explicit
forms are omitted for brevity.

\begin{figure}
\centering
\begin{tikzpicture}[
    scale=0.9,
    line cap=round,
    line join=round,
    >=latex,
    contour/.style={line width=2pt, draw=black!70},
    arrowcontour/.style={
        contour,
        postaction={decorate},
        decoration={
            markings,
            mark=at position 0.56 with {\arrow{>}}
        }
    },
    dasharrow/.style={
        line width=1.5pt,
        draw=black!70,
        dashed,
        postaction={decorate},
        decoration={
            markings,
            mark=at position 0.56 with {\arrow{>}}
        }
    },
    circlearrow/.style={
        contour,
        postaction={decorate},
        decoration={
            markings,
            mark=at position 0.25 with {\arrow{>}}
        }
    },
    regionlabel/.style={font=\normalsize},
    every node/.style={font=\normalsize}
]

\coordinate (O)  at (-3.20,0);
\coordinate (A)  at (-4.95,1.55);
\coordinate (B)  at (-1.45,1.55);
\coordinate (C)  at (-4.95,-1.55);
\coordinate (D)  at (-1.45,-1.55);

\coordinate (Z1) at (0.45,0);
\coordinate (T)  at (2.30,1.55);
\coordinate (S)  at (2.30,-1.55);
\coordinate (Z0) at (4.15,0);

\coordinate (E)  at (5.25,0.92);
\coordinate (F)  at (5.25,-0.92);

\draw[arrowcontour] (-6.55,1.55) -- (A);
\draw[arrowcontour] (A) -- (B);
\draw[arrowcontour] (B) -- (T);
\draw[arrowcontour] (T) -- (7.55,1.55);

\draw[arrowcontour] (-6.55,-1.55) -- (C);
\draw[arrowcontour] (D) -- (S);
\draw[arrowcontour] (S) -- (7.55,-1.55);

\draw[arrowcontour] (E) -- (7.55,0.92);
\draw[arrowcontour] (F) -- (7.55,-0.92);

\draw[dasharrow] (-6.55,0) -- (Z1);
\draw[dasharrow] (Z1) -- (Z0);
\draw[dasharrow] (Z0) -- (7.55,0);

\draw[arrowcontour]
    (A) .. controls (-4.95,0.82) and (-4.25,0.14) .. (O);

\draw[arrowcontour]
    (O) .. controls (-2.15,0.14) and (-1.45,0.82) .. (B);

\draw[arrowcontour]
    (C) .. controls (-4.95,-0.82) and (-4.25,-0.14) .. (O);

\draw[arrowcontour]
    (O) .. controls (-2.15,-0.14) and (-1.45,-0.82) .. (D);

\draw[arrowcontour] (B) -- (Z1);
\draw[arrowcontour] (D) -- (Z1);
\draw[arrowcontour] (Z1) -- (T);
\draw[arrowcontour] (Z1) -- (S);

\draw[arrowcontour] (T) -- (Z0);
\draw[arrowcontour] (S) -- (Z0);
\draw[arrowcontour] (Z0) -- (E);
\draw[arrowcontour] (Z0) -- (F);

\fill (O) circle (3.0pt);
\fill (Z1) circle (3.0pt);
\fill (Z0) circle (3.0pt);

\node[below=6pt] at (O) {$0$};
\node[below=6pt] at (Z1) {$z_1$};
\node[below=6pt] at (Z0) {$z_0$};

\node at (5.15,0.5) {$1$};
\node at (3.2,0.5) {$2$};
\node at (3.2,-0.5) {$3$};
\node at (5.15,-0.5) {$4$};

\node at (1.55,0.5) {$5$};
\node at (-0.6,0.5) {$6$};
\node at (-0.6,-0.5) {$7$};
\node at (1.55,-0.5) {$8$};

\node at (-1.88,0.88) {$9$};
\node at (-4.3,0.82) {$10$};
\node at (-4.3,-0.82) {$11$};
\node at (-2.00,-0.90) {$12$};

%
%
%
%
%

\end{tikzpicture}

\caption{The jump contour $\Sigma^{(2)}$ for $\M^{(2)}$.}
\label{fig:contour}
\end{figure}

\begin{align}
\V^{(2)}_1=&
{\scriptsize
\begin{pmatrix}
  1& 0& 0\\
 0 & 1 &0 \\
  0& \frac{\frac{1}{\gamma(z)}  \tilde{r}_3^*(z^*)}{1+\frac{1}{\gamma(z)} \tilde{r}_3(z) \tilde{r}^*_3(z^*)} \frac{T_{2}}{T_{3}}(z) \e^{\theta_{32}(x,t,z)}&1
\end{pmatrix}
}, \quad
\V^{(2)}_2=
{\scriptsize
\begin{pmatrix}
1  & 0 & 0\\
 0 & 1 &\tilde{r}_3(z)\frac{T_3}{T_2}(z) \e^{\theta_{23}(x,t,z)} \\
 0 & 0 &1
\end{pmatrix}
}, \nonumber\\
\V^{(2)}_3=&
{\scriptsize
\begin{pmatrix}
1  & 0 & 0\\
 0 & 1 &0 \\
 0 & \frac{1}{\gamma(z)}\tilde{r}^*_3(z^*)\frac{T_2}{T_3}(z) \e^{\theta_{32}(x,t,z)}&1
\end{pmatrix}
}, \quad 
\V^{(2)}_{4}=
{\scriptsize
\begin{pmatrix}
  1& 0 & 0\\
 0 & 1 & \frac{\tilde{r}_3(z)}{1+\frac{1}{\gamma(z)} \tilde{r}_3(z) \tilde{r}^*_3(z^*)} \frac{T_{3}}{T_{2}}(z) \e^{\theta_{23}(x,t,z)}\\
 0& 0 &1
\end{pmatrix}
},\nonumber\\
\V^{(2)}_5=&
{\scriptsize
\begin{pmatrix}
  1& \frac{1}{\gamma(z)} \tilde{r}_1^*(z^*) \frac{T_2}{T_1}(z) \e^{\theta_{12}(x,t,z)}& 0\\
  0& 1 &0 \\
  0&0  &1
\end{pmatrix}
},\quad
\V^{(2)}_6=
{\scriptsize
\begin{pmatrix}
 1& 0 & 0\\
 \frac{-\tilde{r}_1(z)}{1-\frac{\tilde{r}_1(z)}{\gamma(z)}  \tilde{r}^*_1(z^*)  } \frac{T_{1}}{T_{2}}(z) \e^{\theta_{21} (x,t,z)}& 1 & 0\\
 0 & 0 &1
\end{pmatrix}
}, \nonumber\\
\V^{(2)}_7=&
{\scriptsize
\begin{pmatrix}
 1& \frac{\frac{1}{\gamma(z)} \tilde{r}^*_1(z^*)}{1-\frac{1}{\gamma(z)} \tilde{r}_1(z) \tilde{r}^*_1(z^*)  } \frac{T_{2}}{T_{1}}(z) \e^{\theta_{12}(x,t,z)}& 0\\
0  & 1 & 0\\
 0 & 0 &1
\end{pmatrix}
},\quad
\V^{(2)}_8=
{\scriptsize
\begin{pmatrix}
  1& 0 & 0\\
  -\tilde{r}_1(z) \frac{T_1}{T_2}(z) \e^{\theta_{21}(x,t,z)}& 1 &0 \\
  0&0  &1
\end{pmatrix}
}, \nonumber \\
\V^{(2)}_9=&
 \begin{pmatrix}
  1& 0 & -\bigl( \tilde{r}^*_2(z^*)+\frac{1}{\gamma(z)}\tilde{r}^*_1(z^*) \tilde{r}_3(z) \bigr) \frac{T_3}{T_1}(z) \e^{\theta_{13}(x,t,z)}\\
   0& 1 &\frac{\tilde r_3(z)+\frac{1}{\gamma(z)} \tilde{r}_1(z) \tilde{r}_2^*(z^*)}{1-\frac{1}{\gamma(z)}  \tilde r_{1}(z) \tilde{r}_1^*(z^*)} \frac{T_3}{T_2}(z) \e^{\theta_{23} (x,t,z)} \\
  0& 0 &1
\end{pmatrix}, \label{E:V2bdashl}  \\
\V^{(2)}_{10}=&
\begin{pmatrix}
1&0&0\\
g_1(z) \frac{T_1}{T_2}(z) \e^{\theta_{21} (x,t,z)}&1&0\\
g_2(z) \frac{T_1}{T_3}(z)\e^{\theta_{31} (x,t,z)}&-\frac{r_3^* (z^*)}{\gamma(z)\delta_1(z) \delta_1^2(\hat{z})} \frac{T_2}{T_3}(z) \e^{\theta_{32} (x,t,z)}&1
\end{pmatrix},\nonumber\\
\V^{(2)}_{11}=&
\begin{pmatrix}
1  &\frac{1}{\gamma(z)}\tilde{r}_1^*(z^*) \frac{\widetilde T_2}{\widetilde T_1}(z)\e^{\theta_{12}(x,t,z) } & -\tilde{r}_2^*(z^*)  \frac{\widetilde T_3}{\widetilde T_1}(z) \e^{\theta_{13} (x,t,z)}\\
0  & 1 & \tilde{r}_3(z) \frac{\widetilde T_3}{\widetilde T_2}(z) \e^{\theta_{23} (x,t,z)}\\
  0&0  &1
\end{pmatrix},\nonumber\\
\V^{(2)}_{12}=&
\begin{pmatrix}
1&\frac{\frac{1}{\gamma(z)} \tilde{r}_1^*(z^*)}{1-\frac{1}{\gamma(z)}  \tilde r_{1}(z) \tilde{r}_1^*(z^*)} \frac{T_2}{T_1} (z)\e^{\theta_{12} }&0\\
0&1&0\\
\bigl( \tilde{r}_2(z)+\frac{1}{\gamma(z)}\tilde{r}^*_3(z^*) \tilde{r}_1(z)\bigr)  \frac{T_1}{T_3}(z) \e^{\theta_{31} } & \frac{\frac{1}{\gamma(z)} \bigl(\tilde r_3^*(z^*)+\tilde{r}_1^*(z^*) \tilde r_2(z)\bigr)}{1-\frac{1}{\gamma(z)}  \tilde r_{1}(z) \tilde{r}_1^*(z^*)} \frac{T_2}{T_3}(z) \e^{\theta_{32} }   &1
\end{pmatrix}, \nonumber  
\end{align}
where
\begin{align}
g_1(z)&= r_1(z) \delta^2_1(z) \delta_1(\hat{z})+\frac{\tilde{r}_1(z)}{1-\frac{1}{\gamma(z)} \tilde{r}_1(z) \tilde{r}^*_1(z^*)  }, \label{E:g1bdsndyj}\\
g_2(z)&= r_2(z) \frac{\delta_1(z)}{\delta_1(\hat{z})} - \frac{1}{ \delta_1(z) \delta_1^2(\hat{z})} \frac{\frac{1}{\gamma(z) }\tilde{r}_1(z) \tilde{r}^*_3(z^*)  }{1-\frac{1}{\gamma(z)} \tilde{r}_1(z) \tilde{r}^*_1(z^*)  }.
\end{align}

We next give the estimates satisfied by the four jump matrices near the origin.

\begin{lemma}\label{L:estofo}
For $z \in \Sigma^{(2)}_{9,10,11,12 }$, we have the estimate
\be \label{E:estinear0}
| \bigl( \V^{(2)}(x,t,z)-\bI \bigr)_{ij}   |  \leq C |z| \e^{t \Re \phi_{ij}(\xi,z)}, \qquad 1 \leq i, j \leq 3.
\ee
\end{lemma}

\begin{proof}
Recalling the asymptotic behavior of the scattering coefficients near the
origin~\eqref{E:Az0} and their symmetry relations \eqref{E:ABrel}, one readily verifies that as $S_{\varepsilon } \ni z  \to 0$,
\be \label{E:xybakns}
r_2(z)= \mathcal{O}(z), \qquad r_3(z)= \mathcal{O}(z),\qquad
\tilde r_2(z)=\mathcal{O}(z), \qquad \tilde r_3(z)= \mathcal{O}(z).
\ee
Applying \eqref{E:xybakns} and the relation
$$
1/\gamma(z)=O(z^{2}),\qquad z\to0,
$$
to the expressions of $\{\V^{(2)}_{j}\}_{j=9}^{12}$; see \eqref{E:V2bdashl}, one
immediately obtains the estimate \eqref{E:estinear0}, except for the $(2,1)$-entry
of $\V^{(2)}_{10}$. We now show that this entry also satisfies
\eqref{E:estinear0}. For this purpose, it suffices to prove that $g_{1}(z)$, defined by \eqref{E:g1bdsndyj}, satisfies
\be \label{E:xydg1zmd}
|g_{1}(z)|\le C|z|, \qquad   z \in \Sigma^{(2)}_{10}.
\ee
 Recalling the definition of $\tilde{r}_1(z)$ and using the identity
$
\delta_{1}(z)=s_{1}(z)\rho(z),
$
a direct calculation gives
\be \label{E:ygzmzyd}
\begin{aligned}
|g_1(z)|& \leq  \left |r_1(z) \delta^2_1(z) \delta_1(\hat{z})+\hat{r}_1(z) \rho^2(z) \rho(\hat{z})\right|+\biggl| \tilde{r}_1(z)-\frac{\tilde{r}_1(z)}{1-\frac{1}{\gamma(z)} \tilde{r}_1(z) \tilde{r}^*_1(z^*)}\biggr|\\
& \leq \bigl | r_1(z) + \frac{\hat{r}_1(z)}{s_1^2(z) s_1(z)} \bigr| | \delta^2_1(z) \delta(\hat z)|+ C|z|^2\\
& \leq C\bigl | r_1(z) + \frac{\hat{r}_1(z)}{s_1^2(z) s_1(z)} \bigr| + C|z|^2,\qquad z \in \Sigma^{(2)}_{10}.
\end{aligned}
\ee
It follows from \eqref{E:hatr123b} that, for $z\in\Sigma_{10}^{(2)}$,
$$
s_{1}(z)=a_{11}(z),\qquad s_{1}(\hat z)=\frac{1}{b_{33}(z)}.
$$
Therefore, using the above formula together with the identity
$
b_{21}=a_{23}a_{31}-a_{21}a_{33}
$,
we obtain
$$
\bigl | r_1(z) + \frac{\hat{r}_1(z)}{s_1^2(z) s_1(z)} \bigr|=\biggl |\frac{a_{21}\bigl( a_{11}b_{11}-a_{33} b_{33}\bigr)+a_{23}a_{31}b_{33}}{a_{11}^2 b_{11}} \biggr |.
$$
Applying again the properties of the scattering coefficients in
\eqref{E:Az0} and \eqref{E:ABrel}, one readily verifies that, for
$z\in\Sigma_{10}^{(2)}$,
\begin{align*}
|a_{21}(z)|&\le C,\qquad
|a_{11}(z)b_{11}(z)-a_{33}(z)b_{33}(z)|\le C|z|,\\
|a_{23}(z)a_{31}(z)b_{33}(z)|&\le C|z|^{2},
\qquad
|a_{11}^{2}(z)b_{11}(z)|\ge c .
\end{align*}
This implies that
$$
\bigl | r_1(z) + \frac{\hat{r}_1(z)}{s_1^2(z) s_1(z)} \bigr| \le C|z|,\qquad z\in\Sigma_{10}^{(2)}.
$$
Hence, by \eqref{E:ygzmzyd}, we obtain~\eqref{E:xydg1zmd}.
The proof is therefore complete.
\end{proof}

Let $D_{\eps}(z_0)$  and  $D_{\eps}(z_1)$   denote two small disks around  $z_0$  and  $z_1$, respectively.  Define $\mathcal{D}=D_{\eps}(z_0) \cup D_{\eps}(z_1)$, then we prove that outside $\mathcal{D}$, $\V^{(2)} - \bI$ is uniformly small as $t \to \infty$.
\begin{lemma}\label{L:yzx}
The jump matrix $\V^{(2)}$ converges to the identity matrix $\bI$  as $t \to \infty$ uniformly for $\xi \in \mathcal{I}_+$ and $z \in \Sigma^{(2)}$ except near the two critical points $z_0$ and $z_1$.
Moreover, the following estimates hold:
\be \label{E:gjV7}
\begin{aligned}
&\|\V^{(2)}-\bI \|_{L^1 (\Sigma^{(2)} \setminus \mathcal{D})}
\leq C t^{-1},\\
&\|\V^{(2)}-\bI \|_{ L^{\infty}(\Sigma^{(2)} \setminus \mathcal{D})}
\leq C t^{-1/2}.
\end{aligned}
\ee
\end{lemma}
\begin{proof}
When $z\in\Sigma^{(2)}\setminus \bigl( \mathcal D  \cup \Sigma^{(2)}_{9,10,11,12} \bigr)$,
$\V^{(2)}-\bI$ decays exponentially as $t\to\infty$. For $z \in \Sigma^{(2)}_{9,10,11,12}$, some additional care is required. 
Next, we show that, for $z\in\Sigma^{(2)}_{9,10,11,12}$, all entries of
$\V^{(2)}-\bI$, except possibly the $(2,1)$- and $(1,2)$-entries, are
exponentially small  as $t\to\infty$. Indeed, a direct calculation gives
$$
\Re \phi_{32}(\xi,z)=\frac{q_0^2 \Im z}{|z|^2} \left( -2 \xi+\frac{2 q_0^2 \Re z}{|z|^2}   \right).
$$
Moreover, for $z\in\Sigma^{(2)}_{9,10,11,12}$, we have
\be \label{E:smgj}
(\Re z)^2 +(\Im z)^2=\frac{q_0^2}{\varepsilon} |\Im z|.
\ee
Then we obtain
$$
\frac{|\Im z|}{|z|^2}=\frac{\varepsilon}{q_0^2}; \qquad
\frac{|\operatorname{Re}z|}{|z|^{2}}
=
\frac{\varepsilon}{q_0^2}\frac{|\operatorname{Re}z|}{|\operatorname{Im}z|}
\to \infty ,
\qquad z\to0.
$$
Hence
$$
|\operatorname{Re}\phi_{23}(\xi,z)|\to\infty,
\qquad z\to0,\quad z\in\Sigma^{(2)}_{9,10,11,12}.
$$
This implies that, for $z\in\Sigma^{(2)}_{9,10,11,12}$,
$|\operatorname{Re}\phi_{23}(\xi,z)|$ admits a positive lower bound.
Consequently, the $(2,3)$- and $(3,2)$-entries of $\V^{(2)}-\bI$ decay
exponentially as $t\to\infty$.
A similar argument applies to $|\operatorname{Re}\phi_{13}(\xi,z)|$.

It remains only to estimate the $(2,1)$-entry and $(1,2)$-entry of $\V^{(2)}-\bI$ on $\Sigma^{(2)}_{9,10,11,12}$.
We take $\bigl(\V_{10}^{(2)}- \bI\bigr)_{21}$ as an example; the estimates on the
other three contours are analogous.   First, one can find
\be \label{E:rephi23}
\Re \phi_{21}(\xi,z)=\Im z(-2\xi+ 2 \Re z).
\ee
Therefore, as $z \to 0$, $\Re \phi_{21}(\xi,z)$ also approaches $0$, which implies that $\| \left(\V^{(2)}_{10}-\bI  \right)_{21} \|_{L^1 \cap L^{\infty}(\Sigma^{(2)}_{10})}$  may no longer decay exponentially. However, Lemma~\ref{L:estofo} shows that when $z \in \Sigma^{(2)}_{10}$, 
$$\left|\left(\V^{(2)}_{10}-\bI  \right)_{21}  \right| \leq C |z| \e^{t \Re \phi_{21}(\xi,z)}.$$ 
Then combined with~\eqref{E:rephi23}  and~\eqref{E:smgj}, it follows that
$$
\left|\left(\V^{(2)}_{10}-\bI  \right)_{21} \right|  \leq C |z| \e^{t \Re \phi_{21}(\xi,z)} \leq C |z|\e^{-c t |\Im z|} \leq C |z| \e^{-c t |z|^2 }, \qquad z \in \Sigma^{(2)}_{10}, \qquad \xi \in \mathcal{I_+}.
$$
Hence, we have
$$
\|\left(\V^{(2)}_{10}-\bI  \right)_{21} \|_{ L^{\infty}(\Sigma^{(2)}_{10})}  \leq C \sup_{u \geq 0} u \e^{-ct u^2} \leq C t^{-1/2}
$$
and
$$
\|\left(\V^{(2)}_{10}-\bI  \right)_{21}\|_{ L^{1}(\Sigma^{(2)}_{10})}  \leq C \int_{0}^{\infty}u e^{-ct u^2} \mathrm{d}u  \leq C t^{-1}.
$$
Thus $\V^{(2)}_{10}$ satisfies estimate~\eqref{E:gjV7}.
\end{proof}

We now examine the behavior of $\M^{(2)}$ near the branch points $\pm q_0$. It is readily observed that $\M^{(2)}$ exhibits no jump in the neighborhood of $\pm q_0$, and the values of   $\M^{(2)}$  at  $\pm q_0$  exist. Furthermore, we have the following lemma:
\begin{lemma}\label{L:M6pmq0}
Let $\M^{(2)}_j$ denote the $j$-th column of the matrix-valued function
$\M^{(2)}$. Then
\begin{equation}\label{E:M6pmq0}
\M^{(2)}_1(x,t,q_0)=\ii \M^{(6)}_3(x,t,q_0),\qquad
\M^{(2)}_1(x,t,-q_0)=-\ii \M^{(6)}_3(x,t,-q_0).
\end{equation}
\end{lemma}

\begin{proof}
We prove only the first identity in \eqref{E:M6pmq0}; the proof of the second
one is completely analogous. Recalling the second transformation, one verifies
that, in a neighborhood of $q_0$,
\begin{equation}\label{E:M2-near-q0}
\M^{(2)}(x,t,z)
=
\M^{(1)}(x,t,z)
\bT^{-1}(z)\e^{\Theta}
\bXi^{-1}(z)\tilde{\F}(z)\bXi(z)
\e^{-\Theta}\bT(z),
\end{equation}
where
$
\bXi(z)=\operatorname{diag}\left(\rho(z),\frac{1}{\rho(z)\rho(\hat z)},\rho(\hat z)\right)
$,
and
\begin{equation}\label{E:tF}
\tilde{\F}(z)=
\begin{cases}
\begin{pmatrix}
1&-\frac{1}{\gamma(z)}\hat{r}_1^*(z^*)&
\hat{r}_2^*(z^*)+\frac{1}{\gamma(z)}\hat{r}_1^*(z^*)\hat{r}_3(z)\\
0&1&-\hat{r}_3(z)\\
0&0&1
\end{pmatrix},
& z\in\Omega_3^+,\\[5mm]
\begin{pmatrix}
1&0&0\\
-\hat{r}_1(z)&1&0\\
\hat{r}_2(z)&\frac{1}{\gamma(z)}\hat{r}_3^*(z^*)&1
\end{pmatrix},
& z\in\Omega_3^-.
\end{cases}
\end{equation}
Here the regions $\Omega_3^\pm$ are shown in Fig.~\ref{fig:regions}.
By Lemma~\ref{L:Mzd}, the boundary values of $\M^{(1)}$ at $q_0$ satisfy
$$
\M^{(1)}_+(x,t,q_0)=\M^{(1)}_-(x,t,q_0)\bPi(q_0).
$$
Moreover, the matrices $\bXi^{\pm1}(z)$, $\e^{\pm\Theta}$ and
$\bT^{\pm1}(z)$ all satisfy the symmetry relation
\[
\bPi^{-1}(z)\mathbf X(\hat z)\bPi(z)=\mathbf X(z).
\]
Consequently, in order to obtain the first identity in \eqref{E:M6pmq0}, it is
enough to prove that
\begin{equation}\label{E:Fpmgs}
\bPi(q_0)\tilde{\F}_+(q_0)\bPi(q_0)=\tilde{\F}_-(q_0),
\end{equation}
where $\tilde{\F}_\pm(q_0)$ denote the corresponding limiting values of
$\tilde{\F}(z)$ as $z\to q_0$ from $\Omega_3^\pm$.

Using the local behavior of $\A(z)$ near the branch points in \eqref{E:Apm},
the symmetry relations for $\A(z)$ and $\B(z)$ in \eqref{E:ABdc} and
\eqref{E:ABrel}, together with the definitions of
$\{\hat r_j(z)\}_{j=1}^3$ in \eqref{E:hatr123}, we obtain
\begin{align*}
&\lim_{z\to q_0}\hat r_1(z)=0,\qquad
\lim_{z\to q_0}\hat r_3(z)=0,\qquad
\lim_{z\to q_0} \hat r_2(z)=\ii,\\
&\lim_{z\to q_0}\left(-\frac{\hat r_1^*(z^*)}{\gamma(z)}\right)
=\frac{a_{12,+}}{a_{11,+}},\qquad
\lim_{z\to q_0}\left(\frac{\hat r_3^*(z^*)}{\gamma(z)}\right)
=\ii\,\frac{a_{12,+}}{a_{11,+}}.
\end{align*}
Substituting these limits into \eqref{E:tF}, a straightforward calculation gives
\eqref{E:Fpmgs}. This proves the first identity in \eqref{E:M6pmq0}, and the
proof is complete.
\end{proof}

\subsection{The outer parametrix and the local parametrix }\label{sub:lG}
In this subsection, we construct the corresponding outer parametrix and local  parametrix.
We begin by considering the following pure-soliton RH problem.
\begin{RHP}[Pure-soliton RH problem]\label{rhp:Msol}
Find a $3 \times 3$ matrix-valued function $\M_{sol}(x,t,z)$ with the following properties:
\bi
\item $\M_{sol}(x,t,\cdot) : \mathbb{C}\setminus \rZ \to \mathbb{C}^{3 \times 3}$ is analytic, where $\rZ=\{ \zeta_j \}_{j=0}^{N-1} \cup \{ \zeta_j^* \}_{j=0}^{N-1}$.

\item  $\M_{sol}$ admits the following asymptotic behaviors
\be \label{E:sayM0in}
  \M_{sol}=\M_{\infty}+\mathcal{O}(\frac{1}{z}), \qquad z \to \infty ; \qquad \M_{sol}=\frac{\ii}{z} \M_0 + \mathcal{O}(1), \qquad z \to 0,
\ee
where $\M_{\infty}$ and $\M_0$ are given by~\eqref{E:M0infty}.
\item $\M_{sol}$ satisfies the symmetries
\begin{equation}\label{E:rhp5-3}
\M_{sol}(x,t,z)=\M_{sol}(x,t,\hat{z}) \bPi(z).
\end{equation}

\item  The following residue conditions hold at  each point $\zeta_j$, $j=0,...,N-1$:
\begin{equation}\label{E:mlstj}
\mathrm{Res}_{z =\zeta_j}\M_{sol}
= \lim_{z\to \zeta_j}\M_{sol} \bpm 0 & 0 &0\\
0&0& 0 \\ \hat{\tau}_{j} \e^{\theta_{31}(x,t,\zeta_j)}&0&0 \epm,
\end{equation}
where $\hat{\tau}_{j}= \tau_j |\delta_1(\zeta_j)|^2 |\delta(\zeta_j)|^2$.
\ei
\end{RHP}

\begin{lemma}\label{L:Msolczwyx}
If  $\frac{\tau_j}{\zeta_j}<0$  holds for all $0\leq j \leq N-1$, then the solution to RH  problem~\ref{rhp:Msol} exists and is unique for any $(x,t) \in \R \times \R_+$.
\end{lemma}

\begin{proof}
See Appendix~\ref{AppB}.
\end{proof}

Let
\be \label{E:Ndarksoliton}
\q_{sol}^{N}(x,t):=-\ii \lim_{z\to\infty}z
\bigl((\M_{ sol})_{21}(x,t,z),(\M_{ sol})_{31}(x,t,z)\bigr)^{\top}.
\ee
Then $\q_{sol}^{N}(x,t)$ is the $N$-dark-soliton solution of the defocusing Manakov system~\eqref{E:demanakovS} generated by the scattering data
$
\{\zeta_j,\hat{\tau}_j\}_{j=0}^{N-1}
$.
Moreover, by the computations in Appendix~\ref{AppB}, one can in fact derive a determinant representation for this solution. For brevity, we omit the explicit formula here and refer the reader to formula~$(3.16)$ in Ref.~\cite{BD2015-1}.

We define the function  $\M^{out}$  as
\be \label{E:Mout}
\M^{out}(x,t,z)=\bP(\infty)^{-1}\M_{\infty}^{-1}\M_{sol}(x,t,z) \bP(z),
\ee
where $\bP(z)$  is given by~\eqref{E:defiofP}. The function $\M^{out}$  serves as the outer parametrix we require, and it satisfies the same residue conditions and  asymptotic properties  as $\M^{(2)}$. 
We now simplify $\M^{out}$ in the two cases: when the compact
velocity interval $\mathcal I_+$ contains exactly one soliton velocity and
when it contains none. 

\begin{lemma}\label{L:Mout-one-none}
Assume first that $\mathcal I_+$ contains exactly one soliton velocity,
say
$$
\mathcal I_+\cap\{\Re\zeta_j\}_{j=0}^{N-1}
=
\{\Re\zeta_{j_0}\},
\qquad
j_0\in\{0,1,\ldots,N-1\}.
$$
 Then, for \(t\) sufficiently large and uniformly for \(\xi\in\mathcal I_+\), one has
\be \label{E:nashhl}
\M^{out}(x,t,z)
=
\M_{out}^{(j_0)}(x,t,z)
\left(
\bI+\mathcal O(\e^{-ct})
\right), \qquad z\in\mathbb C\setminus
\left(
\bigcup_{j=0}^{N-1}(\overline{\mathcal D_j}\cup \overline{\mathcal D_j^*})
\right).
\ee
Here
\be \label{E:Moutjdy}
\M_{out}^{(j_0)}(x,t,z)
=
\bP^{(j_0)}(\infty)^{-1}
\M_{\infty}^{-1}
\M_{sol}^{(j_0)}(x,t,z)
\bP^{(j_0)}(z),
\ee
where $\M_{sol}^{(j_0)}$ denotes the solution of the pure-soliton
RH problem~\ref{rhp:Msol} with the scattering data
$
\{\zeta_j,\hat{\tau}_j\}_{j=0}^{N-1}
$
replaced by the single datum
$$
\{\zeta_{j_0},\tilde{\tau}_{j_0}\},
\qquad
\tilde{\tau}_{j_0}
=
\hat{\tau}_{j_0}
\prod_{\ell<j_0}
\left|
\frac{\zeta_{j_0}-\zeta_{\ell}}
{\zeta_{j_0}-\zeta_{\ell}^*}
\right|^2 .
$$
Moreover,
$$
\bP^{(j_0)}(z)=
\mathrm{diag} \left(
\cP_1^{(j_0)}(z),
\frac{1}{\cP_1^{(j_0)}(z)\cP_1^{(j_0)}(\hat z)},
\cP_1^{(j_0)}(\hat z)
\right),
\quad
\cP_1^{(j_0)}(z)=
\begin{cases}
1, & \Re \zeta_{j_0}\leq \xi \leq m_1,\\[1mm]
\dfrac{z-\zeta_{j_0}}{z-\zeta_{j_0}^*},
& m_0\leq \xi<\Re \zeta_{j_0}.
\end{cases}
$$
If $\mathcal I_+$ contains no soliton velocities, i.e. $ \mathcal I_+\cap\{\Re\zeta_j\}_{j=0}^{N-1}=\emptyset$, then, for $t$ sufficiently large and uniformly for $\xi\in\mathcal I_+$, one has
\be \label{E:choux}
\M^{out}(x,t,z)
=
\M_{\infty}^{-1}\E_+(z)
\left(
\bI+\mathcal O(\e^{-ct})
\right), \qquad z\in\mathbb C\setminus
\left(
\bigcup_{j=0}^{N-1}(\overline{\mathcal D_j}\cup\overline{\mathcal D_j^*})
\right).
\ee

\end{lemma}

\begin{proof}
We first consider the case where $\mathcal I_+$ contains exactly
one soliton velocity. For every $j\neq j_0$, we remove the
pole at $\zeta_j$ by replacing the corresponding residue condition with
a jump on the boundary of a small disk. Let $\mathcal D_j$ be a sufficiently small disk
centered at $\zeta_j$  so that the disks $\mathcal D_j$ and their
conjugates $\mathcal D_j^*$ are mutually disjoint. Inside $\mathcal D_j$
and $\mathcal D_j^*$, define
\be \label{E:Mouthatde}
\hat{\M}^{out}(x,t,z)
=
\M^{out}(x,t,z) \times
\begin{cases}
\bI-\dfrac{\mathbf n_j}{z-\zeta_j}, & z\in\mathcal D_j,\\[3mm]
\bI+
\dfrac{\zeta_j^*}{\zeta_j}
\bPi(\zeta_j^*)
\dfrac{\mathbf n_j}{z-\zeta_j^*}
\bPi(\zeta_j^*),
& z\in\mathcal D_j^*,
\end{cases}
\ee
where
\be \label{E:njddy}
\mathbf n_j=
\begin{cases}
\bpm
0&0&\alpha_j\\
0&0&0\\
0&0&0
\epm,
& \Re \zeta_j>\xi,\\[5mm]
\bpm
0&0&0\\
0&0&0\\
\beta_j&0&0
\epm,
& \Re \zeta_j<\xi.
\end{cases}
\ee
Here $\alpha_j$ and $\beta_j$ are defined in~\eqref{E:alphaj} and
\eqref{E:betaj}, respectively. This transformation removes the poles at
$\zeta_j$ and $\zeta_j^*$ for all $j\neq j_0$, and produces jumps on
$\partial\mathcal D_j$ and $\partial\mathcal D_j^*$ only. More precisely,
$$
\hat{\M}^{out}_+(x,t,z)
=
\hat{\M}^{out}_-(x,t,z)\hat{\V}(x,t,z),
\qquad
z\in \partial\mathcal D_j\cup\partial\mathcal D_j^* .
$$
Since $\|\mathbf n_j\|=\mathcal O(\e^{-ct})$ uniformly for
$\xi\in\mathcal I_+$, it follows that
$$
\|\hat{\V}(x,t,\cdot)-\bI\|_{L^\infty(\partial\mathcal D_j\cup\partial\mathcal D_j^*)}
=
\mathcal O(\e^{-ct}),
\qquad t\to\infty .
$$
It remains to compare $\hat{\M}^{out}$ with the one-soliton outer model.
By construction, $\hat{\M}^{out}$ and $\M_{out}^{(j_0)}$ have the same
residue conditions at $\zeta_{j_0}$ and $\zeta_{j_0}^*$. Hence the quotient
\be \label{E:e}
\mathbf e(x,t,z)
=
\hat{\M}^{out}(x,t,z)
\bigl(\M_{out}^{(j_0)}(x,t,z)\bigr)^{-1}
\ee
has removable singularities at $\zeta_{j_0}$ and $\zeta_{j_0}^*$. Moreover,
$\M_{out}^{(j_0)}$ and its inverse are uniformly bounded on all circles
$\partial\mathcal D_j$ and $\partial\mathcal D_j^*$ with $j\neq j_0$.
Therefore the jump matrix for $\mathbf e$ is still of the form
$\bI+\mathcal O(\e^{-ct})$. In addition, by the same argument as in
Lemma~\ref{L:zyyl}, the function $\mathbf e$ has finite limits
near the possible singular points $\{0,\pm q_0\}$. Thus $\mathbf e$
satisfies a small-norm RH problem.
By the standard small-norm RH theory, we obtain, uniformly for
\(\xi\in\mathcal I_+\),
\be \label{E:fbzmbww}
\mathbf e(x,t,z)
=
\bI+\mathcal O(e^{-ct}),
\qquad
z\in\mathbb C\setminus
\left(
\bigcup_{j=0}^{N-1}(\mathcal D_j\cup\mathcal D_j^*)
\right).
\ee
Combining this estimate with~\eqref{E:e} gives~\eqref{E:nashhl}.

We now assume that \(\mathcal I_+\) contains no soliton velocities. In this case, the above pole-removing transformation is applied to every \(\zeta_j\). Hence \(\hat{\M}^{out}\) is analytic at all discrete eigenvalues and has only exponentially small jumps on the circles \(\partial\mathcal D_j\) and \(\partial\mathcal D_j^*\).
 Define 
$$
 \mathbf e(x,t,z) = \hat{\M}^{out}(x,t,z)\E_+^{-1}(z)\M_{\infty}. 
$$
 Then \(\mathbf e\) is analytic away from these circles, normalized to \(\bI\) at infinity, and has jump matrices of the form \(\bI+\mathcal O(\e^{-ct})\). Hence \(\mathbf e\) satisfies a small-norm RH problem, and the standard small-norm estimate gives~\eqref{E:choux}.
\end{proof}

The above lemma shows that the expression for the \(N\)-dark-soliton solution given in~\eqref{E:Ndarksoliton} can also be simplified in these two cases.
\begin{corollary}\label{C:qsolN-reduction}
Assume first that $\mathcal I_+$ contains exactly
one soliton velocity, i.e. 
$\Re \zeta_{j_0} \in \mathcal{I_+}$. Define
\be \label{E:qsoljoddy}
\q_{sol}^{(j_0)}(x,t)
=
-\ii \lim_{z\to\infty}z
\left(
(\M_{sol}^{(j_0)})_{21}(x,t,z),
(\M_{sol}^{(j_0)})_{31}(x,t,z)
\right)^{\top}.
\ee
Then, uniformly for $\xi\in\mathcal I_+$,
\be \label{E:qsolNbiby1sol}
\q_{sol}^{N}(x,t)
=
\left(\prod_{\ell<j_0}\frac{\zeta_{\ell}}{\zeta_{\ell}^*} \right)
\q_{sol}^{(j_0)}(x,t)
+\mathcal O(\e^{-ct}),
\qquad t\to\infty.
\ee
Moreover, the one-dark-soliton solution $\q_{sol}^{(j_0)}(x,t)$ admits the explicit form
\be \label{E:one-dark-explicit}
\q_{sol}^{(j_0)}(x,t)
=
\q_+\e^{\ii\theta_{j_0}}
\left[
\cos\theta_{j_0}
-\ii\sin\theta_{j_0}
\tanh\Bigl(
q_0\sin\theta_{j_0}
\bigl(x-2q_0\cos\theta_{j_0}t-x_{j_0}\bigr)
\Bigr)
\right],
\ee
where
$$
\zeta_{j_0}=q_0\e^{\ii\theta_{j_0}},
\qquad
x_{j_0}
:=
\frac{1}{2q_0\sin\theta_{j_0}}
\log\left(
-\frac{\tilde{\tau}_{j_0}}
{2\zeta_{j_0}\sin\theta_{j_0}}
\right).
$$
If $\mathcal I_+$ contains no soliton velocities, then, uniformly for
$\xi\in\mathcal I_+$,
\be \label{E:wuguziqingjnjx}
\q_{sol}^{N}(x,t)
=\left(
\prod_{j\in\nabla^+}\frac{\zeta_j}{\zeta_j^*}\right) \q_+
+\mathcal O(\e^{-ct}),
\qquad t\to\infty.
\ee
\end{corollary}

\begin{proof}
We first consider the case where \(\mathcal I_+\) contains the single soliton velocity \(\Re\zeta_{j_0}\). By Lemma~\ref{L:Mout-one-none}, the outer model \(\M^{out}\) differs from the one-soliton outer model \(\M_{out}^{(j_0)}\) by an exponentially small error. More precisely, the corresponding error matrix satisfies 
\[ \mathbf e(x,t,z) = \bI+\mathcal O(\e^{-ct})\frac{1}{z} + \mathcal O\left(\frac{1}{z^2}\right), \qquad z\to\infty, \]
 uniformly for \(\xi\in\mathcal I_+\) and for \(t\) sufficiently large.
Then for large $z$, we have
\be \label{E:Moutwlisjjx}
\M^{out}(x,t,z)
=
\M_{out}^{(j_0)}(x,t,z)
\left(
\bI+\mathcal O(\e^{-ct})\frac{1}{z}
+\mathcal O\left(\frac{1}{z^2}\right)
\right).
\ee
Using~\eqref{E:Mout}, \eqref{E:Moutjdy}, and \eqref{E:Moutwlisjjx}, together
with the reconstruction formula, and comparing the \(z^{-1}\)-coefficients, we
obtain~\eqref{E:qsolNbiby1sol}. The explicit formula \eqref{E:one-dark-explicit} follows by solving the one-pole pure-soliton RH problem for \(\M_{sol}^{(j_0)}\), as in Appendix~\ref{AppB}, and then applying the formula \eqref{E:qsoljoddy}.
The proof of~\eqref{E:wuguziqingjnjx} is completely analogous to that of
\eqref{E:qsolNbiby1sol}, and is therefore omitted for brevity.
\end{proof}


Next, we proceed to construct the local models. In the previous subections, we have shown that the jump matrix $\V^{(2)}$ decays  uniformly  to  $\bI$ except in the vicinity of the two critical points $z_0$ and $z_1$. We will show that $\M^{(2)}$
 can be approximated locally near the two  critical points by solutions of solvable model RH problems. 

We introduce two new variables in the neighborhoods of $z_1$ and $z_0$ respectively:
\begin{align}\label{E:yellr}
 y_{\ell}=\sqrt{2t}(z-z_1), \qquad z \in D_{\epsilon }(z_1);\qquad  y_{r}=\frac{\xi^2}{q_0^2}\sqrt{2t}(z-z_0), \qquad z \in D_{\epsilon }(z_0).
\end{align}
By performing Taylor expansions of the functions  $\theta_{32}(x,t,z)$  and  $\theta_{21}(x,t,z)$  at  $z_1$  and  $z_0$  respectively, we obtain
\begin{align}\label{E:thetabs}
\theta_{21}(x,t,z)-\theta_{21}(x,t,z_1)=-\frac{\ii}{2}y_{\ell}^2,\qquad
\theta_{32}(x,t,z)-\theta_{32}(x,t,z_0)=\frac{\ii}{2}y_{r}^2+\theta_r(x,t,z),
\end{align}
 where
$$
\theta_{r}(x,t,z)=\mathcal{O}(z-z_0)^3, \ \  \text{as $z \to z_0$.}
$$
From Eq.~\eqref{E:tilder123}, we have
\be \label{E:zhsh}
\begin{aligned}
\tilde{r}_1(z)\frac{T_{1}}{T_{2}}(z)  =& \hat{r}_1(z) \rho^2(z) \rho(\hat{z}) \frac{\cP^2_1(z) \cP_1(\hat{z})}{\delta(z) \delta^2(\hat{z})},\\
\tilde{r}^*_3(z^*)\frac{T_{2}}{T_{3}}(z)=&\hat{r}^*_3(z^*) \frac{\delta^2(z) \delta(\hat{z})}{\rho(z)\rho^2(\hat{z})  \cP_1(z) \cP_1^2(\hat{z})}.
\end{aligned}
\ee
Eqs.~\eqref{E:expdelta1sj} and~\eqref{E:yellr}  imply that, for $\xi \in \mathcal{I}_+$
\be \label{E:T12}
\begin{aligned}
&\frac{\delta^2(z) \delta(\hat{z})}{\rho(z)\rho^2(\hat{z})  \cP_1(z) \cP_1^2(\hat{z})}=\e^{-2 \ii \nu \log (-y_r)} d_0^r d_1^r, \qquad z \in D_{\eps}(z_0) \setminus [z_0,\infty),\\
& \rho^2(z) \rho(\hat{z}) \frac{\cP^2_1(z) \cP_1(\hat{z})}{\delta(z) \delta^2(\hat{z})}= \e^{2 \ii \nu \log y_{\ell}}d_0^{\ell} d_1^{\ell}, \qquad  z \in D_{\eps}(z_1) \setminus [0,z_1],
\end{aligned}
\ee
where
\begin{align}
&d_0^r=(\frac{q_0^2}{\xi^2 \sqrt{2t}})^{-2 \ii  \nu} \e^{-2 \chi(z_0)}     \frac{\delta(z_1)}{ \rho^2(z_1)\rho(z_0) \cP_1(z_0) \cP_1^2(z_1)}, \label{E:do1}\\
&d_1^r=\frac {\delta(\hat{z})  \rho^2(z_1)\rho(z_0) \cP_1^2(z_1) \cP_1(z_0)} {\delta(z_1)  \rho(z) \rho^2(\hat{z}) \cP_1(z) \cP^2_1(\hat{z})} \e^{-2 \chi(z)+2\chi(z_0)},\\
&d_0^{\ell}=(\sqrt{2t})^{-2 \ii \nu} \frac{\rho^2(z_1) \rho(z_0) \cP^2_1(z_1) \cP_1(z_0) }{\delta^2(0)\delta(z_1) } \e^{-2 \tilde{\chi} (z_1)},\label{E:dol-1}\\
&d_1^{\ell}=\frac{\delta(z_1) \rho(\hat{z}) \rho^2(z) \cP_1(\hat{z}) \cP_1^2(z)} {\delta(z) \rho^2(z_1)\rho(z_0) \cP_1^2(z_1) \cP_1(z_0)} \e^{-2 \tilde{\chi}(z)+2\tilde{\chi}(z_1)}.
\end{align}
In order to relate $\M^{(2)}$ to the solutions of model RH problems, we define
\be \label{E:Y}
\bY(\xi,t)=\begin{cases}
\bY_L(\xi,t), &z \in D_{\eps}(z_1),\\
\bY_R(\xi,t), & z \in D_{\eps}(z_0),
\end{cases}
\ee
where
\be \label{E:YLYR}
\begin{aligned}
&\bY_R(\xi,t)=\bpm
1& &\\
 & (d_0^{r})^{-1/2}\e^{-\frac{\theta_{32}(x,t,z_0)}{2}}&\\
 & & (d_0^{r})^{1/2}\e^{\frac{\theta_{32}(x,t,z_0)}{2}}
\epm, \\
&\bY_L(\xi,t)=\bpm
(d_0^{\ell})^{-1/2}\e^{-\frac{\theta_{21}(x,t,z_1)}{2}}& &\\
 &(d_0^{\ell})^{1/2}\e^{\frac{\theta_{21}(x,t,z_1)}{2}} &\\
 & & 1
\epm.
\end{aligned}
\ee
\begin{lemma}\label{L:Y}
The function $\bY(\xi,t)$ is uniformly bounded for $\xi \in \mathcal{I}_+$ and $t \geq 2$.
Moreover, the functions $d_0^r(\xi,t)$, $d_0^{\ell}(\xi,t)$, $d_0^r(\xi,z)$ and $d_0^{\ell}(\xi,z)$ satisfy
\begin{align}
&|d_0^r(\xi,t)|=|d_0^{\ell}(\xi,t)|=1, \label{E:d0rd0ell}\\
&|d_1^r -1| \leq C |z-z_0|\left(1+|\log |z-z_0|   |\right),  \quad z \in \Sigma^{(2)} \cap D_{\eps }(z_0), \label{E:d1estr}\\
&|d_1^{\ell} -1| \leq C |z-z_1|\left(1+|\log |z-z_1|   |\right), \quad  z \in \Sigma^{(2)} \cap D_{\eps }(z_1). \label{E:d1estrr}
\end{align}
\end{lemma}

\begin{proof}
Note that 
$$
|\delta(z_0)|=|\delta(z_1)|=|\rho(z_0)|=|\rho(z_1)|=|\cP_1(z_0)|=|\cP_1(z_1)|=1,
$$
 and $\Re \chi(z_0)=\Re \tilde{\chi}(z_1)=0$. Then it follows directly from~\eqref{E:do1} and~\eqref{E:dol-1} that~\eqref{E:d0rd0ell} holds. Observing that $\delta(\hat{z}) /\bigl(\rho(z) \rho^2(\hat{z}) \cP1(z) \cP_1^2(\hat{z}) \bigr)$ analytic
for $z \in  \Sigma^{(2)} \cap D_{\eps }(z_0)$,~\eqref{E:d1estr}  follows from~\eqref{E:L13-st-5}. The proof of~\eqref{E:d1estrr} is similar.
\end{proof}

We can approximate $\M^{(2)}(x,t,z)$ in $\cD$ by the $3 \times 3$-matrix-valued function $\M^{loc}(x,t,z)$ defined by
\be \label{E:Mloc}
\M^{loc}(x,t,z)=\begin{cases}
\bY_L(\xi,t) \M^{X,L}(x,t,y_{\ell}(z)) \bY_L^{-1}(\xi,t), & z \in D_{\eps}(z_1),\\
 \bY_R(\xi,t) \M^{X,R}(x,t,y_r(z)) \bY_R^{-1}(\xi,t), & z \in D_{\eps}(z_0).
\end{cases}
\ee
Here
\be \label{E:defofMXLXR}
\begin{aligned}
\M^{X,L}(y_{\ell})&=\bpm
\M^{PC}\bigl(y_{\ell};\frac{1}{\gamma(z_1)}\hat{r}_1^*(z_1), \hat r_1(z_1) \bigr)&\mathbf{0}\\
\mathbf{0^{\top}}&1
\epm, \\
\M^{X,R}(y_r)&=\bpm
1&\mathbf{0^{\top}}\\
\mathbf{0}&\boldsymbol{ \sigma_1} \M^{PC}\bigl(-y_{r};-\frac{1}{\gamma(z_0)}\hat{r}_3^*(z_0) , \hat r_3(z_0) \bigr) \boldsymbol{\sigma_1}
\epm,
\end{aligned}
\ee
where $\M^{PC}(\zeta;p,q)$ is the solution of model RH problem~\ref{rhpNL}, and $\boldsymbol{\sigma_1}$ is the  Pauli matrix $\boldsymbol{\sigma_1}=\bpm0&1\\ 1&0 \epm$. The following lemma is a direct consequence of~\eqref{E:defofMXLXR} and Lemma~\ref{L:modelasy}.

\begin{lemma}
The functions $\M^{X,L}(y_{\ell})$ and $\M^{X,R}(y_r)$ have the following asymptotic behaviors:
\begin{equation}\label{E:mXasy}
\begin{aligned}
&\mathbf{M}^{X,L}(y_{\ell}) = \bI + \frac{\mathbf{M}_{\infty}^{X,L}}{y_{\ell}} + \mathcal{O}\big(\frac{1}{y_{\ell}^2}\big), \qquad y_{\ell} \to \infty,  \\
&\mathbf{M}^{X,R}(y_r) = \bI + \frac{\mathbf{M}_{\infty}^{X,R}}{y_r} + \mathcal{O}\big(\frac{1}{y_r^2}\big), \qquad y_r\to \infty, 
\end{aligned}
\end{equation}
where  $\mathbf{M}_{\infty}^{X,L}$ and $\mathbf{M}_{\infty}^{X,R}$ are given by
\begin{equation}\label{E:m1Xdef}
\begin{aligned}
\mathbf{M}_{\infty}^{X,L} =\begin{pmatrix}
0 & \ii \beta_{12} & 0 	\\
-\ii \beta_{21} & 0 &  0\\
0 & 0 & 0 \end{pmatrix}, \qquad
\mathbf{M}_{\infty}^{X,R} =\begin{pmatrix}
0 & 0 & 0 	\\
0& 0 & -\ii \beta_{23} \\
0 &  \ii \beta_{32}  & 0 \end{pmatrix}, 
\end{aligned}
\end{equation}
with
\be \label{E:beta1232}
\begin{aligned}
&\beta_{12}=\frac{\sqrt{2 \pi} \e^{-\frac{ \pi \nu}{2}} \e^{-\frac{ \pi}{4}  \ii }  }{\hat r_1(z_1) \Gamma(\ii \nu)}, \qquad
\beta_{21}=\frac{\gamma(z_1) \sqrt{2 \pi} \e^{\frac{- \pi \nu}{2}} \e^{\frac{ \pi}{4}  \ii }  }{\hat r^*_1(z_1) \Gamma(-\ii \nu)},\\
&\beta_{23}=\frac{\gamma(z_0) \sqrt{2 \pi} \e^{-\frac{ \pi \nu}{2}} \e^{\frac{ \pi}{4}  \ii }  }{\hat r^*_3(z_0) \Gamma(-\ii \nu)}, \qquad
\beta_{32}=\frac{\sqrt{2 \pi} \e^{\frac{- \pi \nu}{2}} \e^{\frac{3 \pi}{4}  \ii }  }{\hat r_3(z_0) \Gamma(\ii \nu)}.
\end{aligned}
\ee
\end{lemma}

We then show that $\M^{loc}$ is a good approximation of $\M^{(2)}$ in $\cD$ for large $t$.
Let $X^{\eps}$ be defined as $X^{\eps}=\Sigma^{(2)} \cap \cD$, and let the boundary $\partial \cD$ of $\cD$ be oriented counterclockwise. Then we have the following lemma.
\begin{lemma}
For each $(x,t)$, the function $\M^{loc}(x,t,z)$ defined in~\eqref{E:Mloc} is an analytic and bounded function of $z \in \cD \setminus X^{\eps}$. Across $X^{\eps}$, $\M^{loc}$ obeys the jump condition $\M^{loc}_+= \M^{loc}_- \V^{loc}$, where the jump matrix $\V^{loc}$ satisfies
\begin{align}\label{E:estVloc}
\begin{cases}
\|\V^{(2)}- \V^{loc} \|_{L^{\infty} (X^{\eps})} \leq C t^{-1/2}\log t,\\
\|\V^{(2)}- \V^{loc} \|_{L^{1}(X^{\eps})} \leq C t^{-1}\log t,
\end{cases} \qquad \xi \in \mathcal{I}_+, \ \ t\geq 2.
\end{align}
Furthermore, as $t \to \infty$,
\be \label{E:bzsy}
\begin{aligned}
&\| (\M^{loc})^{-1}-\bI \|_{L^{\infty}(\partial \cD)}=\mathcal{O}(t^{-1/2}),\\
&(\M^{loc})^{-1}(x,t,z)-\bI=-\frac{\bY_R(\xi,t) \M^{X,R}_{\infty} \bY^{-1}_R(\xi,t)}{\frac{\xi^2}{q_0^2} \sqrt{2t}   (z-z_0)}+\mathcal{O}(t^{-1}),  \quad z \in \partial D_{\eps}(z_0),\\
&(\M^{loc})^{-1}(x,t,z)-\bI=-\frac{\bY_L(\xi,t) \M^{X,L}_{\infty} \bY^{-1}_L(\xi,t)}{\sqrt{2t}  (z-z_1)}+\mathcal{O}(t^{-1}),  \quad z \in \partial D_{\eps}(z_1).
\end{aligned}
\ee
\end{lemma}
\begin{proof}
We have
$$
\V^{(2)}-\V^{loc}=\begin{cases}
\bY_L\left( \tilde{\V}- \V^{X,L} \right)(\bY_{L})^{-1}, & z \in  D_{\eps}(z_1) \cap X^{\eps} ,\\
\bY_R\left( \tilde{\V}- \V^{X,R} \right)(\bY_{R})^{-1}, & z \in  D_{\eps}(z_0) \cap  X^{\eps},
\end{cases}
$$
where $\tilde{\V}=\bY^{-1}\tilde{\V} \bY$, and $\V^{X,L}$ and $\V^{X,R}$ denote the jump matrices corresponding to
$\M^{X,L}$ and $\M^{X,R}$, respectively.
Lemma~\ref{L:Y} demonstrates that both $\bY_L$ and $\bY_R$ are bounded, and thus the proof of~\eqref{E:estVloc} reduces to estimating $\tilde{\V}- \V^{X,L}$ and $\tilde{\V}- \V^{X,R}$. We provide estimates for the $L^1$ and $L^{\infty}$ norms of the jump matrix $\tilde{\V}-\V^{X,L}$ on the contour $\Sigma^{(2)}_8 \cap D_{\eps}(z_1)$. Let $X^{\eps}_8:=\Sigma^{(2)}_8 \cap D_{\eps}(z_1)$.  On this contour, only the $(2,1)$ element of the matrix $\tilde{\V}- \V^{X,L}$ is nonzero.   From \eqref{E:d1estrr}, $(\tilde{\V}- \V^{X,L})_{21}$ can be estimated as
$$
\begin{aligned}
\left|(\tilde{\V}- \V^{X,L})_{21}\right| &=\left|-\hat{r}_1(z)d_1^{\ell}(z)\e^{2\ii \nu \log y_{\ell}} \e^{-\frac{\ii}{2} y_{\ell}^2}+ \hat{r}_1(z_1)\e^{2\ii \nu \log y_{\ell}} \e^{-\frac{\ii}{2} y_{\ell}^2} \right|\\
& \leq \left|\e^{2 \ii \nu \log_{\pi}y_{\ell} }\right| \left|\left(d_{1}^{\ell}(z)-1 \right)\hat{r}_1(z)+\left(\hat{r}_1(z)-\hat{r}_1(z_1)\right) \right|  \left |\e^{-\frac{\ii}{2} y_{\ell}^2 } \right| \\
&\leq C |z-z_1|  \left(1+|(\log |z-z_1|)| \right) e^{-ct|z-z_1|^2}, \qquad z \in X^{\eps}_8.
\end{aligned}
$$
Hence, we have
$$
\| (\tilde{\V}- \V^{X,L})_{21} \|_{ L^{\infty}(X^{\eps}_8)}  \leq C \sup_{u \geq 0} u (1+|\log u|)\e^{-ct u^2} \leq C t^{-1/2} \log t
$$
and
$$
\|(\tilde{\V}- \V^{X,L})_{21} \|_{ L^{1}(X^{\eps}_8)}  \leq C \int_{0}^{\infty}u (1+ |\log u|) \e^{-ct u^2} \mathrm{d}u  \leq C t^{-1} \log t.
$$
A similar argument can be applied to the remaining jumps. This completes the proof of~\eqref{E:estVloc}.

The variables $y_{\ell}$ and $y_r$ goes to infinity as $t \to \infty$ if $z \in  \partial  D_{\epsilon}(z_0)$  and $ z \in \partial  D_{\epsilon}(z_1)$, respectively. This follows because
$$
|y_{\ell}|=\sqrt{2t}|z-z_1|, \qquad
| y_{r}|= \frac{\xi^2}{q_0^2}\sqrt{2t}|z-z_0|.
$$
 Recalling the definition~\eqref{E:Mloc}  of $\M^{loc}$ and applying~\eqref{E:mXasy}   together with the residue theorem, we obtain~\eqref{E:bzsy}.

\end{proof}

\subsection{Final transformation: small-norm RH problem}\label{smallRhP}
We define the final transformation to obtain a small-norm RH problem as follows:
\be \label{E:defE}
\cE(x,t,z)=\begin{cases}
\M^{(2)}(x,t,z)(\M^{out})^{-1}(x,t,z), & z\in \C \setminus \cD,\\
\M^{(2)}(x,t,z)(\M^{loc})^{-1}(x,t,z)(\M^{out})^{-1}(x,t,z), &z \in \cD.
\end{cases}
\ee
We will show that $\cE(x,t,z)$ is close to $\bI$ for large $t$ and $\xi \in \mathcal{I}_{+}$.  We first note that although $\M^{(2)}(x,t,z)$ has a singularity at $0$ and  $(\M^{out})^{-1}(x,t,z) $ has singularities at the branch points  $\pm q_0$, the function  $\cE(x,,t,z)$ is well-defined at these points.
\begin{lemma}\label{L:zyyl}
$\cE(x,t,z)=\mathcal{O}(1)$ as $z \to z_{\star}$, where $z_{\star} \in \{0, q_0,-q_0 \}$.
\end{lemma}

\begin{proof}
Near $z_\star\in\{0,q_0,-q_0\}$, we have
\[
\cE(x,t,z)
=
\mathbf W(x,t,z)\M_{sol}^{-1}(x,t,z)\M_\infty \bP(\infty),
\qquad
\mathbf W(x,t,z):=\M^{(2)}(x,t,z)\bP^{-1}(z).
\]
Since $\M_\infty$ and $\bP(\infty)$ are constant matrices, it is enough to
show that the product $\mathbf W(x,t,z)\M_{sol}^{-1}(x,t,z)$ is bounded at
these points.

We first consider the point $z=0$. By tracing the transformations from
$\M$ to $\M^{(2)}$ and using the fact that $\bP(z)$ is diagonal and
bounded as $z\to0$, one obtains
\be \label{E:asyW}
\mathbf W(x,t,z)
=
\frac{1}{z}
\bpm
\star & 0 & \star\\
\star & 0 & \star\\
\star & 0 & \star
\epm
+\mathcal O(1),
\qquad z\to0,
\ee
where $\star$ denotes an unspecified entry. On the other
hand, the expansion of the pure-soliton model near the origin gives
\be \label{E:asyMsoln}
\M_{sol}^{-1}(x,t,z)
=
\bpm
0&0&0\\
\star&\star&\star\\
0&0&0
\epm
+\mathcal O(z),
\qquad z\to0.
\ee
Multiplying~\eqref{E:asyW} and~\eqref{E:asyMsoln}, the possible $z^{-1}$
singularity vanishes.
Therefore $\cE(x,t,z)=\mathcal O(1)$ as $z\to0$.

It remains to consider the branch points. We only discuss $z=q_0$, since
the argument at $z=-q_0$ is identical. Since $\hat q_0=q_0$, the diagonal
matrix $\bP^{-1}$ satisfies
$
(\bP^{-1})_{11}(q_0)=(\bP^{-1})_{33}(q_0).
$
Moreover,  Lemma~\ref{L:M6pmq0} implies
$
\M^{(2)}_1(x,t,q_0)=\ii \M^{(2)}_3(x,t,q_0),
$
and hence
$
\mathbf W_1(x,t,q_0)=\ii \mathbf W_3(x,t,q_0).
$
Writing
$
\mathbf W_3(x,t,q_0)
=
\bigl(\alpha_1(x,t),\alpha_2(x,t),\alpha_3(x,t)\bigr)^{\top}
$,
we obtain the local expansion
\be \label{E:Wasyq0}
\mathbf W(x,t,z)
=
\bpm
\ii\alpha_1 & \star & \alpha_1\\
\ii\alpha_2 & \star & \alpha_2\\
\ii\alpha_3 & \star & \alpha_3
\epm
+\mathcal O(z-q_0),
\qquad z\to q_0.
\ee
On the other hand, the symmetry
$
\M_{sol}(x,t,z)=\M_{sol}(x,t,\hat z)\bPi(z)
$
implies that $\M_{sol}^{-1}$ has the following local structure:
\be \label{E:Msolnq0}
\M_{sol}^{-1}(x,t,z)
=
\frac{1}{z-q_0}
\bpm
\beta_1&\beta_2&\beta_3\\
0&0&0\\
-\ii\beta_1&-\ii\beta_2&-\ii\beta_3
\epm
+\mathcal O(1),
\qquad z\to q_0,
\ee
where $\beta_1,\beta_2,\beta_3$ are functions of $x$ and $t$. By combining~\eqref{E:Wasyq0} and~\eqref{E:Msolnq0}, it follows that $\cE(x,t,z)=\mathcal{O}(1)$ as  $z \to q_0$.
\end{proof}

\begin{figure}
\centering
\begin{tikzpicture}[
    scale=0.9,
    line cap=round,
    line join=round,
    >=latex,
    contour/.style={line width=2pt, draw=black!70},
    arrowcontour/.style={
        contour,
        postaction={decorate},
        decoration={
            markings,
            mark=at position 0.56 with {\arrow{>}}
        }
    },
    dasharrow/.style={
        line width=1.5pt,
        draw=black!70,
        dashed,
        postaction={decorate},
        decoration={
            markings,
            mark=at position 0.56 with {\arrow{>}}
        }
    },
    circlearrow/.style={
        contour,
        postaction={decorate},
        decoration={
            markings,
            mark=at position 0.25 with {\arrow{>}}
        }
    },
    regionlabel/.style={font=\normalsize},
    every node/.style={font=\normalsize}
]

\coordinate (O)  at (-3.20,0);
\coordinate (A)  at (-4.95,1.55);
\coordinate (B)  at (-1.45,1.55);
\coordinate (C)  at (-4.95,-1.55);
\coordinate (D)  at (-1.45,-1.55);

\coordinate (Z1) at (0.45,0);
\coordinate (T)  at (2.30,1.55);
\coordinate (S)  at (2.30,-1.55);
\coordinate (Z0) at (4.15,0);

\coordinate (E)  at (5.25,0.92);
\coordinate (F)  at (5.25,-0.92);

\draw[arrowcontour] (-6.55,1.55) -- (A);
\draw[arrowcontour] (A) -- (B);
\draw[arrowcontour] (B) -- (T);
\draw[arrowcontour] (T) -- (7.55,1.55);

\draw[arrowcontour] (-6.55,-1.55) -- (C);
\draw[arrowcontour] (D) -- (S);
\draw[arrowcontour] (S) -- (7.55,-1.55);

\draw[arrowcontour] (E) -- (7.55,0.92);
\draw[arrowcontour] (F) -- (7.55,-0.92);


\draw[arrowcontour]
    (A) .. controls (-4.95,0.82) and (-4.25,0.14) .. (O);

\draw[arrowcontour]
    (O) .. controls (-2.15,0.14) and (-1.45,0.82) .. (B);

\draw[arrowcontour]
    (C) .. controls (-4.95,-0.82) and (-4.25,-0.14) .. (O);

\draw[arrowcontour]
    (O) .. controls (-2.15,-0.14) and (-1.45,-0.82) .. (D);

\draw[arrowcontour] (B) -- (Z1);
\draw[arrowcontour] (D) -- (Z1);
\draw[arrowcontour] (Z1) -- (T);
\draw[arrowcontour] (Z1) -- (S);

\draw[arrowcontour] (T) -- (Z0);
\draw[arrowcontour] (S) -- (Z0);
\draw[arrowcontour] (Z0) -- (E);
\draw[arrowcontour] (Z0) -- (F);

\fill (O) circle (3.0pt);
\fill (Z1) circle (3.0pt);
\fill (Z0) circle (3.0pt);

\node[below=6pt] at (O) {$0$};
\node[below=6pt] at (Z1) {$z_1$};
\node[below=6pt] at (Z0) {$z_0$};


\def\rsmall{0.72}
\def\gapstart{29}
\def\gapend{90}
\def\arrstart{35}
\def\arrend{90}
\draw[contour]
    ($(Z1)+({\rsmall*cos(\gapend)},{\rsmall*sin(\gapend)})$)
    arc[start angle=\gapend,end angle=360,radius=\rsmall];

\draw[contour]
    ($(Z1)+(\rsmall,0)$)
    arc[start angle=0,end angle=\gapstart,radius=\rsmall];

\draw[->, >=latex, line width=2.0pt, draw=black!70]
    ($(Z1)+({\rsmall*cos(\arrstart)},{\rsmall*sin(\arrstart)})$)
    arc[start angle=\arrstart,end angle=\arrend,radius=\rsmall];
\draw[contour]
    ($(Z0)+({\rsmall*cos(\gapend)},{\rsmall*sin(\gapend)})$)
    arc[start angle=\gapend,end angle=360,radius=\rsmall];

\draw[contour]
    ($(Z0)+(\rsmall,0)$)
    arc[start angle=0,end angle=\gapstart,radius=\rsmall];

\draw[->, >=latex, line width=2.0pt, draw=black!70]
    ($(Z0)+({\rsmall*cos(\arrstart)},{\rsmall*sin(\arrstart)})$)
    arc[start angle=\arrstart,end angle=\arrend,radius=\rsmall];

\end{tikzpicture}

\caption{The jump contour $\Sigma^{\mathcal{E}}$ for $\cE(x,t,z)$.}
\label{fig:contoursigmaE}
\end{figure}

\begin{lemma}\label{L:lsqd}
The function $\cE(x,t,z)$ is analytic at every point of
$
\rZ=\{\zeta_j\}_{j=0}^{N-1}\cup\{\zeta_j^*\}_{j=0}^{N-1}.
$
\end{lemma}

\begin{proof}
Near any point of $\rZ$, we have
$$
\cE(x,t,z)
=
\M^{(2)}(x,t,z)\bP^{-1}(z)
\M_{sol}^{-1}(x,t,z)\M_\infty\bP(\infty).
$$
By construction, $\M^{(2)}\bP^{-1}$ and $\M_{sol}$ satisfy the same residue
conditions at the points of $\rZ$. Hence the possible poles cancel in the
product $\M^{(2)}\bP^{-1}\M_{sol}^{-1}$. The argument is identical to that in Appendix~\ref{AppBproof}. Therefore all singularities of
$\cE$ on $\rZ$ are removable.
\end{proof}

 Let  $\Sigma^{\xcE}=\Sigma^{(2)} \cup \partial \cD$~(see Fig.~\ref{fig:contoursigmaE}),  and define the matrix-valued function $\V^{\xcE}(x,t,z)$ for $z \in \Sigma^{\xcE}$ as follows:
\be \label{E:VEex}
\V^{\xcE}(x,t,z)=\begin{cases}
\M^{out} \V^{(2)}(\M^{out})^{-1},  & z \in \Sigma^{\xcE} \setminus \bar{\cD},\\
\M^{out}(\M^{loc})^{-1}(\M^{out})^{-1}, & z\in \partial \cD,\\
\M^{out}\left( \M^{loc}_- \V^{(7)} (\M^{loc}_+)^{-1} \right)(\M^{out})^{-1}, & z \in  X^{\eps}.
\end{cases}
\ee
A direct verification shows that  $\cE(x,t,z)$  satisfies the following RH problem:
\begin{RHP}\label{rhp:E}
Find a $3 \times 3$ matrix-valued function $\cE(x,t,z)$ with the following properties:
\bi
\item $\cE(x,t,\cdot) : \mathbb{C}\setminus \Sigma^{\xcE}  \to \mathbb{C}^{3 \times 3}$ is analytic.
\item $\cE_+(x,t,z)=\cE_-(x,t,z) \V^{\xcE}(x,t,z), \qquad z \in \Sigma^{\xcE}.$
\item  $\cE(x,t,z)$ admits the following asymptotic behavior
$$  \cE(x,t,z)=\bI+\mathcal{O}(\frac{1}{z}), \qquad z \to \infty.$$

\item   $\cE(x,t,z)=\mathcal{O}(1)$ as $z \to z_{\star}$, where $z_{\star} \in \{0, q_0,-q_0 \} $.
\ei
\end{RHP}

\begin{lemma}\label{L:estWE}
Let $\w^{\xcE}=\V^{\xcE}-\bI$. The following estimates hold uniformly for $t \geq 2$ and $\xi \in \mathcal{I}_+$:
\begin{align}
&\|\w^{\xcE} \|_{ L^{\infty}(\Sigma^{\xcE} \setminus \bar{\cD)}} \leq C t^{-1/2}, \label{E:estwE1bc}\\
&\|\w^{\xcE} \|_{L^1 (\Sigma^{\xcE} \setminus \bar{\cD)}} \leq C t^{-1}, \label{E:estwE1}\\
& \|\w^{\xcE} \|_{L^1 \cap L^{\infty}(\partial \cD)} \leq C t^{-1/2},\label{E:estwE2} \\
& \|\w^{\xcE} \|_{L^1 (X^{\epsilon } )} \leq C t^{-1}\log t,\label{E:estwE3}\\
& \|\w^{\xcE} \|_{L^{\infty} (X^{\epsilon })} \leq C t^{-1/2}\log t. \label{E:estwE4}
\end{align}
\end{lemma}
\begin{proof}
For $z \in \Sigma^{\xcE} \setminus \bar{\cD}$, $\w^{\xcE}= \M^{out} \left( \V^{(2)}-\bI \right) (\M^{out})^{-1}$.
 Then using that $\M^{out}$  and its inverse are uniformly bounded for  $z \in \Sigma^{\xcE}$,  ~\eqref{E:estwE1bc} and~\eqref{E:estwE1} follows from Lemma~\ref{L:yzx}.
For $z \in  \partial \cD$, $\w^{\xcE} = \M^{out}\left( (\M^{loc})^{-1}-\bI \right)(\M^{out})^{-1}$.
Since $\M^{out}$ and its inverse are also bounded for $z \in \Sigma^{\xcE} \cap  \mathcal{\bar{D}}$, the estimates in~\eqref{E:estwE2} follow immediately from~\eqref{E:bzsy}. Finally, when  $z \in X^{\eps}$,
$$\w^{\xcE}=\M^{out}  \M^{loc}_-  \left( \V^{(2)}- \V^{loc} \right)  (\M^{loc}_+)^{-1} (\M^{out})^{-1}.
$$
Thus, by using estimates in~\eqref{E:estVloc} and the boundedness of $\M^{out}$, $(\M^{out})^{- 1}$, $\M^{loc}_{\pm}$ and $(\M_{\pm}^{loc})^{- 1}$, we deduce that~\eqref{E:estwE3} and~\eqref{E:estwE4} hold.
\end{proof}

The estimates in Lemma~\ref{L:estWE} show that
\begin{align*}
\begin{cases}
\|\w^{\xcE} \|_{L^1 (\Sigma^{\xcE})} \leq C t^{-1/2}\\
\|\w^{\xcE} \|_{L^{\infty} (\Sigma^{\xcE})} \leq C t^{-1/2}\log t
\end{cases} \qquad \xi \in \mathcal{I}_+, \ \ t>2.
\end{align*}
Thus by employing  the general inequality $\|f \|_{L^p} \leq \| f\|_{L^1}^{\frac{1}{p}} \|f \|_{L^{\infty}}^{\frac{p-1}{p}}$, we  immediately get
\begin{align}\label{E:wELp}
\|\w^{\xcE} \|_{L^p (\Sigma^{\xcE})} \leq C t^{-1/2}(\log t)^{(p-1)/p}, \qquad  \xi \in \mathcal{I}_+,  \ \ t>2.
\end{align}
For the  contour $\Sigma^{\xcE}$ and a function $\h(z) \in L^2(\Sigma^{\xcE})$,  we define the Cauchy transform $\Ca(\h)(z)$ associated with $\Sigma^{\xcE}$ by
 $$\mathcal{C}(\h)(z) := \frac{1}{2\pi \ii}\int_{\Sigma^{\xcE}}\frac{\h(z')dz'}{z'-z}.$$
It is kown that the left and right non-tangential boundary values $\Ca_+\h$ and $\Ca_-\h$ of $\Ca(\h)$ exist a.e. on $\Sigma^{\xcE}$ and belong to $L^2(\Sigma^{\xcE})$. Let $\mathcal{B}(L^2(\Sigma^{\xcE}))$ denotes the space of bounded linear operators on $L^2(\Sigma^{\xcE})$.  Then $\Ca_{\pm} \in \mathcal{B}(L^2(\Sigma^{\xcE}))$  and $\Ca_+-\Ca_-=I$, where $I$ denotes the identity operator on $L^2(\Sigma^{\xcE})$. Furthermore, we define the operator $\Ca_{\w^{\xcE}}$ by $\Ca_{\w^{\xcE}}(\h)=\Ca_-(\h \w^{\xcE})$. From the preceding analysis, we have known that $\|\w^{\xcE}\|_{L^2(\Sigma^{\xcE})} \to 0$ as $t \to \infty$. Consequently, there exists a $T>0$ such that the operator $I-\Ca_{\w^{\xcE}}$ is invertible whenever $t>T$ and $\xi \in \mathcal{I}_+$. Therefore, we can  define a function $\bu^{\xcE}(x,t,z)$ for $z \in \Sigma^{\xcE}$ and $t>T$ by
\begin{align}\label{E:uE}
\bu^{\xcE}=\bI+(I-\Ca_{\w^{\xcE}})^{-1}\Ca_{\w^{\xcE}}\bI \ \in \bI + L^2(\Sigma^{\xcE}).
\end{align}
According to the standard theory of RH problems~\cite{Lenells2017,TS2016},  $\cE(x,t,z)$ can be expressed as
\be \label{E:Ejfbs}
\cE(x,t,z)=\bI+\frac{1}{2\pi \ii}\int_{\Sigma^{\xcE}}\frac{\bu^{\xcE}(z') \w^{\xcE} (z') }{z'-z} \mathrm  d z', \qquad z \in \C \setminus \Sigma^{\xcE}.
\ee
The following nontangential limit is well-defined:
\begin{align}\label{E:L}
\bL(x,t):=\lim^{\angle}_{z \to \infty}z(\cE(x,t,z)-\bI)=-\frac{1}{2\pi \ii}\int_{\Sigma^{\xcE}}\bu^{\xcE} (x,t,\zeta)\w^{\xcE}(x,t,\zeta) \mathrm{d}\zeta.
\end{align}
\begin{lemma}\label{L:estofL}
As $t \to \infty$,
\begin{align}\label{E:estofL-1}
\bL(x,t)=-\frac{1}{2\pi \ii}\int_{\partial \mathcal{D}} \w^{\xcE}(x,t,\zeta) \mathrm{d}\zeta+\mathcal{O}(t^{-1}\log t).
\end{align}
\end{lemma}
\begin{proof}
The function $\bL(x,t)$ can be rewritten as
$$
\bL(x,t) = -\frac{1}{2\pi \ii}\int_{\partial \mathcal{D}} \w^{\xcE}(x,t,z) \mathrm d z + \bL_1(x,t) + \bL_2(x,t),
$$
where
\begin{align*}
\bL_1(x,t) = -\frac{1}{2\pi \ii}\int_{\Sigma^{\xcE}\setminus \partial \mathcal{D}} \w^{\xcE} (x,t,z) \mathrm{d} z , \qquad
 \bL_2(x,t) = -\frac{1}{2\pi \ii}\int_{\Sigma^{\xcE}} (\bu^{\xcE}(x,t,z)-\bI) \w^{\xcE}(x,t,z) \mathrm{d} z.
\end{align*}
From \eqref{E:wELp} and \eqref{E:uE}, it follows that
\begin{equation}\label{E:yxjs11}
\begin{aligned}
\|\bu^{\xcE} - \bI\|_{L^2(\Sigma^{\xcE})}&\leq \|(I-\Ca_{\w^{\xcE}})^{-1}\Ca_{\w^{\xcE}}\bI \|_{L^2(\Sigma^{\xcE})}
\leq \sum_{j=0}^{\infty}\| \Ca_{\w^{\xcE}}\|^{j}_{\mathcal{B}(L^2(\Sigma^{\xcE}))}\|\Ca_{\w^{\xcE}} \bI \|_{L^2(\Sigma^{\xcE})}\\
&\leq  \frac{\|\Ca_- \|_{\mathcal{B}(L^2(\Sigma^{\xcE}))} \| \w^{\xcE}\|_{L^2(\Sigma^{\xcE})}}{1-\|\Ca_- \|_{\mathcal{B}(L^2(\Sigma^{\xcE}))} \|\w^{\xcE} \|_{L^{\infty}(\Sigma^{\xcE})}}
\leq Ct^{-1/2}(\log t)^{1/2}.
\end{aligned}
\end{equation}
Then the lemma follows from Lemma~\ref{L:estWE} and Eq.~\eqref{E:yxjs11} and straightforward estimates.
\end{proof}
We define the functions $F^{(1)}(x,t)$ and $F^{(2)}(x,t)$ by
\be \label{E:F12}
F^{(1)}(x,t)=-\frac{1}{2\pi \ii} \int_{\partial D_{\eps}(z_0)} \w^{\xcE}(x,t,\zeta) \mathrm{d} \zeta, \quad
F^{(2)}(x,t)=-\frac{1}{2\pi \ii} \int_{\partial D_{\eps}(z_1)} \w^{\xcE}(x,t,\zeta) \mathrm{d} \zeta.
\ee 
For $z\in\partial\mathcal D$, we have
$$
\w^{\xcE}(x,t,z)
=
\M^{out}(x,t,z)(\M^{loc})^{-1}(x,t,z)
(\M^{out})^{-1}(x,t,z)-\bI.
$$
Using Eq.~\eqref{E:bzsy} and  applying the residue theorem, we obtain
\begin{align*}
F^{(1)}(x,t)
&=\frac{q_0^2}{\xi^2 \sqrt{2t}} \bP^{-1}(\infty)\M_{\infty}^{-1} \M_{sol}(x,t,z_0) \bZ_R(\xi,t) \M^{-1}_{sol}(x,t,z_0) \M_{\infty} \bP(\infty)+\mathcal{O}(t^{-1}), \\
F^{(2)}(x,t)
&=\frac{1}{\sqrt{2t}} \bP^{-1}(\infty) \M_{\infty}^{-1}  \M_{sol}(x,t,z_1) \bZ_L(\xi,t) \M^{-1}_{sol}(x,t,z_1) \M_{\infty} \bP(\infty)+\mathcal{O}(t^{-1}), \qquad t \to \infty.
\end{align*}
In the expression above,
\be \label{E:ZRZL}
\bZ_R=
{\scriptsize
\bpm
0&0&0\\
0 &0 &-\ii \beta_{23} (\tilde{d}_0^{r} )^{-1}\e^{\theta_{23}(z_0)}\\
0&\ii \beta_{32}\tilde{d}_0^r\e^{\theta_{32}(z_0)}&0
\epm},\
{\scriptsize
\bZ_L= \bpm
0&\ii \beta_{12} (\tilde{d}_0^{\ell})^{-1} \e^{\theta_{12}(z_1)}&0\\
-\ii \beta_{21} \tilde{d}^{\ell}_0 \e^{\theta_{21}(z_1)}&0&0\\
0& 0&0
\epm},
\ee
where
\be \label{E:d0til}
\begin{aligned}
&\tilde{d}_0^{r}= d_0^{r} \cP^2_1(z_1) \cP_1(z_0)=(\frac{q_0^2}{\xi^2 \sqrt{2t}})^{-2 \ii  \nu} \e^{-2 \chi(z_0)} \frac{\delta(z_1)}{ \rho(z_0)\rho^2(z_1)},\\
&\tilde{d}^{\ell}_0= \frac{d_0^{\ell}}{ \cP_1^2(z_0) \cP_1(z_1)}=(\sqrt{2t})^{-2 \ii \tilde{\nu}}\e^{-2 \tilde{\chi} (z_1)}  \frac{ \rho^2(z_1) \rho(z_0)}{\delta^2(0)\delta(z_1)}.
\end{aligned}
\ee
 Thus by~\eqref{E:estofL-1} we obtain
\be \label{E:estofL}
\begin{aligned}
\bL(x,t)=&\frac{q_0^2}{\xi^2 \sqrt{2t}} \bP^{-1}(\infty) \M_{\infty}^{-1}  \M_{sol}(x,t,z_0) \bZ_R(\xi,t) \M^{-1}_{sol}(x,t,z_0) \M_{\infty}\bP(\infty)\\
+&\frac{1}{\sqrt{2t}} \bP^{-1}(\infty) \M_{\infty}^{-1}   \M_{sol}(x,t,z_1) \bZ_L(\xi,t) \M^{-1}_{sol}(x,t,z_1) \M_{\infty}\bP(\infty)
+\mathcal{O}(t^{-1} \log t), \quad t \to \infty.
\end{aligned}
\ee

\subsection{Proof of the asymptotic formulas~\eqref{E:asyofq} and~\eqref{E:bjgs}}
Taking into account all the transformations we have performed,  we obtain
\begin{align}\label{E:alltran}
\cE(x,t,z)=\widetilde{\bT}^{-1}(\infty) \M_{\infty}^{-1}  \M(x,t,z) \widetilde{\bT}(z)   (\M_{sol}\bP)^{-1}(x,t,z) \M_{\infty} \bP(\infty),
\end{align}
where we have selected $z$ such that $\G(z)=\bI$.
Note that this can be done as long as $z$ lies outside the region $\mathcal{R}$. Hence, for such $z$ we have
\be \label{E:cz}
\M(x,t,z)=\M_{\infty} \widetilde{\bT}(\infty)  \cE(x,t,z) \bP^{-1}(\infty)\M_{\infty}^{-1}  \M_{sol}(x,t,z) \tilde{\Del}^{-1}(z) \Del^{-1}(z).
\ee
From~\eqref{E:cggs} we have
\be \label{E:fcggs}
\mathbf{q}(x,t)=- \ii \lim_{z \to \infty}z \m_{rc}(x,t,z), \qquad \m_{rc}=(\M_{21}, \M_{31})^{\top}.
\ee
Therefore, we need to examine the asymptotic behavior of the functions on the right-hand side of Eq.~\eqref{E:cz} as $z \to \infty$. We can easily obtain that as $z \to \infty$,
\begin{align*}
&\cE(x,t,z)=\bI+\frac{1}{z} \bL(x,t)+\mathcal{O}(\frac{1}{z^2}),\\
&\M_{sol}(x,t,z)=\M_{\infty}+\frac{1}{z}\M^{(1)}_{sol}(x,t)+\mathcal{O}(\frac{1}{z^2}),\\
&\tilde{\Del}^{-1}(z) \Del^{-1}(z)=(\tilde{\Del}^{\infty})^{-1} (\Del^{\infty})^{-1}+\frac{1}{z} \bC + \mathcal{O}(\frac{1}{z^2}),
\end{align*}
where $\bC$  is a diagonal matrix.  Then, by substituting the above asymptotic expansions into~\eqref{E:cz}, a direct calculation yields
\begin{equation} \label{E:mgs}
\m_{rc}(x,t,z)=\frac{1}{z} \check{\M}_\infty \boldsymbol{\sigma}\check{\M}_\infty^{-1} \bpm
 (\M^{(1)}_{sol})_{21} \\
   (\M^{(1)}_{sol})_{31}
\epm +
\frac{1}{z} \check{\M}_\infty  \boldsymbol{\sigma} \check{\bP}(\infty)
\bpm
 \bL_{21}\\
\bL_{31}
\epm
+\mathcal{O}(\frac{1}{z^2}),  \qquad z \to \infty,
\end{equation}
 where
\be \label{E:cheMandsi}
\check{\M}_\infty=\begin{pmatrix}
\frac{q_{2,+}^*}{q_0}  & \frac{q_{1,+}}{q_0}\\
 -\frac{q_{1,+}^*}{q_0} &\frac{q_{2,+}}{q_0}
\end{pmatrix}, \quad
\boldsymbol{\sigma}=\begin{pmatrix}
\frac{ \delta^2(0)}{\delta_1(0)}  & \\
  & \delta_1(0) \delta(0)
\end{pmatrix},\quad
\check{\bP}(\infty)=
\bpm 
\frac{1}{\cP_1(0)} & \\
 & \cP_1(0)
\epm.
\ee
Combining this with reconstruction formula~\eqref{E:fcggs}, we have
\be \label{E:zzz}
\begin{aligned}
\mathbf{q}(x,t)&= \check{\M}_\infty \boldsymbol{\sigma}\check{\M}_\infty^{-1} \bpm
 -\ii (\M^{(1)}_{sol})_{21} \\
-\ii  (\M^{(1)}_{sol})_{31}
\epm +
\check{\M}_\infty \boldsymbol{\sigma} \check{\bP}(\infty)
\bpm
- \ii  \bL_{21}\\
-\ii  \bL_{31}
\epm\\
&=\check{\M}_\infty \boldsymbol{\sigma}\check{\M}_\infty^{-1} \biggl(  
\q_{sol}^N(x,t)+\q_{rad}(x,t)
\biggr),
\end{aligned}
\ee
where $\q^N_{sol}$ is the $N$-dark-soliton solution appearing in
Section~\ref{sub:lG}; see formula~\eqref{E:Ndarksoliton}, and
\be \label{E:qradddyhb}
\q_{rad}(x,t)=\check{\M}_\infty  \check{\bP}(\infty) \bpm
- \ii  \bL_{21}(x,t)\\
-\ii  \bL_{31} (x,t)
\epm.
\ee
Note that  $\check{\M}_\infty$, $\boldsymbol{\sigma}$ and $\check{\M}_\infty^{-1}$ are all unitary matrices, $\check{\M}_\infty \boldsymbol{\sigma}\check{\M}_\infty^{-1}$ is also a unitary matrix. Since unitary matrices are norm-preserving, one can directly verify that $\q^{N}_{msol}(x,t)$, defined by
\be \label{E:msol}
\mathbf{q}^{N}_{msol}(x,t)= \check{\M}_\infty \boldsymbol{\sigma}\check{\M}_\infty^{-1} \bpm
 -\ii (\M^{(1)}_{sol})_{21} \\
-\ii  (\M^{(1)}_{sol})_{31}
\epm ,
\ee
remains a solution to the defocusing Manakov system~\eqref{E:demanakovS}. 

Now we examine the second term on the right-hand side of Eq.~\eqref{E:zzz}. 
From~\eqref{E:estofL} and~\eqref{E:qradddyhb}, as $t \to \infty$, one can obtain
\be \label{E:zzhsydl}
\q_{rad}(x,t)=\frac{1}{\sqrt{t}}
\bpm
L_A(x,t)\\
L_B(x,t)
\epm
+\mathcal{O}(t^{-1} \log t),
\ee
where
\begin{align}
L_A(x,t)&=\frac{q_0^2}{\xi^2 \sqrt{2}} \bigg[ \beta_{32}\tilde{d}_0^r \e^{\theta_{32}(x,t,z_0)}(\M_{sol})_{23}(x,t,z_0) (\M_{sol}^{-1})_{21}(x,t,z_0)    \label{E:bLAB}\\
&- \beta_{23}(\tilde{d}_0^r)^{-1} \e^{\theta_{23}(x,t,z_0)}(\M_{sol})_{22}(x,t,z_0) (\M_{sol}^{-1})_{31}(x,t,z_0) \bigg] \nonumber \\
&+ \frac{1}{\sqrt{2}}  \bigg[- \beta_{21}\tilde{d}_0^{\ell} \e^{\theta_{21}(x,t,z_1)}(\M_{sol})_{22}(x,t,z_1) (\M_{sol}^{-1})_{11}(x,t,z_1) \nonumber\\
&+ \beta_{12}(\tilde{d}_0^{\ell})^{-1} \e^{\theta_{12}(x,t,z_1)}(\M_{sol})_{21}(x,t,z_1) (\M_{sol}^{-1})_{21}(x,t,z_1) \bigg], \nonumber \\
L_{B}(x,t)&=\frac{q_0^2}{\xi^2 \sqrt{2}} \bigg[  \beta_{32}\tilde{d}_0^r \e^{\theta_{32}(x,t,z_0)}(\M_{sol})_{33}(x,t,z_0) (\M_{sol}^{-1})_{21}(x,t,z_0)  \label{E:bLABhb} \\
&- \beta_{23}(\tilde{d}_0^r)^{-1} \e^{\theta_{23}(x,t,z_0)}(\M_{sol})_{32}(x,t,z_0) (\M_{sol}^{-1})_{31}(x,t,z_0) \bigg] \nonumber \\
&+ \frac{1}{\sqrt{2}}  \bigg[- \beta_{21}\tilde{d}_0^{\ell} \e^{\theta_{21}(x,t,z_1)}(\M_{sol})_{32}(x,t,z_1) (\M_{sol}^{-1})_{11}(x,t,z_1) \nonumber \\
&+\beta_{12}(\tilde{d}_0^{\ell})^{-1} \e^{\theta_{12}(x,t,z_1)}(\M_{sol})_{31}(x,t,z_1) (\M_{sol}^{-1})_{21}(x,t,z_1) \bigg]. \nonumber
\end{align}
According to the computations in Appendix~\ref{AppB}, $\M_{sol}$ can be solved
explicitly by solving certain linear algebraic systems. Hence, after a
lengthy but direct calculation, one obtains
$$
q_{2,+} \biggl((\M_{ sol})_{21}(x,t,z),(\M_{ sol})_{23}(x,t,z)\biggr)
=
q_{1,+}
\biggl((\M_{ sol})_{31}(x,t,z),(\M_{ sol})_{33}(x,t,z)\biggr).
$$
Therefore, we have $(\M_{sol}^{-1})_{21}(x,t,z_0)=(\M_{sol}^{-1})_{21}(x,t,z_1)=0$.
Moreover, it follows from~\eqref{E:Msolderlie} that
$$
(\M_{sol})_{22}(x,t,z_1)=-\frac{q_{2,+}^*}{q_0},\qquad
(\M_{ sol})_{32}(x,t,z_1)=\frac{q_{1,+}^*}{q_0}.
$$
Using the symmetry relation
$
\M_{sol}(x,t,z_0)=\M_{sol}(x,t,z_1)\bPi(z_0)
$ and simplifying the resulting expression, we obtain
\be \label{E:LABdhj}
\begin{aligned}
\bpm
L_A\\
L_B
\epm
&=
-\frac{1}{\sqrt{2}} \biggl(  
\beta_{21} \tilde{d}_0^{\ell}+\beta_{23}(\ii \frac{q_0^3}{\xi^3} ) (\tilde{d}_0^{r})^{-1}
\biggr)\e^{\theta_{21}(x,t,z_1)} (\M_{sol}^{-1})_{11}(x,t,z_1) \bpm
(\M_{sol})_{22}(x,t,z_1)\\
(\M_{sol})_{32}(x,t,z_1)
\epm\\
&=-\frac{\sqrt{2 \pi} \gamma^2 (z_1) \e^{-\frac{\pi \nu}{2}+ \frac{\pi}{4} \ii}  \rho^2(z_1) \rho(z_0)}{\hat{r}_1^*(z_1) \Gamma(-\ii \nu) \delta(z_1)} (\frac{q_0^2}{\xi^2 \sqrt{2t}})^{2 \ii \nu} \e^{2 \chi(z_0)}  \e^{\theta_{21}(x,t,z_1)} (\M_{sol}^{-1})_{11}(x,t,z_1)
\bpm
-\frac{q_{2,+}^*}{q_0}\\
\frac{q_{1,+}^*}{q_0}
\epm.
\end{aligned}
\ee
Moreover,
$$
\check{\M}_\infty \boldsymbol{\sigma}\check{\M}_\infty^{-1} \bpm
-\frac{q_{2,+}^*}{q_0}\\
\frac{q_{1,+}^*}{q_0}
\epm=\frac{ \delta^2(0)}{\delta_1(0)}  \bpm
-\frac{q_{2,+}^*}{q_0}\\
\frac{q_{1,+}^*}{q_0}
\epm.
$$
It then follows from~\eqref{E:zzz} that, uniformly for $\xi\in\mathcal I_+$ as
$t\to\infty$,
\be \label{E:asjjgdddgs}
\q(x,t)=\q^N_{msol}(x,t)+\frac{f_{rad}(x,t)}{\sqrt{t}}(\M_{sol}^{-1})_{11}(x,t,z_1)\bpm
-\frac{q_{2,+}^*}{q_0}\\
\frac{q_{1,+}^*}{q_0}
\epm+\mathcal{O}(t^{-1} \log t),
\ee
where
$$
f_{rad}(x,t)=-\frac{\sqrt{2\pi} \gamma^2 (z_1) \e^{-\frac{\pi \nu}{2}+ \frac{\pi}{4} \ii}  \rho^2(z_1) \rho(z_0)}{\hat{r}_1^*(z_1) \Gamma(-\ii \nu) \delta(z_1)} (\frac{q_0^2}{\xi^2 \sqrt{2t}})^{2 \ii \nu} \e^{2 \chi(z_0)}\e^{\theta_{21}(x,t,z_1)} \frac{ \delta^2(0)}{\delta_1(0)} .
$$
We next show that, in the cases described in Theorems~\ref{asy-th-1} and~\ref{Th:2m}, the
asymptotic formula~\eqref{E:asjjgdddgs} reduces to
\eqref{E:asyofq} and~\eqref{E:bjgs}, respectively.

Assume first that $\mathcal I_+$ contains exactly one soliton velocity,
say $\Re \zeta_{j_0} \in \mathcal{I}_+$. By~\eqref{E:qsolNbiby1sol}, the $N$-soliton part is
approximated by the one-soliton profile $\q_{sol}^{(j_0)}(x,t)$. Hence
the leading term $\q_{msol}^N(x,t)$ satisfies
\begin{align*} \q_{msol}^N(x,t) =& \check{\M}_\infty \boldsymbol{\sigma}\check{\M}_\infty^{-1} \q_{sol}^{(j_0)}+\mathcal O(\e^{-ct})\\ =&\delta_1(0) \delta(0) \left(\prod_{\ell<j_0}\frac{\zeta_{\ell}}{\zeta_{\ell}^*} \right) \q_+\e^{\ii\theta_{j_0}} \left[ \cos\theta_{j_0} -\ii\sin\theta_{j_0} \tanh\Bigl( q_0\sin\theta_{j_0} \bigl(x-2q_0\cos\theta_{j_0}t-x_{j_0}\bigr) \Bigr) \right]\\ &+\mathcal O(\e^{-ct}), \qquad t \to \infty. \end{align*}
Moreover, \eqref{E:nashhl} implies that $\M_{sol}$ is approximated by
the one-soliton model $\M_{sol}^{(j_0)}$.  More
precisely,
\[
\bigl(\M_{sol}^{-1}\bigr)_{11}(x,t,z_1)
=
\left(
\prod_{j<j_0}
\frac{z_1-\zeta_j}{z_1-\zeta_j^*}
\right)
\bigl((\M_{sol}^{(j_0)})^{-1}\bigr)_{11}(x,t,z_1)
+\mathcal O(\e^{-ct}),
\qquad t\to\infty.
\]
Substituting these relations into~\eqref{E:asjjgdddgs}, we obtain
\begin{align*}
\q(x,t)
={}&
\delta_1(0)\delta(0)
\left(
\prod_{\ell<j_0}\frac{\zeta_\ell}{\zeta_\ell^*}
\right)
\q_{sol}^{(j_0)}(x,t)
\\
&+
\left[f_{rad}(x,t) \cdot
\biggl(
\prod_{j<j_0}
\frac{z_1-\zeta_j}{z_1-\zeta_j^*}
\biggr)
\right] \frac{1}{\sqrt{t}}
\bigl((\M_{sol}^{(j_0)})^{-1}\bigr)_{11}(x,t,z_1)
\bpm
-\dfrac{q_{2,+}^*}{q_0}\\[1mm]
\dfrac{q_{1,+}^*}{q_0}
\epm
+\mathcal O(t^{-1}\log t).
\end{align*}
The one-soliton prefactor can be expressed in terms of $\alpha(\xi)$, while
the product of the first two factors in the $t^{-1/2}$ coefficient can be
written as the modulated linear wave $A(\xi)\e^{\ii\Phi(x,t)}$. This gives
\eqref{E:asyofq}.

We now assume that $\mathcal I_+$ contains no soliton velocities. Then
\eqref{E:wuguziqingjnjx} and~\eqref{E:choux} imply that, in
\eqref{E:asjjgdddgs}, one should replace $\q_{msol}^N(x,t)$ by
$
\delta_1(0)\delta(0)
\biggl(
\prod_{j\in\nabla^+}\frac{\zeta_j}{\zeta_j^*}
\biggr)\q_+
$,
and replace $\bigl(\M_{sol}^{-1}\bigr)_{11}(x,t,z_1)$ by
\[
\biggl(
\prod_{j\in\nabla^+}
\frac{z_1-\zeta_j}{z_1-\zeta_j^*}
\biggr)
\frac{1}{\gamma(z_1)}.
\]
Consequently, the coefficient of the leading term is again expressed
through $\alpha(\xi)$, and the $t^{-1/2}$ coefficient reduces to the
corresponding modulated linear wave. This yields~\eqref{E:bjgs}.


\section{ Concluding remarks}
 In this work, we apply the Deift--Zhou nonlinear steepest descent method to the
defocusing Manakov system on a nonzero background under
Assumptions~\ref{As:1}. Compared with the scalar defocusing NLS equation with
similar NZBCs, the vector case exhibits an additional
\(\mathcal O(t^{-1/2})\) dispersive correction in the soliton region, revealing a
genuinely vectorial asymptotic effect.
Although the assumptions imposed on the initial data are relatively strong, we
expect that the results can be extended to weaker Sobolev-type settings. In such
a framework, the \(\bar\partial\)-steepest descent method would be more suitable,
as in the scalar defocusing NLS case~\cite{CuJe2016,WF2023,WF2022}.   Such an extension, however, would require a more detailed investigation.   Nevertheless, this work reveals the characteristics of solutions for the defocusing Manakov system with NZBCs~\eqref{E:bjtj} and represents the first progress in understanding the long-time behavior of solutions to vector NLS equations with NZBCs. 

 The results of this work also open up a number of interesting  directions:
$\mathrm{(1)}$ It remains to analyze the defocusing Manakov system \eqref{E:demanakovS} in the remaining space-time regions, namely the solitonless region $|\xi|>q_0$ and the transition region $|\xi|\approx q_0$. For the scalar defocusing NLS equation, these regimes have been studied in~\cite{WF2022,WF2023}. In particular, Wang and Fan~\cite{WF2023} showed that the leading term in the transition region is described by a solution of the Painlev\'e II equation. One may therefore expect that, in the present vector setting, the transition asymptotics should be governed by a Painlev\'e-type model.
$\mathrm{(2)}$ A natural direction is to extend the present analysis to more
general multi-component NLS systems with NZBCs, including the
$N$-component defocusing NLS system with $N\geq3$ and the focusing multi-component NLS system.  
$\mathrm{(3)}$ It would also be interesting to investigate the defocusing Manakov system under non-parallel or asymmetric boundary conditions. The IST for the former case has been developed in~\cite{ABP2022}, but the corresponding long-time asymptotic behavior remains open. The analysis in the present paper may provide a useful starting point for this problem.
$\mathrm{(4)}$ The square matrix Schr\"odinger system with NZBCs provides another related direction~\cite{PFLP2018}. Since its IST can be viewed as a more direct matrix generalization of the scalar theory, its long-time asymptotic analysis may be technically more accessible than that of the vector Manakov system. Nevertheless, because such systems arise in several physical contexts, their asymptotic behavior remains worth studying. $\mathrm{(5)}$ A numerical inverse scattering theory for the defocusing Manakov system with NZBCs is also desirable. The scalar case has recently been investigated in~\cite{GPT2025}, but extending numerical inverse scattering techniques to the vector case is expected to be substantially more challenging.
$\mathrm{(6)}$ Finally, the methods developed here may be applicable to other coupled integrable systems with nonzero background, such as coupled mKdV equations~\cite{XFL2023}, coupled Hirota equations~\cite{HMYZ2025}, and coupled Gerdjikov--Ivanov equations~\cite{MZ-2023}. Although the IST formulation for these systems has been established, their
long-time asymptotic analysis remains to be explored.

We hope that the results of this work  will motivate further investigations on the related problems.

\appendix
	
\section{Proofs of some results from section~\ref{s:rhch}}\label{App:AAA}

\subsection{Proof of ~\eqref{E:Azinfty} and~\eqref{E:Az0}}\label{App:pofzas}
We prove the estimates for \(a_{21}(z)\) and \(a_{23}(z)\). The remaining
entries are treated in the same way. Set
$$
\Delta \bQ_{\pm}(x,0):= \bQ(x,0)-\bQ_{\pm}.
$$
We first justify that the large-\(z\) expansions of the relevant columns of
\(\boldsymbol{\mu}_-\) remain valid in the strip-like domain \(S_\varepsilon\). Recall that
\[
S_\varepsilon=\{z\in\mathbb C:\ |\Im z|\le \varepsilon\}\setminus(B_1\cup B_2),
\]
where \(B_1\) and \(B_2\) are the disks centered at
\(\ii q_0^2/(2\varepsilon)\) and \(-\ii q_0^2/(2\varepsilon)\), respectively,
with radius \(q_0^2/(2\varepsilon)\). Hence, for \(z\in S_\varepsilon\),
\[
|\Im z|\le \varepsilon,\qquad
\left|\Im\frac{q_0^2}{z}\right|\le \varepsilon,
\]
and therefore
\[
|\Im\lambda(z)|
=\frac12\left|\Im z-\Im\frac{q_0^2}{z}\right|
\le \varepsilon .
\]
Consider the Volterra equation for the first column of $ \boldsymbol{\mu}_{-1}$. Writing
$$
\omega(x,z):=\E_{-}^{-1}(z) (x,0,z) \boldsymbol{\mu}_{-1}(x,0,z).
$$
one has
$$
\omega(x,z)
=
\begin{pmatrix}1\\0\\0\end{pmatrix}
+
\int_{-\infty}^{x}
\widetilde \G(x-y,z)\Delta \bQ_{-}(y,0)\E_{-}(z)\omega (y,z)\,\mathrm{d}y ,
$$
where
$$
\widetilde \G(x-y,z)
=
\operatorname{diag}
\left(
1,\,
\e^{\ii z(x-y)},\,
\e^{2 \ii \lambda(z)(x-y) }
\right)
\E_{-}^{-1}(z).
$$
Since \(2|\Im\lambda(z)|\le 2\varepsilon\) for \(z\in S_\varepsilon\), the
exponential factors in the kernel are controlled by the assumed exponential
decay of $\Delta \bQ_{-}$. Thus the corresponding Neumann series converges
uniformly on compact subsets of $S_\varepsilon\setminus\{\pm q_0\}$. It follows
that \(\mu_{-,1}(x,0,z)\) is analytic in \(S_\varepsilon\setminus\{\pm q_0\}\).
Consequently, by the preceding analyticity argument, the large-\(z\) expansion of $\boldsymbol{\mu}_{-1}(x,0,z)$ given in  Ref.~\cite[Corollary~2.28]{BD2015-1} remains valid as \(z\to\infty\) within
\(S_\varepsilon\):
\be \label{E:muinfty}
 \boldsymbol{\mu}_{-1}(x,0,z)= \bpm 1\\0\\0  \epm + \frac{\ii}{z} \bpm 0\\ q_1(x,0) \\ q_2(x,0)   \epm   +\mathcal{O}(\frac{1}{z^2}).
\ee
The same argument applies to $\boldsymbol{\mu}_{-3}$, and gives
\be \label{E:muinfty3}
  \boldsymbol{\mu}_{-3}(x,0,z)=\bpm 0\\ \frac{\q_-}{q_0} \epm + \frac{-\ii}{q_0 z} \bpm \q^{\dagger} \cdot \q_- \\ 0 \\ 0   \epm    +\mathcal{O}(\frac{1}{z^2}),  \qquad  z \to \infty,
\quad z \in  S_{\varepsilon }.
\ee

We now estimate \(a_{21}\). By \cite[Eq.~(A.20)]{BD2015-1}, the scattering matrix admits the representation
\be \label{E:Ajfbs}
\begin{aligned}
\A(z)=&\int_{0}^{\infty}\e^{-\ii y  \blambda(z)} \E_+^{-1}(z) [\bQ(y,0)- \bQ_{+}] \boldsymbol{\mu}_-(y,0,z)
\e^{\ii y  \blambda(z)} \mathrm{d}y \\
&+\E_+^{-1}(z)\E_-(z)\left[\bI+  \int_{-\infty}^{0}\e^{-\ii y  \blambda(z)} \E_-^{-1}(z) [\bQ(y,0)- \bQ_{-}] \boldsymbol{\mu}_-(y,0,z)
\e^{\ii y  \blambda(z)}  \mathrm{d}y  \right].
\end{aligned}
\ee
Taking the $(2,1)$-entry yields
\be \label{E:a21jfbs}
\begin{aligned}
a_{21}(z)=&\int_{0}^{\infty}
\begin{pmatrix}
 0 & (\q_+^{\perp })^{\dagger}/q_0
\end{pmatrix}
 \Delta \bQ_{+}(y,0)\boldsymbol{\mu}_{-1}(y,0,z)
\e^{-\ii y  (\lambda+k)(z)} \mathrm{d}y \\
&+\e^{\ii (\theta_+-\theta_-)}\left[  \int_{-\infty}^{0}\begin{pmatrix}
 0 & (\q_-^{\perp })^{\dagger}/q_0
\end{pmatrix}  \Delta \bQ_{-}(y,0) \boldsymbol{\mu}_{-1}(y,0,z)
\e^{-\ii y  (\lambda+k)(z)} \mathrm{d}y   \right].
\end{aligned}
\ee
Substituting the expansion~\eqref{E:muinfty} of \(\boldsymbol{\mu}_{-,1}\) into this formula gives
\be \label{E:hlbidejf}
\begin{aligned}
a_{21}(z)=&
\int_0^{+\infty}
\frac{(\q_+^\perp)^\dagger}{q_0}
\bigl(\q(y,0)-\q_+\bigr)\e^{-\ii yz}\,\mathrm d y\\
&+\e^{\ii(\theta_+-\theta_-)}
\int_{-\infty}^0
\frac{(\q_-^\perp)^\dagger}{q_0}
\bigl(\q(y,0)-\q_-\bigr)\e^{-\ii yz}\,\mathrm d y
+\mathcal O(z^{-2}).
\end{aligned}
\ee
Here we used
$$
\bpm 0&  \frac{(\q_\pm^\perp)^\dagger}{q_0} \epm
\Delta \bQ_\pm(y,0)
\begin{pmatrix}0\\ q_1(y,0)\\ q_2(y,0)\end{pmatrix}=0.
$$
Since \(\q_+=\e^{\ii(\theta_+-\theta_-)}\q_-\), we have
\[
\e^{\ii(\theta_+-\theta_-)}(\q_-^\perp)^\dagger=(\q_+^\perp)^\dagger .
\]
Moreover, \((\q_\pm^\perp)^\dagger\q_\pm=0\). Hence the two integrals in~\eqref{E:hlbidejf} combine
into
$$
a_{21}(z)=
\int_{-\infty}^{+\infty}
F(y)\e^{-\ii yz}\,\mathrm d y+\mathcal O(z^{-2}),
\qquad
F(y):=\frac{(\q_+^\perp)^\dagger\q(y,0)}{q_0}.
$$
By the boundary condition, \(F(y)\to0\) as \(y\to\pm\infty\). Under the
assumed exponential decay of the initial data and its derivatives, \(F,F'\) and \(F''\)
are exponentially integrable. Therefore, for \(z\in S_\varepsilon\), two integrations
by parts give
\[
\int_{-\infty}^{+\infty}F(y)\e^{-\ii yz}\,\mathrm d y
=
-\frac{1}{z^2}
\int_{-\infty}^{+\infty}F''(y)\e^{-\ii yz}\,\mathrm d y
=\mathcal O(z^{-2}).
\]
Consequently,
$
a_{21}(z)=\mathcal O(z^{-2}),\qquad z\to\infty,\quad z\in S_\varepsilon .
$

Next we estimate \(a_{23}(z)\). Using the large-\(z\) expansion~\eqref{E:muinfty3} of \(\boldsymbol{\mu}_3\) in the corresponding integral representation of \(a_{23}\), the leading term vanishes
because
$$
\begin{pmatrix}
0&(\q_{\pm}^{\perp })^{\dagger}/q_0
\end{pmatrix}
 [\bQ(y,0)- \bQ_{\pm}]
\bpm 0\\ \frac{\q_-}{q_0} \epm =0.
$$
The next term is of order \(z^{-1}\). Hence
$
a_{23}(z)=\mathcal O(z^{-1}),\quad S_\varepsilon  \ni z\to\infty.
$

Finally, the behavior as \(z\to0\) follows from the symmetry~\eqref{E:ABdc}.
For instance, we have 
$
a_{23}(z)=-\frac{\ii q_0}{z}a_{21}(\hat z)
$.
Since $S_\varepsilon \ni \hat z\to\infty$ as $S_\varepsilon \ni z\to0$, we obtain  $a_{21}(\hat z)=\mathcal O(\hat z^{-2})$. Thus we get
$$
a_{23}(z)=\mathcal{O}(z),\qquad z\to0,\quad z\in S_\varepsilon .
$$
 The remaining estimates in~\eqref{E:Az0} can be proved in a similar manner.

\subsection{The derivation of~\eqref{E:Vex}}\label{Apppof258}
In Ref.~\cite{BD2015-1}, the reflection coefficients $\rho_1(z)$ and $\rho_2(z)$ are defined as:
\be \label{E:fsxs}
\rho_1(z)=\frac{b_{13}(z)}{b_{11}(z)}, \qquad
\rho_2(z)=\frac{a_{21}(z)}{a_{11}(z)}.
\ee
As shown in Ref.~\cite[Eq.3.2]{BD2015-1}, the jump matrix $\V(x,t,z)$ is given by
\be \label{E:VinRef}
\V(x,t,z)=\e^{\Theta}
\bpm
1-\frac{|\rho_2|^2}{\gamma}+\rho_1^*\left[\hat{\rho}_1^*+ \frac{\ii q_0}{z \gamma} \rho_2^* \hat{\rho}_2 \right] & -\frac{\rho_2^*}{\gamma}-\frac{q_0^2}{z^2 \gamma^2} \rho_2^* |\hat{\rho}_2|^2+\frac{\ii q_0}{z \gamma} \hat{\rho}_1^* \hat{\rho}_2^*&
-\frac{\ii q_0}{z \gamma} \rho_2^* \hat{\rho}_2-\hat{\rho}_1^* \\
\rho_2-\frac{\ii q_0}{z}\rho_1^* \hat{\rho}_2& 1+\frac{q_0^2}{z^2 \gamma}|\hat{\rho}_2|^2& \frac{\ii q_0}{z} \hat{\rho}_2\\
-\rho_1^*& -\frac{\ii q_0}{z \gamma} \hat{\rho}_2^* & 1
\epm \e^{-\Theta},
\ee
where $\hat{\rho}_j(z)= \rho_j(\hat{z})$ for $j=1,2$.  We now simplify $\V$ by applying symmetry properties~\eqref{E:ABdc} and~\eqref{E:ABrel}of the scattering matrices $\A(z)$ and $\bB(z)$.

Let  $\V_{ij}$  denote the $(ij)$-entry of the matrix $\V$. We take the two seemingly most complex elements in $\V$, namely  $\V_{11}$  and  $\V_{12}$ , as examples to demonstrate how to simplify $\V$.  By substituting the definitions of $\{\rho_j\}_{j=1}^2$ and $\{\hat{\rho}_j \}_{j=1}^2$ into the expression, we obtain:
\begin{align*}
\V_{11}&=1+\frac{a_{21} b_{12}}{a_{11}b_{11}}-\frac{a_{31}}{a_{11}}\left[\frac{a_{13}}{a_{33}}+ \frac{\ii q_0}{z \gamma(z)}(-\gamma(z)\frac{b_{12}}{b_{11}})(\frac{\ii z a_{23}}{q_0 a_{33}})  \right]=1+\frac{a_{21} b_{12} }{a_{11} b_{11}}-\frac{a_{31}}{a_{11}} \left[\frac{a_{13}}{a_{33}}+ \frac{b_{12} a_{23}}{b_{11} a_{33}}     \right]\\
&=1+\frac{a_{21} b_{12} }{a_{11} b_{11}}-\frac{a_{13} a_{31} b_{11}+ a_{23} a_{31} b_{12}}{a_{11}b_{11} a_{33}}=1+ \frac{a_{21} b_{12} }{a_{11} b_{11}}+\frac{a_{31} b_{13}}{a_{11} b_{11}}\\
&=1-\frac{1}{\gamma(z)}|r_1(z)|^2-|r_2(z)|^2,  \qquad z \in \R.
\end{align*}
Similarly, for  $\V_{12}$, we have
\begin{align*}
\V_{12} \e^{\theta_{21}}&=\frac{ b_{12}}{b_{11}}+\frac{a_{23} b_{12} b_{32} }{a_{33} b_{11} b_{33}}+\frac{a_{13} b_{32}}{a_{33} b_{33}}
=\frac{ b_{12} }{b_{11}}+ \frac{b_{12} b_{32} a_{23} + a_{13} b_{32}b_{11}}{b_{11}b_{33} a_{33}} \\
&=\frac{ b_{12}}{b_{11}}+\frac{b_{12} b_{32}\left(b_{13} b_{21}-b_{11}b_{23} \right)+b_{32}b_{11} \left(b_{12} b_{23}-b_{13}b_{22}\right)}{b_{11} a_{33} b_{33}}\\
&=\frac{ b_{12}}{b_{11}}+\frac{b_{12} b_{32}b_{13} b_{21} -b_{32}b_{11} b_{13}b_{22}}{b_{11} a_{33} b_{33}}=\frac{ b_{12}}{b_{11}}+\frac{b_{32}b_{13} \left(b_{12} b_{21} -b_{11}b_{22} \right) }{b_{11} a_{33} b_{33}}\\
&=\frac{ b_{12}}{b_{11}} - \frac{b_{32} b_{13}}{b_{11}b_{33}}=\frac{1}{\gamma(z)}\left(-r_1(z) +r_2(z)r_3(z) \right)^*, \qquad z \in \R .
\end{align*}
The remaining elements can be simplified in a similar manner, and one ultimately obtains~\eqref{E:Vex}.

\subsection{Proof of Lemma~\ref{L:wyxRH1}} \label{AppBproof}
We first show that the determinant of any solution to RH problem~\ref{RHP:zc} is $\gamma(z)$.  Suppose $\M(x,t,z)$ is a solution to RH problem~\ref{RHP:zc}. Since $\M$ satisfies jump condition~\eqref{E:Jump} and $\det \V =1$, it follows that $\det \M$  has no jump across $\R \setminus{0}$. Furthermore, it is straightforward to verify that each point $\zeta_j$ and $\zeta_j^*$ is a removable singularity of $\det \M$. Therefore, $\det \M$ is analytic for $z \in \C \setminus \{0\}$. Additionally, by considering the asymptotic behavior of $\M$ as $z \to \infty$ and as $z \to 0$, one can conclude that
\be \label{E:db1}
\det \M= 1+\frac{f_1}{z}+\frac{f_2}{z^2},
\ee
where $f_1$ and $f_2$ are  to be determined. Since $\M$ satisfies $z \to \hat{z}$ symmetry (see~\eqref{E:RHP11}), we immediately obtain  $(\det \M)(z)=(\det \M)(\hat{z}) \det \bPi(z)$, which shows that
\be \label{E:db2}
\det \M= -\frac{f_2}{q^2_0}-\frac{f_1}{z}-\frac{q_0^2}{z^2}.
\ee
By comparing~\eqref{E:db1} and~\eqref{E:db2}, we immediately obtain $f_1=0$  and  $f_2=-q_0^2$. Therefore , we have $\det \M =\gamma(z)$.

We now prove the uniqueness of the solution to RH problem~\ref{RHP:zc}. Suppose $\boldsymbol{\mathcal{M}}$ is another solution to RH problem~\ref{RHP:zc}. Since  $\det \M = \gamma(z)$, we know that  $\M^{-1}(x,t,z)$  is well-defined for  $z \in \mathbb{C} \setminus \big(\rZ \cup  \{0,\pm q_0 \} \big)$. Consider the function  $\boldsymbol{\mathcal{H}}(x,t,z)=\boldsymbol{\mathcal{M}}\M^{-1}$. It is easy to verify that $\boldsymbol{\mathcal{H}}$  has no jump on  $\R \setminus \{0,\pm q_0 \}$, so  $\boldsymbol{\mathcal{H}}$  is analytic for  $z \in \mathbb{C} \setminus \big(\rZ \cup  \{0,\pm q_0 \} \big)$. Next, we should examine the behavior of  $\boldsymbol{\mathcal{H}}$  near the origin, the discrete spectrum,  and  the branch points.  It can be easily shown that  $\boldsymbol{\mathcal{H}}$ is $\mathcal{O}(1)$ near the origin. The proof of this is similar to the proof that  $\cE(x,t,z)$  is well-defined at the origin; details can be found in the proof of Lemma~\ref{L:zyyl}. Next, we proceed to study the properties of $\boldsymbol{\mathcal{H}}$ in the vicinity of the branch points $\pm q_0$. We use the second symmetry  in~\eqref{E:RHP11} to rewrite $\boldsymbol{\mathcal{H}}$ as
\be
\boldsymbol{\mathcal{H}}(x,t,z)=-\frac{1}{\gamma(z)}\boldsymbol{\mathcal{M}}(x,t,z)  \mathbf{\Gamma}\M^{\dagger}(x,t,z^*)\bJ .
\ee
Now, combining with the growth conditions ~\eqref{E:gcc}  , as $z \to q_0$ from $\C_+$, we have
$
\boldsymbol{\mathcal{H}}(x,t,z)=\mathcal{O}(1).
$
Moreover, this result remains valid if $z \to q_0$ from $\C_-$. Therefore, we conclude that $q_0$ is a removable singularity of $\boldsymbol{\mathcal{H}}$. By a similar argument, one can show that $-q_0$ is also a removable singularity of $\boldsymbol{\mathcal{H}}$. It now remains to analyze the behavior of $\boldsymbol{\mathcal{H}}$ near the discrete spectrum. In fact, due to symmetry, it suffices to examine the behavior of $\boldsymbol{\mathcal{H}}$ in the vicinity of each $\zeta_j$.  Based on the residue condition~\eqref{E:mlstjzc}, we may assume that $\boldsymbol{\mathcal{M}}$ and $\M(x,t,k)$ have the folowing asymptotic expansion at $\zeta_j$:
\begin{align}
&\M(x,t,z)=\begin{pmatrix}
 a_1 & b_1 & c_1\\
 a_2 & b_2 & c_2\\
 a_3 &b_3  &c_3
\end{pmatrix}+
\frac{1}{z-\zeta_j}\begin{pmatrix}
 C_{\zeta_j} c_1 &0 &0 \\
 C_{\zeta_j} c_2 &0 &0\\
   C_{\zeta_j} c_3 &0 &0
\end{pmatrix}+\mathcal{O}(z-\zeta_j),\label{E:asybM1z}\\
&\boldsymbol{\mathcal{M}} (x,t,z)=\begin{pmatrix}
 A_1 &B_1  &D_1 \\
 A_2 & B_2 &D_2 \\
  A_3& B_3 &D_3
\end{pmatrix}+\frac{1}{z-\zeta_j}\begin{pmatrix}
 C_{\zeta_j} D_1 &0  &0 \\
  C_{\zeta_j} D_2 &0  &0\\
  C_{\zeta_j} D_3&0  &0
\end{pmatrix}+\mathcal{O}(z-\zeta_j), \label{E:asybM1zgx}
\end{align}
where $C_{\zeta_j}=\tau_j \e^{\theta_{31}(x,t,\zeta_j)}$. Since $\det \M(x,t,z)=\gamma(z)$, a straightforward calculation  yields
$$
\M^{-1}(x,t,z)=\frac{1}{\gamma(\zeta_j)}\begin{pmatrix}
 (\mathbf{b} \times \mathbf{c})^{\top}\\
\mathbf{\star}  \\
\mathbf{\star}
\end{pmatrix} +\frac{C_{\zeta_j}}{\gamma(\zeta_j) (z-\zeta_j)}\begin{pmatrix}
\mathbf{0}\\
  \mathbf{0}\\
  (\mathbf{c} \times \mathbf{b})^{\top}
\end{pmatrix}+\mathcal{O}(z-\zeta_j),
$$
where $\mathbf{b}= \bpm b_1&b_2&b_3 \epm ^{\top}$, $\mathbf{c}= \bpm c_1&c_2&c_3 \epm ^{\top}$, $\mathbf{0}= \bpm 0&0&0 \epm $ and $'\mathbf{\star}'$  denotes an unspecified row vector. Combining the above formula with the asymptotic expansion ~\eqref{E:asybM1zgx}, by a simple calculation one can find that $\boldsymbol{\mathcal{H}}(x,t,z)=\boldsymbol{\mathcal{M}} \M^{-1}$ has no negative powers as $z\to \zeta_j$ i.e., $\zeta_j$ is a removable singularity of $\boldsymbol{\mathcal{H}}(x,t,z)$.  Based on the above analysis,  one can conclude that  $\boldsymbol{\mathcal{H}}(x,t,z)$ is holomorphic on  $\C$ . Noting that  $\boldsymbol{\mathcal{H}}(x,t,z) \to \bI$  as  $z \to \infty$, Liouville's theorem implies that $\boldsymbol{\mathcal{H}}(x,t,z)$  is identically  $\bI$, i.e., $\boldsymbol{\mathcal{M}}=\M$.
\section{Unique solvability of the pure-soliton RH problem~\ref{rhp:Msol}} \label{AppB}

In this appendix, we prove Lemma~\ref{L:Msolczwyx}.
 RH problem~\ref{rhp:Msol} is essentially a linear system. The residue conditions~\eqref{E:mlstj} and the asymptotic behaviors given by~\eqref{E:sayM0in} yield
\be \label{E:Bkgh}
\begin{aligned}
\M_{sol}(x,t,z)&=\M_{\infty} + \frac{\ii }{z} \M_0 +\sum_{j=0}^{N-1} \left(\frac{\mathrm{Res}_{z=\zeta_j} \M_{sol}(x,t,z)}{z- \zeta_j} +\frac{\mathrm{Res}_{z=\zeta_j^*} \M_{sol}(x,t,z)}{z- \zeta_j^*}   \right)\\
&=\M_{\infty} + \frac{\ii }{z} \M_0 \\
&+\sum_{j=0}^{N-1} \left( \frac{\displaystyle \lim_{z \to \zeta_j} \M_{sol}(x,t,z)}{z-\zeta_j} \cdot \begin{pmatrix}
0  &0  & 0\\
 0 &0  & 0\\
 \hat{C}_j &0  &0
\end{pmatrix}+
\frac{\displaystyle\lim_{z \to \zeta^*_j} \M_{sol}(x,t,z)}{z-\zeta_j^*} \cdot \begin{pmatrix}
0  &0  & \hat{C}_j^*\\
 0 &0  & 0\\
 0&0  &0
\end{pmatrix} \right),
\end{aligned}
\ee
where $\hat{C}_j(x,t)=\hat{\tau_j} \e^{-2 \ii \theta_1(x,t,\zeta_j)}=\tau_j |\delta_1(\zeta_j)|^2 |\delta(\zeta_j)|^2  \e^{-2 \ii \theta_1(x,t,\zeta_j)}$. Note that here we have utilized the symmetry properties satisfied by  $\M_{sol}$  to derive its residue conditions at  $\zeta_j^*$. For details on this part, the reader may refer to Ref.~\cite[Lemma 3.3]{BD2015-1}. In this appendix, we use $\M_{ij}$ to denote the $(i,j)$-th element of $\M_{sol}$.
Considering the element in the second row and third column of the above expression, we obtain
\be \label{E:M23}
\M_{23}(x,t,z)=\frac{q_{1,+}}{q_0}+\sum_{j=0}^{N-1} \frac{\hat{C}_j^*}{z-\zeta_j^*} \M_{21}(x,t,\zeta_j^*),
\ee
where $q_{+,1}$ denotes the first element of the vector $\q_+$. By symmetry~\eqref{E:rhp5-3}, it is straightforward to verify that $\M_{21}(x,t,\zeta_j^*)=\frac{\ii \zeta_j}{q_0}\M_{23}(x,t,\zeta_j)$. Therefore, Eq.~\eqref{E:M23} can be rewritten as
\be \label{E:M23y}
\M_{23}(x,t,z)=\frac{q_{1,+}}{q_0}+\sum_{j=0}^{N-1} \frac{1}{z-\zeta_j^*} \frac{\ii \zeta_j  \hat{C}_j^*}{q_0}\M_{23}(x,t,\zeta_j).
\ee
In Eq.~\eqref{E:M23y}, if we let $z$ take on the values $\zeta_j$, $j=0,..,N-1$, we obtain the following system of linear equations:
\be \label{E:xxxt}
\left(\bI -  \tilde{\G} \right) \bX=\mathbf{f},
\ee
where
\be
\big(\tilde{\G} \big)_{ij}=\frac{\ii q_0 \hat{C}_j^*}{\zeta_j^*} \frac{1}{\zeta_i-\zeta_j^*}, \qquad
\bX=\left(\M_{23}(x,t,\zeta_0),...,\M_{23}(x,t,\zeta_{N-1})   \right)^{\top}, \qquad
\mathbf{f}=\frac{q_{1,+}}{q_0}\left(1,...,1  \right)^{\top}.
\ee
 From Ref.~\cite[Eq.(3.5a)]{BD2015-1}, we know that  $\frac{\tau_j}{\zeta_j} \in \R$. In fact, when $\frac{\tau_j}{\zeta_j}>0$, the corresponding soliton solution is singular, while for  $\frac{\tau_j}{\zeta_j}<0$, the resulting soliton solution is regular.   We will next prove that if  $\frac{\tau_j}{\zeta_j}<0$  for all  $j = 0, \ldots, N-1$ , then  $\det \left( \bI -  \tilde{\G} \right)>0$  for all $(x,t) \in \R \times \R_+$. Therefore, the reconstructed solitons are all regular.

To this end, we define $\varpi_j=- \frac{\hat{C}_j^*}{\zeta_j^*}$. Since $\e^{\theta_{31}(x,t,\zeta_j)}$  is real-valued and satisfies $\e^{\theta_{31}(x,t,\zeta_j)}>0$, it is straightforward to observe that when $\frac{\tau_j}{\zeta_j}<0$, we have $\varpi_j>0$. Let $y_j=-\ii \zeta_j$, then we have $\Re y_j>0$ and $\frac{\ii }{\zeta_i -\zeta_j^*}=\frac{1}{y_i+y_j}$. Thus, system~\eqref{E:xxxt} can be rewritten as
\be \label{E:cxxt}
\left(\bI+\tilde{\bC} \mathbf{\Upsilon} \tilde{\bC}  \right) \left(\tilde{\bC} \bX  \right) = \tilde{\bC}\mathbf{f},
\ee
where
$$
\big( \mathbf{\Upsilon} \big)_{ij}=\frac{1}{y_i+y_j^*}, \qquad
\tilde{\bC}=\mathrm{diag}\left(\tilde{c}_0,..., \tilde{c}_{N-1}  \right), \quad \tilde{c}_j=\sqrt{q_0  \varpi_j}>0.
$$
Note that  $\det \left(\bI -  \tilde{\G} \right)=\det \left(\bI+\tilde{\bC} \mathbf{\Upsilon} \tilde{\bC}  \right) $. Therefore, it suffices to prove that $\det \left(\bI+\tilde{\bC} \mathbf{\Upsilon} \tilde{\bC}  \right) >0$. We prove this by showing that $\bI+\tilde{\bC} \mathbf{\Upsilon} \tilde{\bC} $ is a positive definite matrix. Indeed, for any complex vector $\boldsymbol{\omega}=\left(\omega_0,...,\omega_{N-1}  \right)^{\top} \ne (0,...,0)^{\top}$, we have
\begin{align*}
\boldsymbol{\omega}^{\dagger} \left(\bI+\tilde{\bC} \mathbf{\Upsilon} \tilde{\bC}  \right) \boldsymbol{\omega}&=\|\boldsymbol{\omega} \|^2+\int_{0}^{+\infty} \sum_{j,k=0}^{N-1}\tilde{c}_j   \tilde{c}_k \e^{-(y_k+y_j^*)s}\omega_k^*\omega_j \mathrm{d}s\\
&=\|\boldsymbol{\omega} \|^2+\int_{0}^{+\infty} \left|\sum_{j=0}^{N-1}\tilde{c}_j \omega_j \e^{-y_j^*s} \right|^2\mathrm{d}s>0.
\end{align*}
Therefore, $\det \left(\bI -  \tilde{\G} \right)>  0 $ for all  $(x,t) \in \R \times \R_+$, and thus linear system~\eqref{E:xxxt} is uniquely solvable for each $(x,t)$. Then, combining this with~\eqref{E:M23y}, we immediately conclude that $\M_{23}(x,t,z)$ can be uniquely determined. Then, through a completely analogous computation, one can solve for $\M_{13}(x,t,z)$ and $\M_{33}(x,t,z)$. Therefore, the third column of $\M_{sol}$ can be uniquely determined for  all $(x,t) \in \R \times \R_+$. Moreover, it is straightforward to observe that the first column of  $\M_{sol}$ can be determined from its third column using symmetry~\eqref{E:rhp5-3}. Finally, from Eq.~\eqref{E:Bkgh}, we immediately deduce that
\be \label{E:Msolderlie}
[\M_{sol}]_2=\left[\M_{\infty}+\frac{\ii}{z} \M_0 \right]_2.
\ee
Now we have fully determined  $\M_{sol}$ for all $(x,t) \in \R \times \R_+$.

\section{The model RH problem} \label{App:ccc}
In this appendix, we present some fundamental results on model RH problem. The jump contour $X=\cup_{j=1}^4 X_j$(see Figure \ref{fig:B}) for the model problems is given by 
\begin{equation}\label{E:Xj}
\begin{aligned}
		&X_1=\{\zeta \in\C:\zeta=r \e^{\frac{\pi i}{4}},0 \le r\le\infty\},\quad &X_2=\{\zeta \in\C:\zeta =r \e^{\frac{3\pi i}{4}},0\le r\le\infty\},\\
		&X_3=\{\zeta \in\C:\zeta=r \e^{\frac{5\pi i}{4}},0\le r\le\infty\},\quad &X_4=\{\zeta \in\C:\zeta=r \e^{\frac{7\pi i}{4}},0\le r\le\infty\}.
	\end{aligned}
\end{equation}
 We assume that $p$ and $q$ belong to a subset $\mathbb{D}$ of the complex plane.
Define the parameter $\nu=\nu(p,q)$ by
$
\nu=-\frac{1}{2 \pi} \log (1-pq)
$, \  $pq<0$.
  Then we define  the  model RH problem as follows.

\begin{RHP}\label{rhpNL}
		The $2 \times 2$ matrix-valued function $\mathbf{M}^{PC}(\zeta;p,q)$ satisfies the following properties:
		\begin{enumerate}
			\item  $\mathbf{M}^{PC}(\cdot\ ;p,q):~\C \setminus X\to\C^{2\times2}$  is analytic for $\zeta \in\C \setminus X$.
			
			\item The function $\mathbf{M}^{PC}(\zeta;p,q)$ is continuous on $X\setminus\{0\}$ and satisfies the jump condition
			$$
			\big(\mathbf{M}^{PC}(\zeta;p,q)\big)_+=\big(\mathbf{M}^{PC}(\zeta;p,q)\big)_- \mathbf{V}^{PC}(\zeta;p,q),\quad \zeta \in X \setminus \{0\},
			$$
			where the jump matrix $\mathbf{V}^{PC}$ is defined by
			$$
			\begin{aligned}
				&\mathbf{V}^{PC}_1 = \left(\begin{array}{ccc}
					1 &  p \zeta^{-2 \ii \nu } \e^{\frac{\ii \zeta^2}{2}}  \\
					0 & 1  
				\end{array}\right),
 & \mathbf{V}^{PC}_2= \left(\begin{array}{ccc}
					1 & 0 \\
					-\frac{q}{1-pq} \zeta^{2 \ii \nu }  \e^{-\frac{\ii \zeta^2}                                   {2}} & 1 
				\end{array}\right),  \\
				&\mathbf{V}^{PC}_3= \left(\begin{array}{ccc}
					1 & \frac{p}{1-pq} \zeta^{-2 \ii \nu } \e^{\frac{\ii \zeta^2}{2}} \\
					0 & 1 
				\end{array}\right),  &
\mathbf{V}^{PC}_4= \left(\begin{array}{ccc}
					1 & 0  \\
					-q \zeta^{2  \ii \nu} \e^{-\frac{\ii \zeta^2}{2}} & 1 
				\end{array}\right).
			\end{aligned}
			$$
			with $\zeta^{2 \ii \nu}=\e^{2\ii  \nu  \log (\zeta)  }$ .
			\item  $\mathbf{M}^{PC}(\zeta;p,q)\to \bI_{2 \times 2}$  as $\zeta\to\infty$.
			\item  $\mathbf{M}^{PC}(\zeta;p,q) \to \mathcal{O}(1)$ as $\zeta \to 0$.
		\end{enumerate}
	\end{RHP}

\begin{figure}
   \centering
   \begin{tikzpicture}[scale=0.7]
\draw [black,very thick,->](-3,-3)--(1.5,1.5);
\draw [black,very thick](1.5,1.5)--(3,3);
\draw [black,very thick,->](-3,-3)--(-1.5,-1.5);
\draw [black,very thick,->](-3,3)--(1.5,-1.5);
\draw [black,very thick](1.5,-1.5)--(3,-3);
\draw [black,very thick,->](-3,3)--(-1.5,1.5);

\node  at (0,-0.5) {$0$};
\node  at (3.1,2.4) {$X_1$};
\node  at (-3.1,2.4) {$X_2$};
\node  at (-3.1,-2.4) {$X_3$};
\node  at (3.1,-2.4) {$X_4$};

 \end{tikzpicture}
   \caption{ The jump contour $X=\cup_{j=1}^4 X_j$  for $\M^{PC}(\zeta;p,q)$. }\label{fig:B}
\end{figure}

RH problem~\ref{rhpNL} can be solved by the standard parabolic-cylinder
model construction, as in~\cite{pz93,Lenells2017}.  We omit the details and only give
the large-$\zeta$ asymptotics of its solution.
\begin{lemma}\label{L:modelasy}
The solutions $\mathbf{M}^{PC}(\zeta;p,q)$ to RH problems~\ref{rhpNL}  exhibits the following asymptotic behavior:
\begin{equation}\label{E:mXasymoxia}
\mathbf{M}^{PC}(\zeta;p,q) = \bI_{2 \times 2} + \frac{\mathbf{M}_{\infty}^{PC}(p,q)}{\zeta} + \mathcal{O}\big(\frac{1}{\zeta^2}\big), \qquad \zeta \to \infty,  
\end{equation}
where the error terms are uniform with respect to $\mathrm{arg} \zeta \in [0, 2\pi]$ and $p, q$ in compact subsets of $\mathbb{D}$. Here, the function $\mathbf{M}_{\infty}^{PC}(p,q)$ is given by
\begin{equation}\label{E:m1Xdefbuzbz}
\mathbf{M}_{\infty}^{PC}(p,q) =\begin{pmatrix}
0 & \ii \beta^{PC}_{12} 	\\
-\ii \beta^{PC}_{21} & 0 \\
 \end{pmatrix}, 
\end{equation}
where
\be \label{E:beta1232}
\beta^{PC}_{12}=\frac{\sqrt{2 \pi} \e^{-\frac{ \pi \nu}{2}} \e^{-\frac{ \pi}{4}  \ii }  }{q \Gamma(\ii \nu)}, \qquad
\beta^{PC}_{21}= \frac{\sqrt{2 \pi} \e^{\frac{- \pi \nu}{2}} \e^{\frac{ \pi}{4}  \ii }  }{p \Gamma(-\ii \nu)}.
\ee
Here, $\Gamma(\cdot)$ denotes the Gamma function. Moreover, for each compact subset $K$ of $\mathbb{D}$, we have
\begin{align*}
\sup_{p,\ q \in K} \sup_{\zeta \in \mathbb{C} \setminus  X} |\mathbf{M}^{PC}(\zeta;p,q)| < \infty.
\end{align*}
\end{lemma}

\noindent{\bf Acknowledgements}
The authors would like to thank the anonymous referees for their careful reading of the manuscript and for their detailed and constructive comments. Their suggestions led to substantial improvements in the presentation of the main results and the organization of the nonlinear steepest descent analysis.
This work is supported by National Natural Science Foundation of China (Grant Nos. 12471234, 12271490, 12571268) and Science Foundation of Henan Academy of Sciences (Grant No. 20252319002).\\ \\
{\bf Data Availability} The data that supports the findings of this study are available within the article.\\ \\
{\bf Declarations}\\ \\
{\bf Conflict of interest}  The authors have no conflict of interest to declare that are relevant to the content of this article.\\


\medskip

\end{document}